\documentclass[11pt]{article}

\usepackage{subcaption}

\usepackage{color}
\usepackage{xspace}
\usepackage[linesnumbered,boxed,ruled,vlined]{algorithm2e} 
\usepackage{tikz}
\usepackage{tkz-graph}
\usetikzlibrary{matrix, arrows,shapes, positioning, automata}
\usetikzlibrary{patterns}
\usepackage[margin = 1in]{geometry}
\usepackage{fullpage}
\usepackage{float}
\usepackage{xcolor}
\usepackage{appendix}
\usepackage{graphicx}
\usepackage{amsfonts}
\usepackage{amssymb}
\usepackage{amsmath}
\usepackage{amsthm}
\usepackage{url}
\usepackage{cite}

\makeatletter
\newtheorem*{rep@theorem}{\rep@title}
\newcommand{\newreptheorem}[2]{%
\newenvironment{rep#1}[1]{%
 \def\rep@title{#2 \textbf{\ref{##1}}}%
 \begin{rep@theorem}}%
 {\end{rep@theorem}}}
\makeatother

\newreptheorem{theorem}{\textbf{Theorem}}
\newreptheorem{lemma}{\textbf{Lemma}}

\DeclareRobustCommand\drawcircle[1]{%
  \begin{tikzpicture}
    \draw[#1] (0pt,0pt) circle [radius = 2pt];
  \end{tikzpicture}%
}

\DeclareRobustCommand\drawrectangle[1]{%
  \begin{tikzpicture}
    \draw[#1] (0pt,0pt) rectangle (3pt,3pt);
  \end{tikzpicture}%
}

\DeclareRobustCommand\drawcover[1]{%
  \begin{tikzpicture}
    \draw[#1] (0pt,0pt) rectangle (6pt,6pt);
  \end{tikzpicture}%
}

\theoremstyle{definition}

\newtheorem{lemma1}{\textbf{Lemma}}
\newtheorem{thm}[lemma1]{\textbf{Theorem}}
\newtheorem{defn}[lemma1]{\textbf{Definition}}

\newtheorem{cor}[lemma1]{\textbf{Corollary}}

\newtheorem{property}[lemma1]{\textbf{Property}}

\newtheorem{claim}[lemma1]{\textbf{Claim}}

\newtheorem{conjecture}[]{\textbf{Conjecture}}
\newtheorem{obs}[lemma1]{\textbf{Observation}}

\newenvironment{proofofprop}[1]{\noindent {\em Proof of Property #1: }\ignorespaces}{}

\newcommand{\topic}[1]{\vspace{0.2cm}\noindent{\bf #1:}}

\newcommand{\e}{\epsilon}
\newcommand{\R}{\mathbb{R}}
\newcommand{\floor}[1]{\lfloor #1 \rfloor}

\newcommand{\level}[1]{\text{\emph{level}-}{#1}}

\newcommand{\pset}{\mathcal{P}}
\newcommand{\ov}{\overrightarrow}

\newcommand{\YY}[1]{\mathsf{YY}_{#1}}

\newcommand{\Yao}{\mathsf{Y}}

\newcommand{\normp}{\mathcal{P}^{\text{n}}_m}
\newcommand{\auxp}{\mathcal{P}^{\text{a}}_m}
\newcommand{\vangle}{\angle}
\newcommand{\ffill}{\mathsf{Refine}}
\newcommand{\mth}{\mathsf{Proj}}

\newcommand{\calP}{\mathcal{P}}
\newcommand{\calS}{\Phi}
\newcommand{\sset}{\Phi}
\newcommand{\calA}{\mathcal{A}}
\newcommand{\tcalA}{\widehat{\mathcal{A}}}
\newcommand{\calB}{\mathcal{B}}

\newcommand{\pair}{\phi}
\newcommand{\vpair}{\varphi}
\newcommand{\gadget}{\mathsf{G}}
\newcommand{\onseg}{partition}
\newcommand{\outseg}{apex}
\newcommand{\stype}{near-empty}
\newcommand{\ttype}{empty}
\newcommand{\ntree}{\mathcal{T}}
\newcommand{\rtree}{\mathcal{T}^{R}}

\newcommand{\projref}{\mathsf{Proj\text{-}Refn}}

\newcommand{\hinge}[1]{\Lambda_{#1}}
\newcommand{\rhinge}[1]{\Lambda^{(+)}_{#1}}
\newcommand{\lhinge}[1]{\Lambda^{(-)}_{#1}}

\definecolor{red}{rgb}{0.89, 0.0, 0.13}
\newcommand{\eat}[1]{}

\title{Odd Yao-Yao Graphs are Not  Spanners}

\author{Yifei Jin\\
Tsinghua University, China  \\
jin-yf13@mails.tsinghua.edu.cn \\
\and
Jian Li\\
Tsinghua University, China \\
lijian83@mail.tsinghua.edu.cn \\
\and
Wei Zhan \\
Princeton University, USA \\
weizhan@cs.princeton.edu \\
}

\date{}

\eat{
\author[1]{Yifei Jin}
\author[2]{Jian Li}
\author[3]{Wei Zhan}
\affil[1]{Tsinghua University, Beijing 100084, China\\
  \texttt{jin-yf13@mails.tsinghua.edu.cn}}
\affil[2]{Tsinghua University, Beijing 100084, China\\
  \texttt{lijian83@mail.tsinghua.edu.cn}}
\affil[3]{Princeton University, USA\\
    \texttt{weizhan@cs.princeton.edu}}

\authorrunning{Yifei Jin, Jian Li, Wei Zhan} 

\Copyright{Yifei Jin, Jian Li, Wei Zhan}


\EventEditors{John Q. Open and Joan R. Acces}
\EventNoEds{2}
\EventLongTitle{42nd Conference on Very Important Topics (CVIT 2016)}
\EventShortTitle{CVIT 2016}
\EventAcronym{CVIT}
\EventYear{2016}
\EventDate{December 24--27, 2016}
\EventLocation{Little Whinging, United Kingdom}
\EventLogo{}
\SeriesVolume{42}
\ArticleNo{23}
}

\begin{document}

\maketitle

\begin{abstract}
  
It is a long standing open problem whether Yao-Yao graphs $\YY{k}$ are all spanners~\cite{li2002sparse}.
Bauer and Damian~\cite{bauer2013infinite}  showed that 
all $\YY{6k}$ for $k \geq 6$  are spanners.
Li and Zhan~\cite{li2016almost}  generalized their result and proved that all even Yao-Yao graphs $\YY{2k}$ are spanners (for $k\geq 42$). 
However, their technique cannot be extended to odd Yao-Yao graphs,
and whether they are spanners are still elusive.
In this paper, we show that, surprisingly,
for any integer $k \geq 1$, there exist 
odd Yao-Yao graph $\YY{2k+1}$ instances, which are not spanners. 

\end{abstract}


\section{Introduction}
\label{sec:intro}
Let $\pset$ be a set of points in the Euclidean plane $\R^2$.
The complete Euclidean graph defined
on set $\pset$ is the edge-weighted graph with vertex set $\pset$ and edges connecting all pairs of
points in $\pset$,
where the weight of each edge is the Euclidean distance between its two end points.
Storing the complete graph requires quadratic space, which is very expensive.
Hence, it is desirable to use a sparse subgraph to approximate the complete graph.
This is a classical and well-studied topic in computational
geometry (see e.g., \cite{ gabriel1969new, toussaint1980relative, aurenhammer1991voronoi, yao1982constructing, li2002sparse}).
In this paper, we study the so called \emph{geometric $t$-spanner},
formally defined as follows (see e.g.,~\cite{sack1999handbook}).

\begin{defn} (Geometric $t$-Spanner)
  A graph $G$ is a {\em geometric $t$-spanner} of the complete Euclidean graph if
  (1) $G$ is a subgraph of the complete Euclidean graph;
  and (2) for any pair of points $p$ and $q$ in $\calP$,
  the shortest path between $p$ and $q$ in $G$ is no longer than $t$ times the  Euclidean distance between $p$ and $q$.
\end{defn}

The factor $t$ is called the \emph{stretch factor} or \emph{dilation factor} of the spanner in the literature. 
 If the maximum degree of $G$ is bounded by a constant $k$,  we say that $G$ is a \emph{bounded-degree spanner}.
The concept of geometric spanners was first proposed by L.P. Chew~\cite{chew1986there}.
See the comprehensive survey by Eppstein~\cite{eppstein1999spanning} for related
topics about geometric spanners.
Geometric spanners
have found numerous applications in wireless ad hoc and sensor networks.
We refer the readers to the books by Li~\cite{li2008wireless} and Narasimhan and Smid~\cite{narasimhan2007geometric} for more details.

 \emph{Yao graphs} are one of the first approximations of complete Euclidean graphs, introduced independently by Flinchbaugh and Jones~\cite{flinchbaugh1981strong} and
 Yao~\cite{yao1982constructing}.

\begin{defn}[Yao Graph $\Yao_{k}$]
	\label{def:yao}
 Let $k$ be a fixed integer. Given a set of points $\pset$ in the Euclidean plane $\R^2$, the Yao
 graph $\Yao_{k}(\calP)$ is defined as follows. Let $C_u(\gamma_1, \gamma_2]$ be the cone with apex $u$, which consists
 of the rays with polar angles in the half-open  interval $(\gamma_1, \gamma_2]$. For each point $u \in \pset$,
 $\Yao_k(\calP)$ contains an edge connecting $u$ to a
 nearest neighbor $v$ in each cone $C_u(j\theta, (j+1)\theta]$, for $\theta = 2\pi/k$ and $j \in [0, k-1]$.
 We generally consider Yao graphs as undirected graphs.  For a \emph{directed Yao graph}, we add directed edge
 $\overrightarrow{uv}$ to the graph instead.
\end{defn}

 Molla~\cite{el2009yao} showed that $\Yao_{2}$ and $\Yao_{3}$ may not be spanners. On the other hand, it has been proven that all $\Yao_{k}$
 for $k \geq 4$ are spanners.  Bose et al.~\cite{bose2012pi} proved that $\Yao_{4}$ is a $663$-spanner. Damian and Nelavalli~\cite{damian2017improved} improved this to 54.6 recently.  Barba et al.~\cite{barba2014new} showed
 that  $\Yao_{5}$  is a  $3.74$-spanner. Damian and Raudonis~\cite{damian2010yao} proved that the 
 $\Yao_6$ graph is a $17.64$ spanner. Li et al.~\cite{li2001power, li2002sparse}  first proved that all $\Yao_{k}, k > 6$ are
 spanners with stretch  factor at most $1/(1- 2\sin(\pi/k))$.
 Later Bose et al.~\cite{bose2004approximating, bose2012pi} also obtained the same
 result independently. Recently,  Barba et al.~\cite{barba2014new} reduced  the stretch factor of
 $\Yao_{6}$ from 17.6 to 5.8  and improved the stretch factors to  $1/(1-2\sin(3\pi/4k))$  for odd
 $k \geq 7$.

 However, a Yao graph may not have bounded degree. This can be a serious  limitation in certain wireless
 network applications since each node has very limited energy and communication capacity, and can
 only communicate with a small number of neighbors.
 To address the issue, Li et al.~\cite{li2002sparse} introduced  \emph{Yao-Yao} graphs (or
 \emph{Sparse-Yao} graphs
 in the literature). A Yao-Yao graph  $\YY{k}(\calP)$ is obtained
 by removing some edges from  $\Yao_{k}(\calP)$ as follows:

 \begin{defn} [Yao-Yao Graph $\YY{k}$]
 (1) Construct the directed Yao graph, as in Definition~\ref{def:yao}.
 (2) For each node $u$ and each cone rooted at $u$ containing two or more incoming edges, retain a shortest incoming edge and discard the other incoming edges in the cone.  We can see that the maximum degree in $\YY{k}(\calP)$ is upper-bounded by $2k$.
 \end{defn}

As opposed to Yao graphs, the spanning property of Yao-Yao graphs
is  not well understood yet.
Li et al.~\cite{li2002sparse} provided some empirical evidence, suggesting that $\YY{k}$ graphs are
$t$-spanners for some  sufficiently large  constant $k$. However, there is no theoretical proof
yet, and it is still an open problem~\cite{li2016almost, li2008wireless, bauer2013infinite, li2002sparse}.
It is also listed as Problem 70 in the Open Problems Project.\footnote{\url{http://cs.smith.edu/~orourke/TOPP/P70.html}}

  \begin{conjecture}[see \cite{bauer2013infinite}]\label{cj}
  	There exists a constant $k_0$ such that for any integer $k>k_0$, any Yao-Yao graph $\YY{k}$ is a geometric spanner.
  \end{conjecture}

  Now, we briefly review the previous results about Yao-Yao graphs.	
  It is known that $\YY{2}$ and $\YY{3}$ may not be spanners since $\Yao_{2}$ and $\Yao_{3}$ may not be
  spanners~\cite{el2009yao}.
  Damian and
  Molla~\cite{damian2009spanner,el2009yao} proved that $\YY{4}, \YY{6}$ may not be spanners.
  Bauer et al.~\cite{barba2014new} proved that $\YY{5}$ may not be spanners.
  On the positive side, Bauer and Damian~\cite{bauer2013infinite} showed that for any integer $k \geq 6$, any Yao-Yao graph $\YY{6k}$ is a spanner
  with the stretch factor at most $11.67$ and the factor becomes  4.75 for $k \geq 8$. Recently, Li and
  Zhan~\cite{li2016almost} proved that for any integer $ k \geq 42$, any even Yao-Yao graph
  $\YY{2k}$  is a spanner with the stretch factor $6.03 + O(k^{-1})$.

  From these positive results, it is quite tempting
  to believe Conjecture~\ref{cj}.
  However, we show in this paper that, surprisingly,
  Conjecture~\ref{cj} is false for odd Yao-Yao graphs.

  \begin{thm}
    \label{thm:odd}
    For any $k \geq 1$, there exists a class of point set instances
    $\{\calP_m\}_{m\in \mathbb{Z}^+}$
    such that the stretch factor
    of  $ \YY{2k+1}(\calP_m)$ cannot be bounded by any constant,
    as $m$ approaches infinity.\footnote{Here, $m$ is a parameter in our recursive construction. We
      will explain it in detail in Section~\ref{sec:normal}.  Roughly speaking, $m$ is the level of recursion and the number of points in $\calP_m$ increases with $m$. }
  \end{thm}

  \subparagraph{Related Work}

      It has been proven that in some special cases, Yao-Yao graphs are spanners~\cite{jia2003local,
      kanj2012certain, wang2003distributed,damian2008simple}. Specifically, it was shown that $\YY{k}$
  graphs are spanners in \emph{civilized graphs}, where the ratio of the maximum
edge length to the minimum edge length is bounded by a constant~\cite{jia2003local,
  kanj2012certain}.

  Besides the Yao and Yao-Yao graph, the \emph{$\Theta$-graph} is another common geometric $t$-spanner.
  The difference between $\Theta$-graphs and Yao graphs is that in a
  $\Theta$-graph, the nearest  neighbor to $u$  in a cone $C$ is a point $v \neq u$ lying in $C$ and minimizing the Euclidean
  distance between $u$ and the orthogonal projection of $v$ onto the bisector of $C$.  It is known that
   except for  $\Theta_{2}$ and $\Theta_{3}$~\cite{el2009yao}, for  $k =
  4$~\cite{barba2013stretch},  $5$~\cite{bose2015theta5}, $6$~\cite{bonichon2010connections}, $\geq
  7$~\cite{ruppert1991approximating,bose2013spanning}, $\Theta_{k}$-graphs are all geometric spanners. We note that,
  unfortunately,  the degrees of $\Theta$-graphs may not be bounded.

      Recently, some variants  of geometric $t$-spanners such as weak $t$-spanners and power
      $t$-spanners have been studied.  In weak $t$-spanners, the path between two points may be
      arbitrarily long, but must remain within a disk of radius $t$-times the Euclidean distance
      between the points. It is known that all Yao-Yao graphs $\YY{k}$  for $k > 6$ are weak
      $t$-spanners~\cite{grunewald2002distributed,schindelhauer2004spanners,schindelhauer2007geometric}. In 
      power $t$-spanners, the Euclidean distance $|\cdot|$ is replaced by $|\cdot|^{\kappa}$ with a
      constant $\kappa \geq 2$. Schindelhauer et al.~\cite{ schindelhauer2004spanners, schindelhauer2007geometric} proved that for $k > 6$, all Yao-Yao graphs $\YY{k}$
      are power $t$-spanners for some constant $t$.  Moreover, it is known that any $t$-spanner is also
      a weak $t_1$-spanner and a power $t_2$-spanner for some $t_1, t_2$ depending only on
      $t$. However, the converse is not true~\cite{schindelhauer2007geometric}.

Our counterexample is inspired by the concept of fractals. 
Fractals have been used to construct  examples for $\beta$-skeleton graphs with unbounded stretch
factors~\cite{eppstein2002beta}. Here a $\beta$-skeleton graph is defined to contain exactly those edges $ab$ such that no point $c$ forms an angle $\angle acb$ greater than $\sin^{-1} 1/\beta$ if $\beta > 1$ or $\pi - \sin^{-1} \beta $ if $\beta < 1$.
Schindelhauer et al.~\cite{schindelhauer2007geometric} used the same example to prove that there
exist graphs which are weak spanners but not $t$-spanners.
However, their examples cannot serve as counterexamples to the conjecture that odd Yao-Yao graphs
are spanners.





\section{Overview of our Counterexample Construction}
\label{sec:counter3}

We first note that both the counterexamples for $\YY{3}$ and $\YY{5}$ are not weak $t$-spanners~\cite{el2009yao,barba2014new}.
However, Yao-Yao graphs $\YY{k}$ for $ k \geq 7$ are all weak $t$-spanners \cite{grunewald2002distributed,schindelhauer2004spanners, schindelhauer2007geometric}.
Hence, to construct the counterexamples for $\YY{k}$ for $k \geq 7$,  the previous ideas for
$\YY{3}$ and $\YY{5}$ cannot be used. We will construct a class of instances $\{ \calP_{m} \}_{m\in
  \mathbb{Z}^+}$ such that all points in $\calP_m$ are placed in a bounded area.
 Meanwhile, there exist shortest paths in $\YY{2k+1}(\calP_m)$ whose lengths approach infinity as $m$ approaches infinity.

Our example contains two types of points, called \emph{normal points} and \emph{auxiliary
  points}. Denote them by $\normp$ and $\auxp$ respectively and $\calP_{m} = \normp \cup \auxp$. The normal points form the basic skeleton, and the auxiliary points are used to break the edges
connecting any two normal points that are far apart.

We are inspired by the concept of fractals to construct the normal points.  A fractal can be
contained in a bounded area, but its length may diverge.
In our counterexample, the shortest path between two specific
normal points is a fractal-like polygonal path.  Here a polygonal path refers to  a curve specified by a
sequence of points and  consists of the line segments connecting the consecutive points.
Suppose the two specific points are $A$ and $B$, $AB$
is horizontal, and $|AB|=1$.
When $m=0$, the polygonal path is just the line segment $AB$.
When $m$ increases by one, we replace each line segment in the current polygonal path by a {\em
  sawteeth-like} path (see Figure~\ref{fig:fractal0}). If the angle between each segment of  the sawteeth-like path and the base segment
(i.e., the one which is replaced) is $\gamma$, the total length of the path increases by a factor of
$\cos^{-1} \gamma$. An important observation here is that the factor is independent of the number of
sawteeth (see Figure~\ref{fig:fractal1}). If we repeat this process directly,
the length of the resulting path would increase to infinity as
$m$ approaches  infinity  since $\cos^{-1} \gamma > 1$ (see Figure~\ref{fig:fractal2}). However, we need to make sure that such a path is indeed
in a Yao-Yao graph and it is indeed  the shortest path from $A$ to $B$.
There are two technical difficulties we need to overcome.

\begin{figure}[t]
\captionsetup[subfigure]{justification=centering}
  \centering
    \begin{subfigure}{0.4\textwidth}
    \centering
    \includegraphics[width = 1\textwidth]{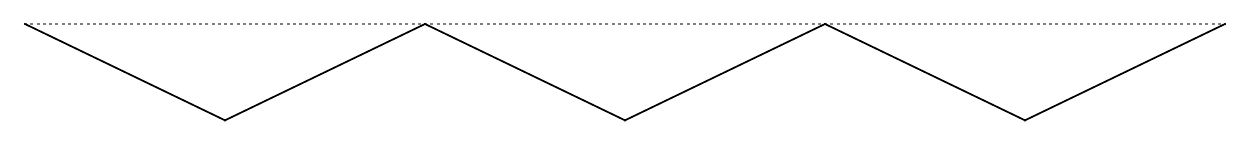}
    \caption{Replace a horizontal segment by a sawteeth-like path.}
    \label{fig:fractal0}
    \end{subfigure}%
    \qquad
    \begin{subfigure}{0.4\textwidth}
    \centering
    \includegraphics[width = 1\textwidth]{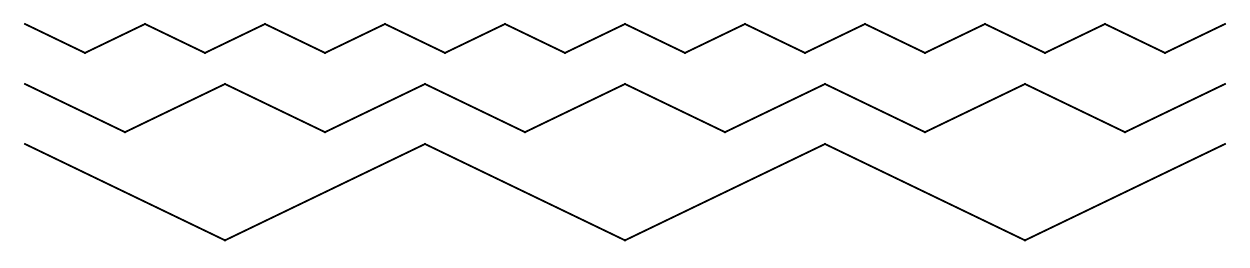}
    \caption{The lengths of
    the sawteeth-like paths are independent of the number of sawteeth.}
    \label{fig:fractal1}
    \end{subfigure}

    \begin{subfigure}{0.55\textwidth}
    \centering
    \includegraphics[width = 1\textwidth]{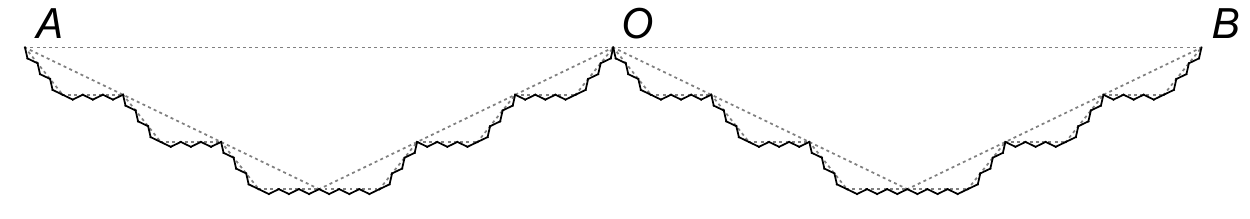}
    \caption{Replace the segments by sawteeth-like path recursively.}
    \label{fig:fractal2}
    \end{subfigure}

    \begin{subfigure}{0.41\textwidth}
    \centering
    \includegraphics[width = 1\textwidth]{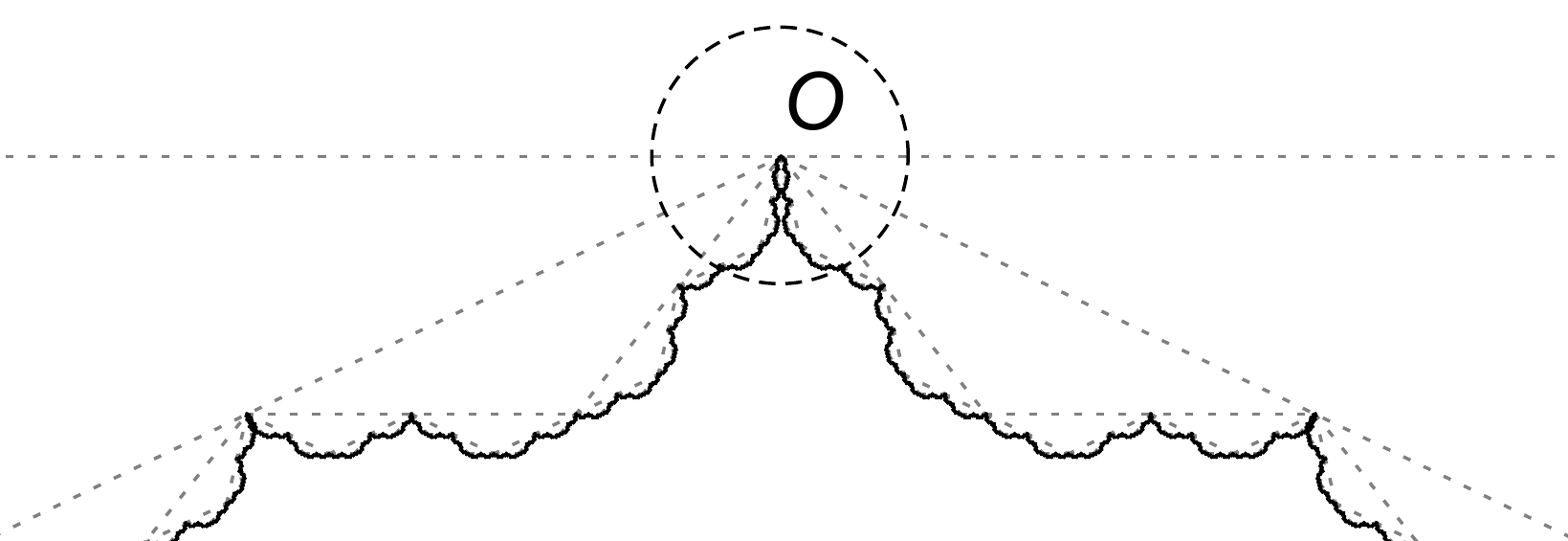}
    \caption{An enlarged view of Figure~\ref{fig:fractal2} around point $O$.}
    \label{fig:fractal3}
    \end{subfigure}

    \begin{subfigure}{0.55\textwidth}
    \centering
    \includegraphics[width = 1\textwidth]{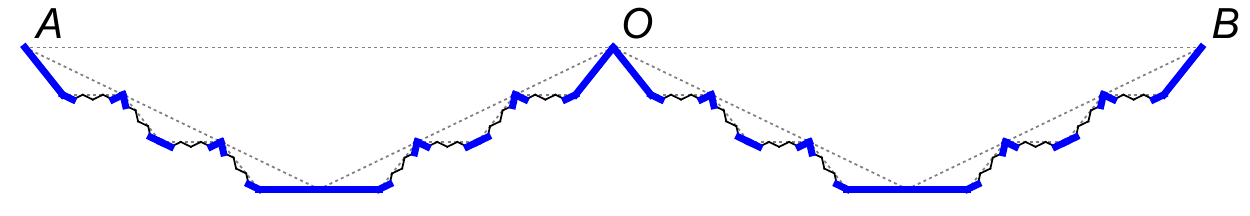}
    \caption{Do not replace the bold segments  further.}
    \label{fig:fractal4}
    \end{subfigure}

    \begin{subfigure}{0.55\textwidth}
    \centering
    \includegraphics[width = 1\textwidth]{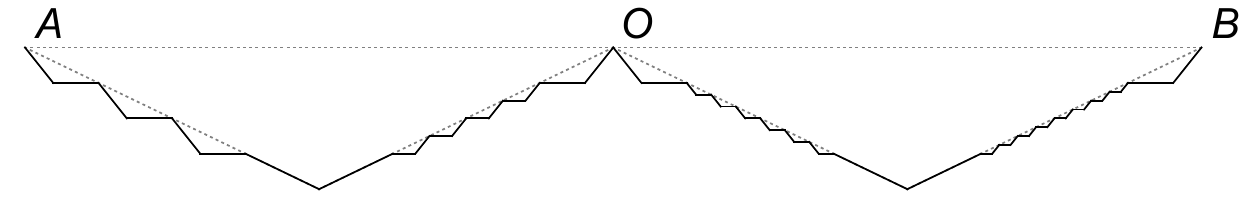}
    \caption{The paths have different numbers of sawteeth and the sizes of sawteeth may not be the same.}
    \label{fig:fractal5}
    \end{subfigure}
  \caption{The overview of the counterexample construction. Figure~\ref{fig:fractal0}-\ref{fig:fractal5} illustrate the fractal and its variants.}
  \label{fig:fractal}
\end{figure}

\begin{enumerate}
	\item
	As $m$ increases, the polygonal path may intersect itself. See Figure~\ref{fig:fractal3}. The polygonal path intersects itself around the point $O$.
	This is relatively easy to handle:	we do not recurse for those segments that may cause self-intersection. See Figure~\ref{fig:fractal4}. We do not replace the bold segment further.
    We need to make sure that the total length of such segments is proportionally small (so that the total length can keep increasing as $m$ increases).	
    	
	\item
	In the Yao-Yao graph defined over the normal points constructed  in the recursion, there may be some edges connecting points that are far apart. Actually, how to break such edges is the main difficulty of the problem.
	We outline the main techniques below.
\end{enumerate}

	First, we {\em do not} replace all current segments using the same sawteeth, like in the usual fractal construction. Actually, for each segment,
      we will choose a polygonal path such that the paths have  different numbers of
      sawteeth and the sizes of the sawteeth in the path may not be the same. See
      Figure~\ref{fig:fractal5}.  Finally,  we construct them in a specific sequential
      order. Actually, we organize the normal points in an $m$-level \emph{recursion tree} $\ntree$
      and generate them in a DFS preorder traversal of the tree. We describe  the details in Section~\ref{sec:normal}.

	Second, we group the normal points into a collection of sets such that each normal point belongs
    to exactly one set. We call such a set a \emph{hinge set}. Refer to Figure~\ref{fig:longshort}
    for an overview.  Then, we specify a \emph{total order} of the hinge sets. Call the edges in the
    Yao-Yao graph $\YY{2k+1}(\normp)$ connecting any two normal points in the same hinge set or two
    adjacent hinge sets (w.r.t. the total order) \emph{hinge connections} and call the other edges
    \emph{long range connections}. We describe the details in Section~\ref{sec:edge}.

     As we will see, all possible long range connections have a relatively simple form.
    Then, we show that
    we can break all long range connections by adding a set $\auxp$ of {\em auxiliary
    points}. Each auxiliary point has a unique \emph{center} which is the normal point closest to it.
    Let the minimum distance between any two normal
    points in $\normp$ be $\Delta$. The distance between an auxiliary point and its center is much less than $\Delta$. Naturally, we can extend  the concepts of hinge set and
    long range connection to include the auxiliary points.  An \emph{extended hinge set} consists of
    the normal points in a hinge set and the auxiliary points centered on these normal points.  We
    will see that the auxiliary points break all long range connections and  introduce no new long
    range connection.  We describe the details in Section~\ref{sec:aux}. 		

	Finally, according to the process above, we can see that the shortest path between the normal
    points $A$ and $B$ in $\YY{2k+1}(\calP_{m})$ for $m\in \mathbb{Z}^+$ should pass through
    all extended hinge sets in order. Thereby, the length of the shortest path between $A$ and
    $B$ diverges as $m$ approaches infinity. We describe the details in
    Section~\ref{sec:length}.



\section{The Positions of Normal Points}
\label{sec:normal}

In this section, we describe the positions of normal points.
Note that, in the section, we only care about the positions of the points. The segments in any figure of this section are used to illustrate the
relative positions of the points. Those segments may not represent the edges in Yao or Yao-Yao
graphs. See Figure~\ref{fig:overview} for an overview of the positions of the normal points.

\begin{figure}[t]
  \centering
  \includegraphics[width = 1\textwidth]{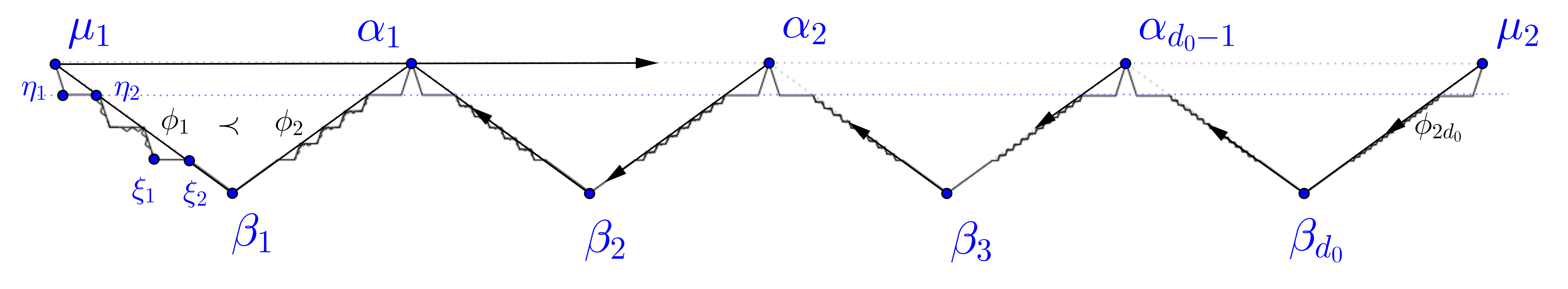}
  \caption{The overview of the positions of normal points. There exists a point at each intersection of these segments.
  $\mu_1\mu_2$ is    horizontal. $\{\alpha_1, \alpha_2, \ldots , $ $\alpha_{d_0 -1}
    \}$ partitions the segment $\mu_1\mu_2$ into $d_0$ equal parts. For each $\beta_i$, $\angle
    \alpha_{i-1}\beta_i\alpha_i = \pi -
    \theta$ and $|\alpha_{i-1}\beta_i| = |\beta_i\alpha_i|$. We call $\{\alpha_1, \alpha_2, \ldots, \alpha_{d_0 -1} \}$ the \onseg\
    set and  $\{\beta_1, \beta_2, \ldots, \beta_{d_0} \}$ the \outseg\ set of pair $(\mu_1,\mu_2)$.
  }
  \label{fig:overview}
\end{figure}

\subsection{Some Basic Concepts}

Let $k \geq 3$ be a fixed positive integer.\footnote{Note that the cases $k=1,2$ have been proved in \cite{el2009yao}.} We consider $\YY{2k+1}$ and let $\theta =2\pi/(2k+1)$.

\begin{defn}[Cone Boundary]
  Consider  any two points $u$ and $ v$.  If the polar angle of $\ov{uv}$ is    $j\theta=j \cdot 2\pi/(2k+1)$
  for   some integer $j \in [0, 2k]$, we call the ray $\ov{uv}$  a cone boundary for point $u $.
\end{defn}

Note that in an odd Yao-Yao graph, if $\ov{uv}$ is a cone boundary, its reverse $\ov{vu}$ is not a cone boundary.
In retrospect, this property is a key difference between odd Yao-Yao graphs and even Yao-Yao graphs,
and our counterexample for odd Yao-Yao graphs will make crucial use of the property.
We make it explicit as follows.

\begin{property}
  \label{prop:cone}
  Consider two points $u $ and $ v$ in $\calP$. If $\ov{uv}$ is a cone boundary in $\YY{2k+1}(\calP)$,
  its reverse $\ov{vu}$ is not a cone boundary.
\end{property}

\begin{defn} [Boundary Pair]
\label{def:pair}
  A boundary pair consists of  two ordered points, denoted by $(w_1, w_2)$, such that $\ov{w_1w_2}$  is a cone boundary of  point $w_1$.
\end{defn}

For convenience, we refer to the word \emph{pair} in the paper as the boundary pair defined in Definition~\ref{def:pair}. 
According to Property~\ref{prop:cone}, if $(w_1, w_2)$ is a pair, its reverse $(w_2, w_1)$ is not a pair. For a pair $\pair = (w_1, w_2)$, we call $w_1$ the first point in $\pair$ and $w_2$ the second point in $\pair$.  Moreover, if a pair $\pair$ is $(u , \cdot)$ or $(\cdot, u)$, we say that the point $u$ belongs to $\pair$ (i.e., $u \in \pair$).

\eat{
  \begin{corollary}
	In odd Yao-Yao graph, if two points $(w_1, w_2)$ is a pair, its inverse $(w_2, w_1)$ is not a
	pair.
  \end{corollary}
}

\topic{Gadget}
Now, we introduce the concept of a gadget generated by a pair
$\pair=(w_1,w_2)$.
Such a gadget is a collection of points which is a superset of $\pair$
(see Figure~\ref{fig:counter3}).
If the recursive level $m$ increases by 1, we
use a gadget generated by pair $\pair$ to replace  $\pair$.

One gadget $\gadget_{\pair}$
consists of three groups of points. We explain them one by one. See
Figure~\ref{fig:counter3} for an example.

\begin{figure}[t]
  \centering
  \includegraphics[width = 1\textwidth]{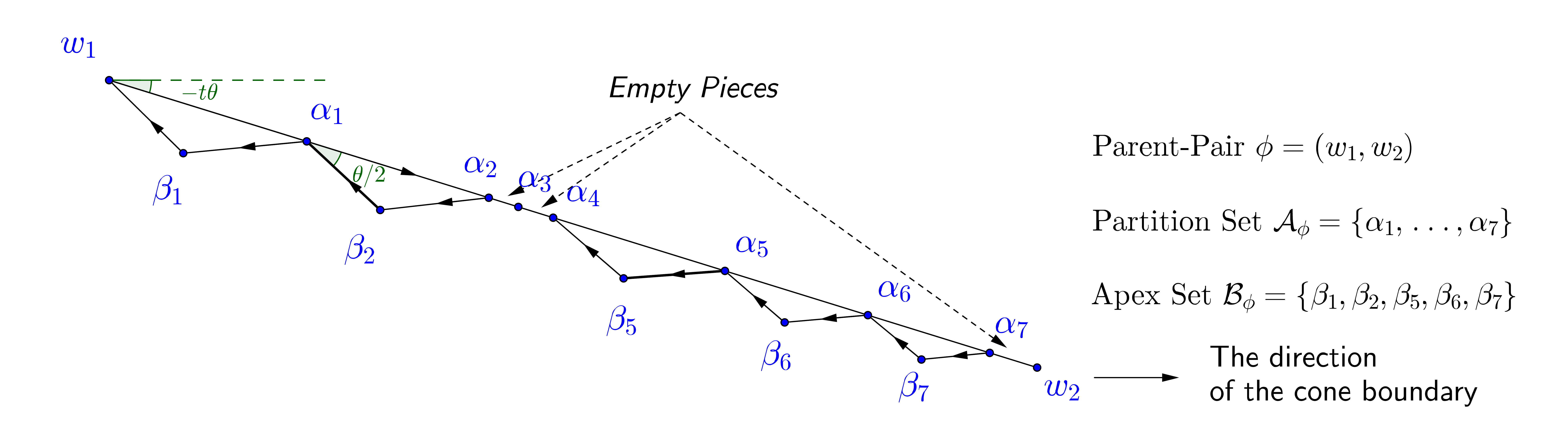}
  \caption{An example of one gadget. $\pair = (w_1,w_2)$ is the parent-pair in the gadget.  $\calA_{\pair} = \{\alpha_1, \alpha_2, $ $\alpha_3,   \ldots,   \alpha_7 \}$ is the \onseg\ set and  $\calB_{\pair} = \{\beta_1, \beta_2, \beta_5,
    \beta_6, \beta_7\}$ is the \outseg\ set.  There  are eight pieces, in which  $w_1\alpha_1, \alpha_1\alpha_2, \alpha_4\alpha_5, \alpha_5\alpha_6, \alpha_6\alpha_7$ are non-empty pieces and $\alpha_2\alpha_3, \alpha_3\alpha_4, \alpha_7 w_2$ are empty pieces. }
  \label{fig:counter3}
\end{figure}

\begin{enumerate}
\item
  The first group is the pair $\pair=(w_1, w_2) $.
  We call the pair the \emph{parent-pair} of the gadget $\gadget_{\pair}$.
\item
  The second group is a set $\calA_{\pair}$ of points on the segment of $(w_1, w_2)$. We call the
  set $\calA_{\pair}$  a \emph{partition set} and call the  points of $\calA_{\pair}$ the
  \emph{\onseg\ points} of  $\pair$.   For example, in  Figure~\ref{fig:counter3}, $\{\alpha_1, \alpha_2, \ldots, \alpha_7\}$
  (here, $|w_1 \alpha_i | < |w_1 \alpha_j|$ if $i < j$) is
  a partition set of $(w_1,  w_ 2)$.
 The set $\calA_{\pair}$ divides the segment into $|\calA_{\pair}|+1$ parts,
  each we call a \emph{piece} of the segment.  There are two types of pieces. One is called  an \emph{empty piece} and the other a
  \emph{non-empty piece}. Whether a piece is empty or not is determined in the process of the construction, which we will
  explain in Section~\ref{subsec:construc}.
\item
  For each non-\ttype\ piece, $\alpha_{i-1}\alpha_{i}$, we
  add a point $\beta_i$
  such that $\angle \alpha_{i-1}\beta_i\alpha_i =
  \pi-\theta$ and  $|\alpha_{i-1}\beta_i| = |\beta_i\alpha_{i}|$.\footnote{Note that the subscript $i$
    of $\beta_i$ is consistent with the subscript of the piece $\alpha_{i-1}\alpha_{i}$. Hence, the
    subscripts may not be consecutive among all $\beta_i$s.}
  All $\beta_i$s are on the same
  side of $w_1w_2$. We call such a point $\beta_i$ an \emph{\outseg\ point}  of  $(w_1,
  w_2)$.  Let $\calB_{\pair}$ be the set of apex points generated by $\pair$, which is called
  the \emph{apex set} of pair $\pair$.
  $\calB_{\pair}$ is the third group of points. For any empty piece,  we do not
  add the corresponding apex point.   In Figure~\ref{fig:counter3},
  $\{\beta_1, \beta_2,\beta_5, \beta_6, \beta_7\}$ is an apex set of $(w_1, w_2)$.
\end{enumerate}

We summarize the above construction in the following definition.

\begin{defn}[Gadget]
  \label{defn:gadget}
  A gadget $\gadget_{\pair}$ generated by a pair $\pair$ is a set of points which  consists of  the pair $\pair$, a partition set $\calA_{\pair}$ and an apex set
  $\calB_{\pair}$ of $\pair$.  We denote the gadget by $\gadget_{\pair}[\calA_{\pair}, \calB_{\pair}]$.
\end{defn}

\eat{
For convenience, we say that the \emph{segment of a pair} is the segment connecting the two points of
the pair and the \emph{length of a pair} is the length of that segment.
}

Consider a gadget $\gadget_{\pair}[\calA_{\pair}, \calB_{\pair}]$, where $\pair = (w_1, w_2)$. For
any non-\ttype\ piece $\alpha_{i-1}\alpha_i$ and
the corresponding \outseg\ point $\beta_i$, the rays $\ov{\beta_{i}\alpha_{i-1}}$ and $\ov{\alpha_i\beta_i}$
(note the order of the points) are cone boundaries.\footnote{ Suppose the polar angle of $w_1w_2$ is
  $-t\theta$. Note that $(2k+1) \theta = 2\pi$. Then, we can obtain that the polar angle of
  $\ov{\alpha_i\beta_i}$ is $(k-t+1)\theta$ and the polar angle of $\ov{\beta_{i}\alpha_{i-1}}$ is
  $(k-t)\theta$. Hence, $\ov{\alpha_i\beta_i}$ and $\ov{\beta_{i}\alpha_{i-1}}$ are cone boundaries.}  Thus, each point $\beta_i \in \calB_{\pair}$ induces two pairs   $(\beta_i,
\alpha_{i-1} )$ and $ (\alpha_i, \beta_i)$.
We call all pairs $(\beta_i, \alpha_{i-1} )$ and $ (\alpha_i, \beta_i)$ induced by points in $\calB_{\pair}$ the
\emph{child-pairs} of $(w_1, w_2)$, and we say that they are
{\em siblings} of each other.
Now, we define the order of the child-pairs of pair $(w_1,w_2)$,
based on their distances to $w_1$.  Here,  the distance
from a point $w$ to a pair $\pair$ is the shortest distance from $w$ to any point of $\pair$.

\begin{defn}[The Order of the Child-pairs]
  \label{defn:order}
  Consider a gadget $\gadget_{(w_1, w_2)}$.
  Suppose $\sset$ is the set of the child-pairs of $(w_1,
  w_2)$.
  Consider two pairs $\pair, \vpair$ in $\sset$.
  Define the order $\pair \prec  \vpair$, if $\pair$ is closer to $w_1$ than  $\vpair$.
\end{defn}
For example, in Figure~\ref{fig:counter3}, $  (\beta_2,\alpha_1) \prec (\alpha_5, \beta_5)$. We emphasize that the order of the child-pairs depends on the direction of their parent-pair.

\subsection{The Recursion Tree}
\label{subsec:tree}

\begin{figure}[t]
  \centering
  \begin{minipage}[center]{1.0\linewidth}
    \centering
    \includegraphics[width = 0.8\textwidth]{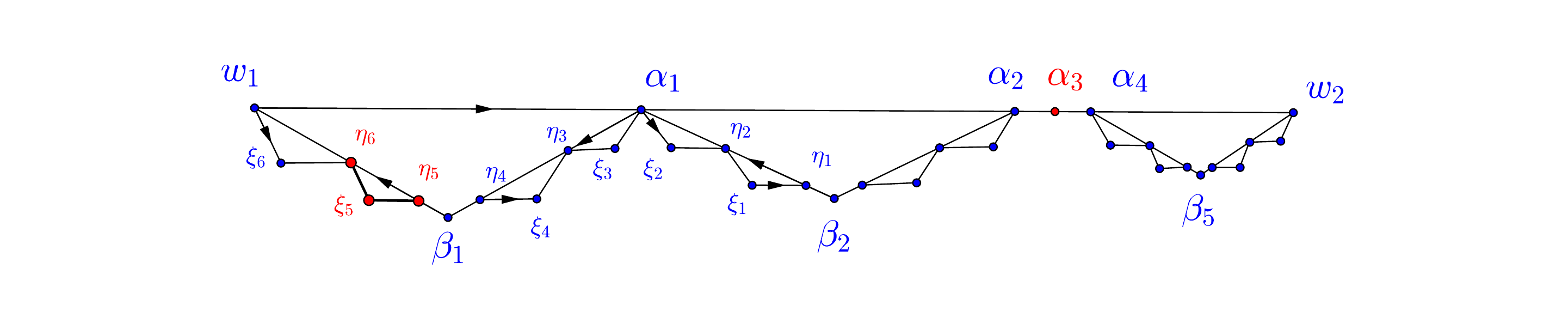}
    \caption{An example of the gadgets which are generated in a recursive manner. $\alpha_3$ is an isolated partition point.
      The arrow of a segment indicates the order of two points in the
      pair. For example, the arrow from $w_1$ to $w_2$ indicates
      that $(w_1, w_2)$ is a pair.}
    \label{fig:recur}
  \end{minipage}

  \begin{minipage}[center]{1.0\linewidth}
    \centering
    \includegraphics[width = 0.8\textwidth]{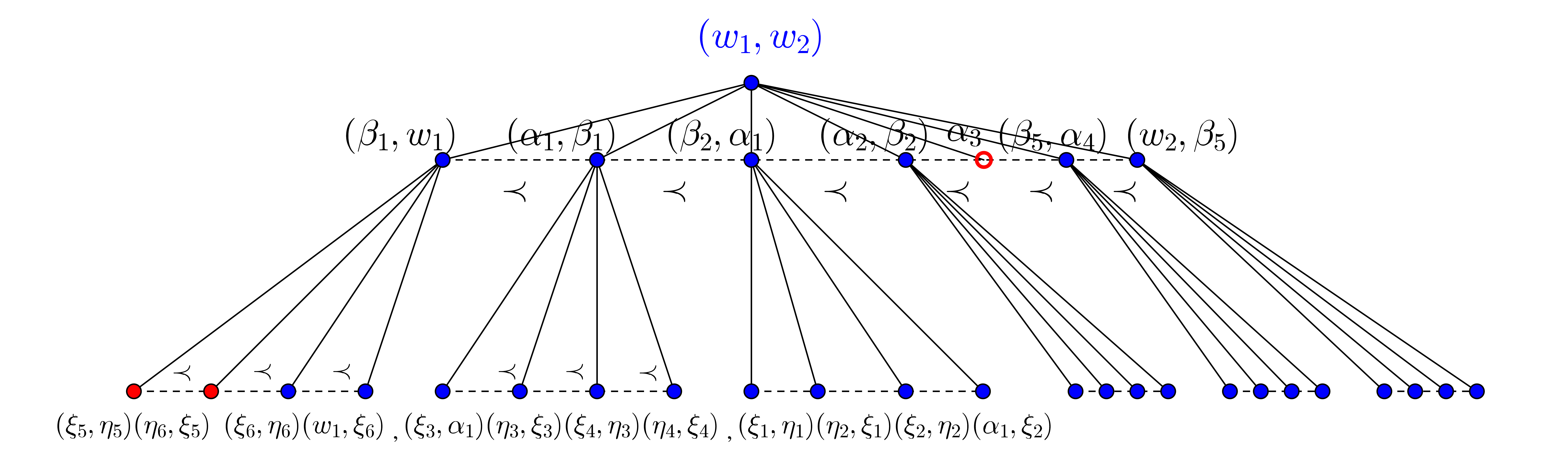}
    \caption{
    The recursion tree of our construction. Each node of the tree represents a pair (e.g.,
      $(\beta_1, \omega_1)$) or a
      point (e.g., $\alpha_3$) in Figure~\ref{fig:recur}.
      Pair $(w_1,w_2)$ is the root at $\level{0}$. Any
      pair at  $\level{(i+1)}$ is generated from a pair at $\level{i}$.
 }
    \label{fig:tree}
  \end{minipage}
\end{figure}

In this subsection, we construct  an $m$-level tree. When the recursion level increases by 1, we need to replace each current pair by  a gadget generated by the pair. The recursion can be naturally represented as a tree $\ntree$.
Each node of the tree represents either a pair or a point.
To avoid confusion, we use \emph{point} to express
a point in $\mathbb{R}^2$ and \emph{node} to express a vertex in the tree.
The pair $(\mu_1,\mu_2)$ is the root of the tree
($\level{0}$).
The child-pairs of $(\mu_1,\mu_2)$ are the child-nodes of the root
(they are at $\level{1}$).
Recursively, each child-pair of a pair $\pair$  is a child-node of the node $\pair$ in $\ntree$.
Besides, there are some {\em  partition points} of the
\ttype\ pieces (e.g., the point $\alpha_3$ in Figure~\ref{fig:recur})
which may not belong to any pair. We call it an \emph{isolated} point. Let an isolated  point be a leaf in  $\ntree$
and the parent of such a point be its parent pair. For example, the parent of $\alpha_3$
is the pair $(\omega_1, \omega_2)$ in Figure~\ref{fig:recur}.
We provide the recursion tree in Figure~\ref{fig:tree}, which corresponds to the points in Figure~\ref{fig:recur}.

\eat{
The pairs with the same parent are siblings. According to
Definition~\ref{defn:order}, given two sibling pairs $\pair$ and $\vpair$,
we say $\pair$ is \emph{to the left} of
$\vpair$ in tree $\ntree$ if $\pair \prec \vpair$. We note that here
``\emph{left}"does not mean $\pair$ is on the left hand side of
$\vpair$ geometrically. }

The nodes with the same parent are siblings. According to Definition~\ref{defn:order},
 we define a total order ``$\prec$''of them.  In tree $\ntree$, if ``$\vpair \prec \pair$'', we place $\vpair$ to the left of $\pair$, (e.g., $(\xi_5, \eta_5)$ is at the left of $(\eta_6, \xi_5)$ in Figure~\ref{fig:tree}).
However, ``$\vpair \prec \pair$'' does not mean that $\vpair$ is on the left hand side of $\pair$ geometrically.
For example, in Figure~\ref{fig:recur}, pair $(\xi_5,\eta_5) \prec (\eta_6,\xi_5 )$ in
the tree $\ntree$, but in the Euclidean plane, point $\eta_5$ is on the right side of $\eta_6$.

For a pair $\pair$ (corresponding to a node in $\ntree$),
we use $\ntree_{\pair}$ to denote the subtree rooted at $\pair$ (including $\pair$), or all the points involved in the subtree.

Our counterexample $\calP_m$ corresponds to a recursion tree with $m$ levels.
We have not yet specified how to
choose the partition set for each gadget and decide which pieces are empty for each pair.
We will do it in the next subsection.
We note that we {\em do not}
construct the tree level by level, but rather according to the  DFS preorder.

\subsection{The Construction}
\label{subsec:construc}

Now, we describe the process of generating the $m$-level recursion tree $\ntree$. See Figure~\ref{fig:steps} for an example.
We call a pair \emph{a leaf-pair} if it is a leaf node in the tree and \emph{an internal-pair} otherwise.
W.l.o.g, we assume that the root of $\ntree$ is  $(\mu_1, \mu_2)$ and $\mu_1\mu_2$ is horizontal.
The tree is generated according to the DFS preorder, starting  from the root.
When we are visiting an internal-pair, we generate its gadget.
Note that generating its gadget is equivalent to
generating its children in $\ntree$.
We, however, do not
visit those children immediately after their generation. They will
be visited later according to the DFS preorder.
Whether a pair is a leaf or not
is determined as the gadget being created. Note that not all leaf-pairs are at $\level{m}$.   

The process generating the gadget for an internal-pair includes two steps, which are called \emph{projection} and \emph{refinement}. We will explain the detail soon.  We denote the procedure
to construct the recursion tree $\ntree$ by $\mathsf{GenGadget}(\pair)$ and the pseudocode can be found in  Algorithm~\ref{alg:gadget}.
We call the points generated by Algorithm~\ref{alg:gadget} \emph{normal points} and denote  them by $\normp$ where $m$ is the level of the tree and n represents the
word ``normal''.

\topic{Root gadget}
 Let $d_0$ be a large positive constant integer.\footnote{$d_0$ depends on $k$, but on the number of points. We will determine the exact value of $d_0$ in
  Section~\ref{sec:length}.} Consider a pair  $ \pair = (\mu_1, \mu_2)$.   Let $\calA_{\pair}$ be its
partition set which contains  points
$$
\alpha_{i}  =  \mu_1 \cdot \frac{d_0-i}{d_0}+ \mu_2 \cdot \frac{i}{d_0}, i \in [1, d_0-1].
$$
For convenience,  let $\alpha_0 = \mu_1, \alpha_{d_0} = \mu_2$.   The points in $\calA_{\pair}$ partition the segment $\mu_1\mu_2$
into $d_0$ pieces with equal length $|\mu_1\mu_2|/ d_0$. All pieces in the root gadget are non-empty.
For each piece $\alpha_{i-1}\alpha_i$,
we add an \outseg\ point $\beta_i$  below $\mu_1\mu_2$. Let $\calB_{\pair} = \{\beta_i\}_{i\in[1,d_{0}]}$ be the apex set.  See
Figure~\ref{fig:step1} for an example.

\begin{figure}[t]
  \centering
  \includegraphics[width = 1\textwidth]{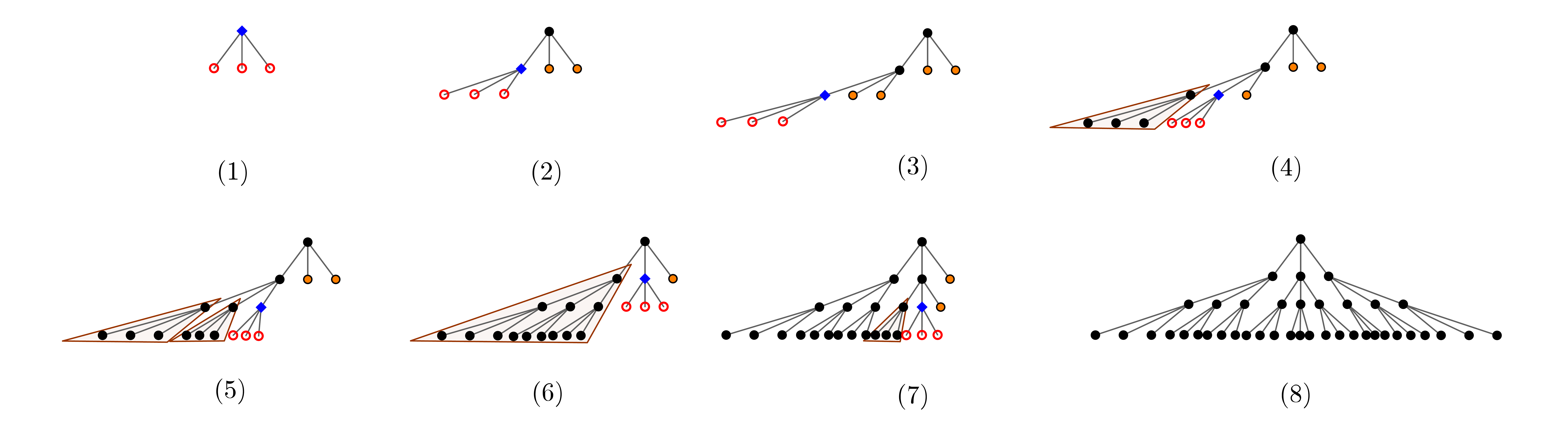}
  \caption{ The process of generating a tree according to the DFS preorder.  In each subfigure,     \drawrectangle{blue, thick, rotate = 45} represents a
    node we are visiting.  The nodes  generated in the step are denoted by  \drawcircle{red, thick}.
    \drawcircle{black, fill=black} represents a node  which has already been
    visited. \drawcircle{black,  fill=orange} represents a node which has been created but not
    visited yet. The nodes covered by light brown triangles are related to the projection process. }
  \label{fig:steps}
\end{figure}

\begin{algorithm}[t]
	\caption{$\mathsf{GenGadget}(\pair)$: Generate the Normal Points in $\ntree_{\pair}$}
	\label{alg:gadget}
	\uIf {$\pair$ is a leaf-pair }{
		Return \;
	}
	\Else{
		$\gadget_{\pair}\leftarrow \projref(\pair)$\;
	}
	\ForEach{child-pair $\vpair$ of $\pair$}{
		$\mathsf{GenGadget}(\vpair)$ \;
	}
\end{algorithm}

\eat{
 Next, we discuss the  process \emph{projection} and \emph{refinement}. For convenience, we say that
 the \emph{segment   of a pair} is the segment connecting the two points of the pair and the
 \emph{length of a pair} is the length of that segment.  In the process, we always maintain the
 following property.

\begin{property}
\label{prop:p2}
Consider an internal-pair $\pair$. Suppose  $\vpair$ is a sibling of $\pair$. The length of pair $\vpair$ is at least half of the length of pair $\pair$.
\end{property}
}

\begin{figure}[t]
  \centering
  \includegraphics[width = 1\textwidth]{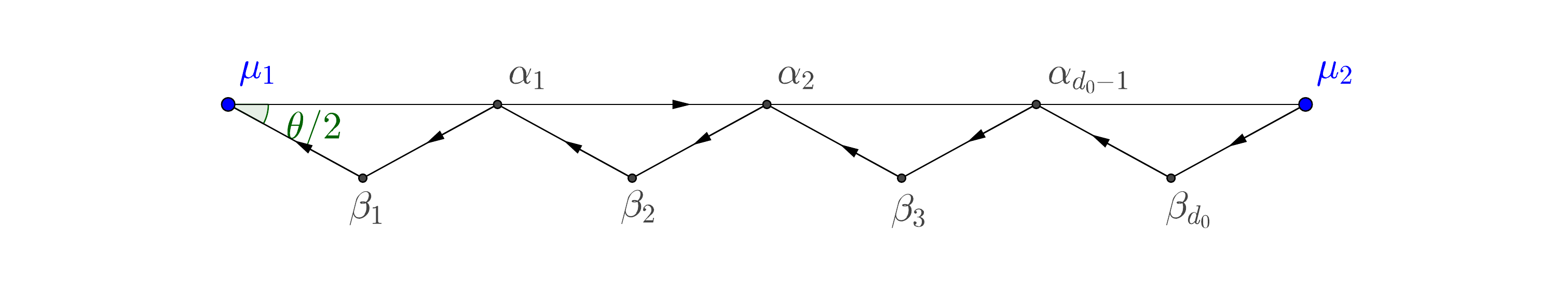}
  \caption{The root gadget $\gadget_{\pair}(\calA_{\pair}, \calB_{\pair})$ where $\pair = (\mu_1,
    \mu_2)$. $\mu_1\mu_2$ is horizontal. $\calA_{\pair}$ is  the equidistant partition. Each piece
    is non-\ttype.}
  \label{fig:step1}
\end{figure}

\topic{Projection and Refinement} Let $\ntree_{\pair}$ be
  the set of points in the subtree rooted at $\pair$.
The projection and refinement process for $\pair$ is slightly more
complicated, as it depends on the subtrees $\ntree_{\vpair}$
rooted at the siblings $\vpair$ of $\pair $ such that $\vpair \prec \pair$.  Recall that when we visit a pair
$\pair$ (and generate $\gadget_{\pair}$) in the tree $\ntree$ according to the  DFS
preorder, we have already visited all pairs $\vpair \prec \pair$.

The projection and refinement generate the partition points of pair $\pair$.
The purpose of the projection is to restrict all possible \emph{long range connections} to a relatively simple form.
See Section~\ref{sec:edge} for the details.  The purpose of the refinement is to make the sibling pairs have relatively the same length, 
hence, make it possible to repeat the projection process recursively. 
Formally speaking, the refinement maintains the following property over the construction. 
\eat{
In Section~\ref{sec:edge}, we will see that the projection restricts all possible \emph{long range connections} to a relatively simple form.
Meanwhile, the refinement maintains Property~\ref{prop:p2} over the construction.
}

\begin{property}
\label{prop:p2}
We call  the segment connecting the two points of the pair the \emph{segment of the pair} and call the length of that segment  the \emph{length of the pair}.  
Consider an internal-pair $\pair$. Suppose  $\vpair$ is a sibling of $\pair$. The length of pair $\vpair$ is at least half of the length of pair $\pair$.
\end{property}

\topic{-- Projection}
Consider a pair $(\beta, \alpha)$ with the set $\Phi$ being its child-pairs. We decide whether a pair in $\Phi$ is a leaf-pair or an internal-pair after introducing the process projection and refinement. 
We provide that the  property of the order here and prove it in the end of the section. 

\begin{property}
\label{prop:leaf-internal}
Consider a pair $(\beta, \alpha)$ with the set $\Phi$ of its child-pairs.
For $\pair_1, \pair_2, \pair_3 \in \Phi$, if $\pair_1 $ $ \prec $ $\pair_2$ $ \prec $$ \pair_3$ and $\pair_1$ and $\pair_3$ are two internal-pairs, then $\pair_2$ is an
internal-pair.
\end{property}

\eat{
Consider a pair $(\beta, \alpha)$ with the set $\Phi$ of its child-pairs. We decide whether a pair in $\Phi$ is a leaf-pair or an internal-pair after introducing the process projection and refinement. Before that, we claim that among the pairs in $\Phi$, there is no leaf-pair between two internal-pairs based on our construction. It means that for $\pair_1, \pair_2, \pair_3 \in \Phi$, if
$\pair_1 $ $ \prec $ $\pair_2$ $ \prec $$ \pair_3$ and $\pair_1$ and $\pair_3$ are two internal-pairs, then $\pair_2$ is an
internal-pair. Besides, we define the first internal-pair in the direction $\ov{\beta\alpha}$ as the
\emph{first internal-pair} of $\Phi$.
}

Next, we  describe the projection operation for a pair $\pair \in \Phi$.  W.l.o.g., suppose $\pair$ is
an internal-pair since only internal-pairs have children.   We define the first internal-pair in direction $\ov{\beta\alpha}$ as the \emph{first internal-pair} of $\Phi$.
Depending on whether $\pair$ is the first internal-pair of $\Phi$, there are two cases.

\begin{figure}[t]
  \centering
  \includegraphics[width = 0.45\textwidth]{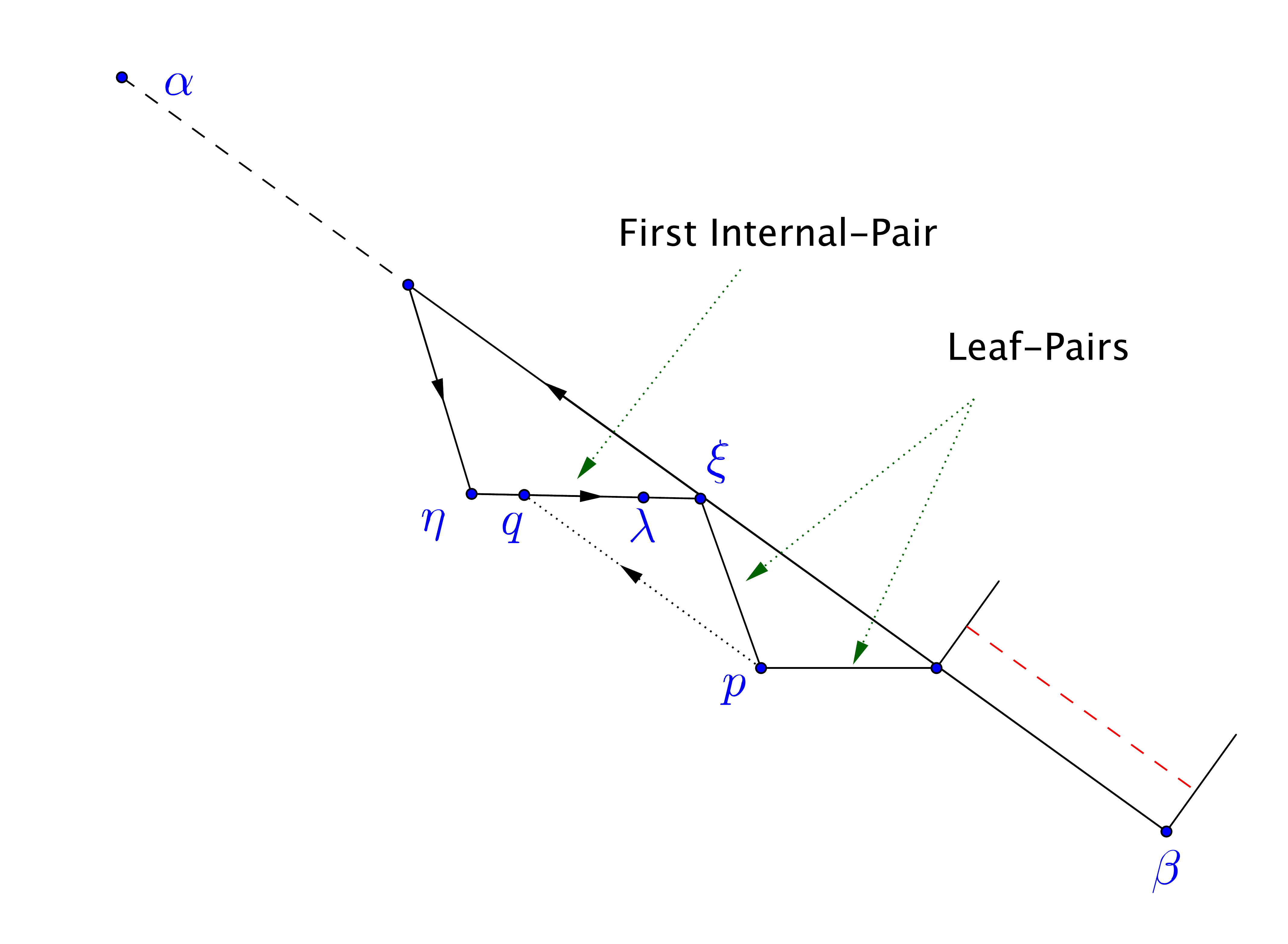}
  \caption{An example of the projection for the first internal-pair $\pair = (\eta, \xi)$. First, we add a point $\lambda$
  such that $|\lambda \xi| = \delta/d_{0}$ where $\delta$ is the length $|\eta \xi|$. Second,  for each leaf-pair $\vpair \prec \pair$, project
  its apex point $p$  to the segment of   $\pair$ along the direction $\protect\overrightarrow{\beta \alpha}$, i.e., add the point $q$ in the figure. }
  \label{fig:secondlast}
\end{figure}

\begin{figure}[t]
  \centering
  \includegraphics[width = 1\textwidth]{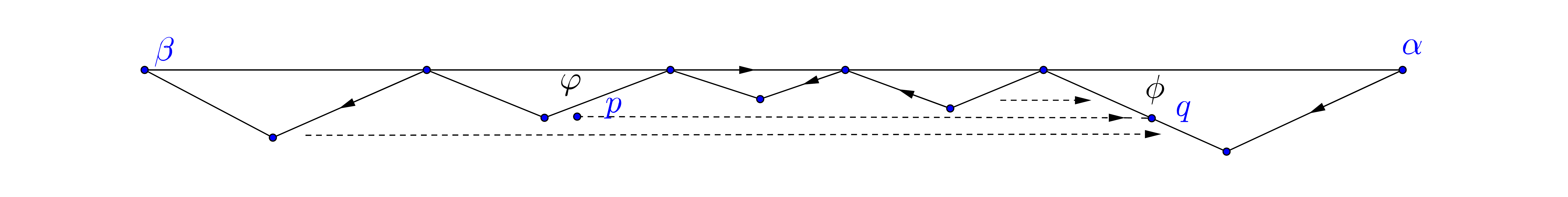}
  \caption{The projection for pair $\pair$. Here, $p$ is a point in subtree $\ntree_{\vpair}$.  $q$
    is a projection point of $p$, i.e., the point on segment of pair $\pair$ such that
    $pq$ is parallel to $\beta\alpha$.  The set $\mth[\bigcup_{\vpair \prec \pair,\vpair \in \Phi} \ntree_{\vpair}]$
    consists of all projection points of  $\bigcup_{\vpair \prec \pair,\vpair \in \Phi} \ntree_{\vpair}$  on
    segment of $\pair$. }
  \label{fig:matching}
\end{figure}

\begin{itemize}
\item Pair $\pair$ is the first internal-pair of $\Phi$: In Figure~\ref{fig:secondlast},  suppose pair $\pair = (\eta,
  \xi)$ is the first internal child-pair of $(\beta, \alpha)$ and
 the length of $\pair$ is $\delta$. Point $\xi$ is the partition point in $\pair$.
First, we add a point $\lambda$ on the segment of $\pair$ such  that $ |\xi\lambda| = \delta/ d_{0}$.
Second, for each leaf-pair $\vpair \prec \pair$, project the  apex point in $\vpair$ to the segment
of  $\pair$ along the direction $\overrightarrow{\beta\alpha}$, \footnote{If the projected point falls outside the segment of $\pair$, we do not need to add a normal point.} e.g., project $p$ to $q$ in Figure~\ref{fig:secondlast}.  Note that the length of leaf-pair
$\vpair$ is at   least $\delta/2$ according to Property~\ref{prop:p2}. Thus, there is no point
between $\lambda$ and $\xi$ as long as $d_0 > 2$. Formally, we denote the operation by
\begin{align}
\label{eqn:proj01}
 \tcalA_{\pair}\leftarrow\mth\left[ \bigcup\nolimits_{\vpair \prec \pair,\vpair \in \Phi} \ntree_{\vpair} \right] \cup \lambda.
\end{align}

\item Pair $\pair$ is not the first  internal-pair: According to the DFS preorder,  we have already
  constructed the subtrees   rooted at  $\vpair \prec \pair$. We project all points $p \in \bigcup_{\vpair
    \prec \pair, \vpair \in \Phi} \ntree_{\vpair}, $ to the segment of $\pair$ along the direction $\overrightarrow{\beta\alpha}$. Let the
  partition set $\tcalA_{\pair}$ of $\pair$ be the set of the projected points falling inside the segment of $\pair$. If several points
  overlap, we keep only one of them. See Figure~\ref{fig:matching} for an example. Formally, we denote the operation by
  \begin{align}
 \label{eqn:proj02}
 \tcalA_{\pair}\leftarrow\mth\left[ \bigcup\nolimits_{\vpair \prec \pair,\vpair \in \Phi} \ntree_{\vpair} \right].
\end{align}
  
\end{itemize}

\topic{-- Refinement}
After the projection, we obtain a candidate partition set $\tcalA_{\pair}$ of $\pair$ (defined in~\eqref{eqn:proj01} and ~\eqref{eqn:proj02}).
However, note that the length between the pieces may differ a lot. In order to maintain Property~\ref{prop:p2},
we add some other points to ensure that all non-empty pieces of $\pair$ have approximately the same length.
We call this process  the \emph{refinement} operation.

 W.l.o.g., suppose pair $\pair$ has unit length and $|\tcalA_{\pair}| = n $ and $n > d_0$. 
 \footnote{
 If $n \leq d_0$, we repeatedly split the inner pieces
(i.e., all pieces except for the two pieces incident on the points of $\pair$ ) into two equal-length pieces until the number of the points in $\tcalA_{\pair}$ is larger than $d_0$. }

\eat{
See Figure~\ref{fig:thirdlast} for an example. The point $q_1, q_2$
and $q_3$ are projected points. But $|q_1q_2| \neq |q_2q_3|$.
}
\eat{
\begin{figure}[t]
  \centering
  \includegraphics[width = 0.5\textwidth]{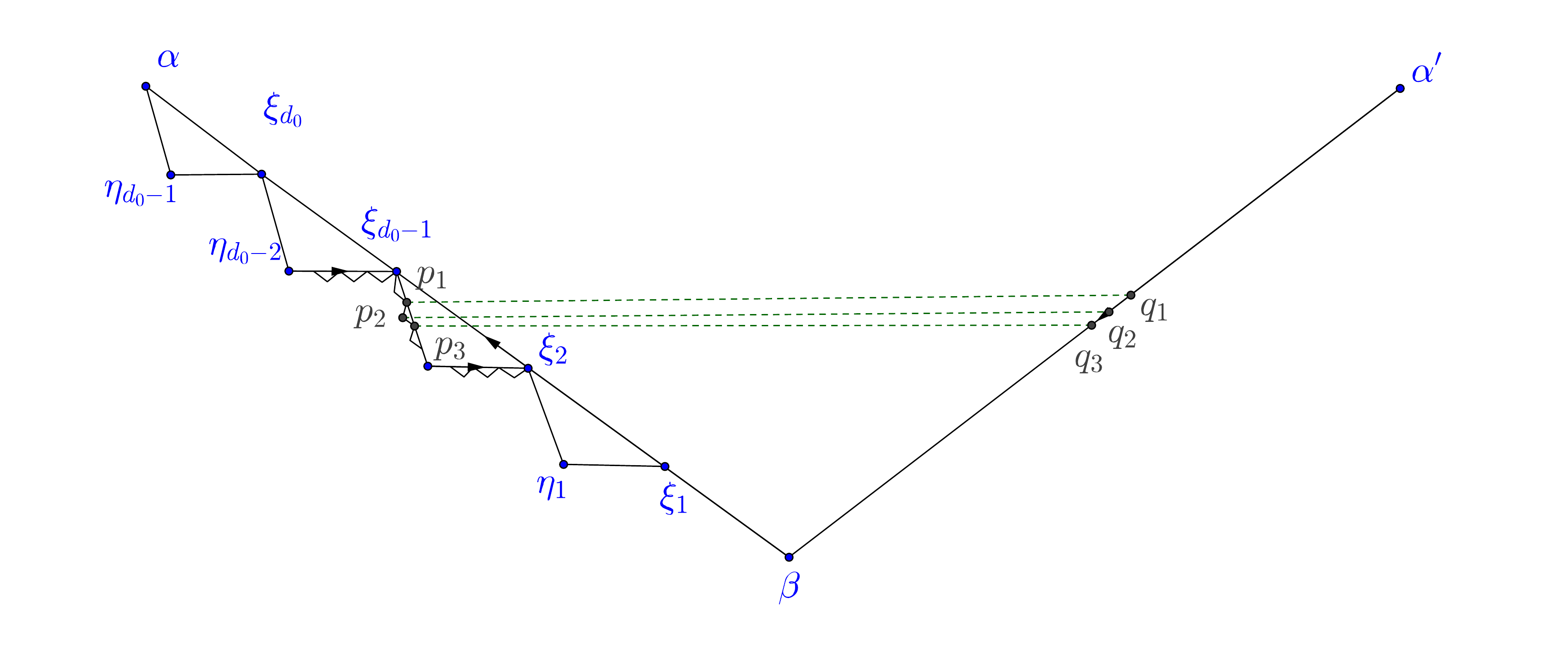}
  \caption{The lengths of the  pieces after projection may be different. Here, the projection points  of $\{p_1, p_2, p_3\}$ are $\{ q_1, q_2, q_3\}$. But $| q_1 q_2| \neq | q_2 q_3|$.   }
  \label{fig:thirdlast}
\end{figure}
}

Suppose $\pair = (u_1, u_2)$. We distinguish into two cases based on whether the first point $u_1$ is
a partition point or an apex point. 

\begin{itemize}
\item  If $u_1$ is an apex point, we mark the piece incident on $u_2$. See Figure~\ref{fig:refinement:01} for an illustration, in which 
 $(u_1, u_2) = (\beta, \alpha_1)$ and piece $\alpha_1\eta_1$ is the marked piece.

\item  If $u_1$ is a partition point, we mark the pieces incident on $u_1$ and $u_2$. See Figure~\ref{fig:refinement:02} for an illustration, in which 
$(u_1, u_2)  = (\alpha_2, \beta)$ and piece $\alpha_2\eta_2$  and $\xi_2\beta$ are the marked pieces.
\end{itemize}

We do not add any point in the marked pieces under refinement.  Consider two sibling pairs $ (\beta, \alpha_1)$ and $ (\alpha_2, \beta)$ where $\beta$ is an apex point. Suppose 
$\alpha_1\eta_1, \alpha_2\eta_2, \xi_2\beta$ are the marked pieces of the two pairs and $\xi_1$ is the point on the segment $\beta\alpha_1$ which is projected to $\xi_2$.\footnote{Point $\xi_1$ must exist since $\xi_2$ is a projected point and there is no point 
in the marked piece $\xi_2 \beta$.} 
Then  $|\beta\xi_1| =  |\beta\xi_2|$ and $|\alpha_1 \eta_1| = |\alpha_2 \eta_2|$ after the refinement. See Figure~\ref{fig:refinement} for an example.

\eat{
\begin{itemize}
\item  If $u_1$ is an apex point, we mark the piece incident on $u_2$ (e.g., piece $\alpha_1\eta_1$ if $(u_1, u_2)$ is $(\beta, \alpha_1)$ in
  Figure~\ref{fig:refinement}).

\item  If $u_1$ is a partition point, we mark the pieces incident on $u_1$ and $u_2$ (e.g., piece
  $\alpha_2\eta_2$  and $\xi_2\beta$ if $(u_1, u_2) $ is $(\alpha_2, \beta)$) in Figure~\ref{fig:refinement}).
\end{itemize}
 }

Denote the length of the $i$th piece (defined by $\tcalA_{\pair}$)
by $\delta_i$.  Let $\delta_{o} =  1/n^2$.  Except for the marked pieces\footnote{Keeping  the  marked pieces  unchanged maintains Property~\ref{prop:p1} and helps a lot to decompose the normal points into hinge sets. See Section~\ref{sec:edge} for details.}, for
  each other piece  which is at least twice longer than $\delta_{o}$, we place
  $\floor{\delta_i/\delta_{o}}-1$ equidistant points on the piece, which divide the piece into
  $\floor{\delta_i/\delta_{o}}$ equal-length  parts.

We call this process the \emph{refinement} and denote
the resulting point set by
\begin{align}
\label{eqn:refine}
 \calA_{\pair}\leftarrow\ffill[\tcalA_{\pair}].
\end{align}

The number of points added in the refinement process is at most
$O(n^2)$ since  the segment of pair $\pair$ has unit length and $\delta_{o} \geq 1/n^2$.
We call each piece whose length is  less than $\delta_{o}$ a \emph{short piece}. The short
pieces remain unchanged before and after the refinement. Moreover, the refinement does not introduce
any new short piece for the pair.

\begin{figure}[t]
  \centering
      \centering
    \captionsetup[subfigure]{justification=centering}
    \begin{subfigure}{0.6\textwidth}
      \centering
        \includegraphics[width = 1\textwidth]{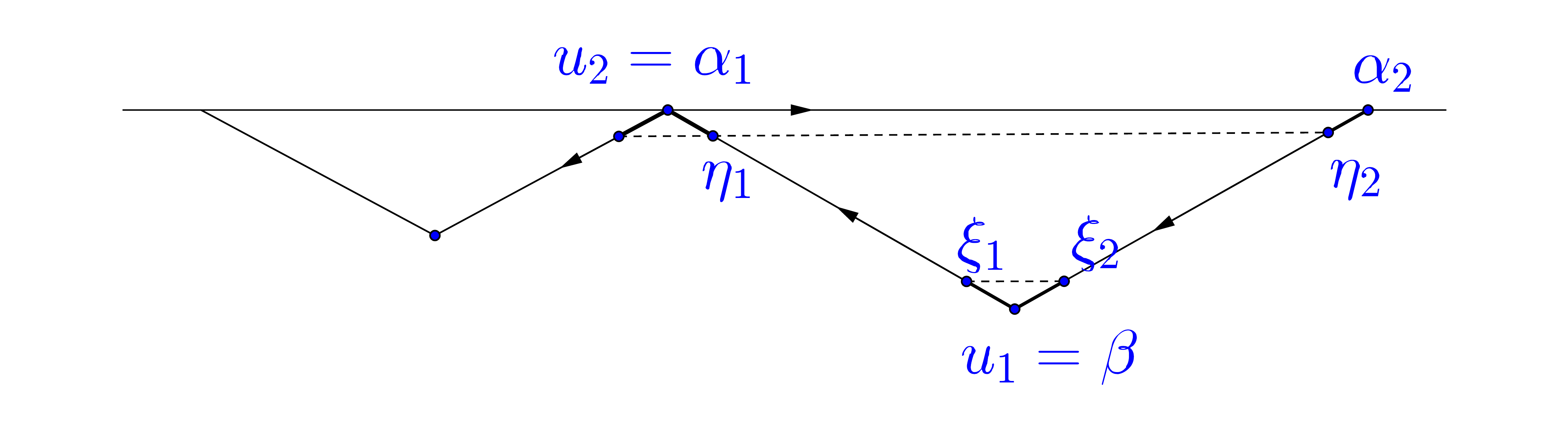}
      \caption{ $(u_1, u_2) = (\beta, \alpha_1)$}
      \label{fig:refinement:01}
    \end{subfigure}
    
    \begin{subfigure}{0.6\textwidth}
      \centering
        \includegraphics[width = 1\textwidth]{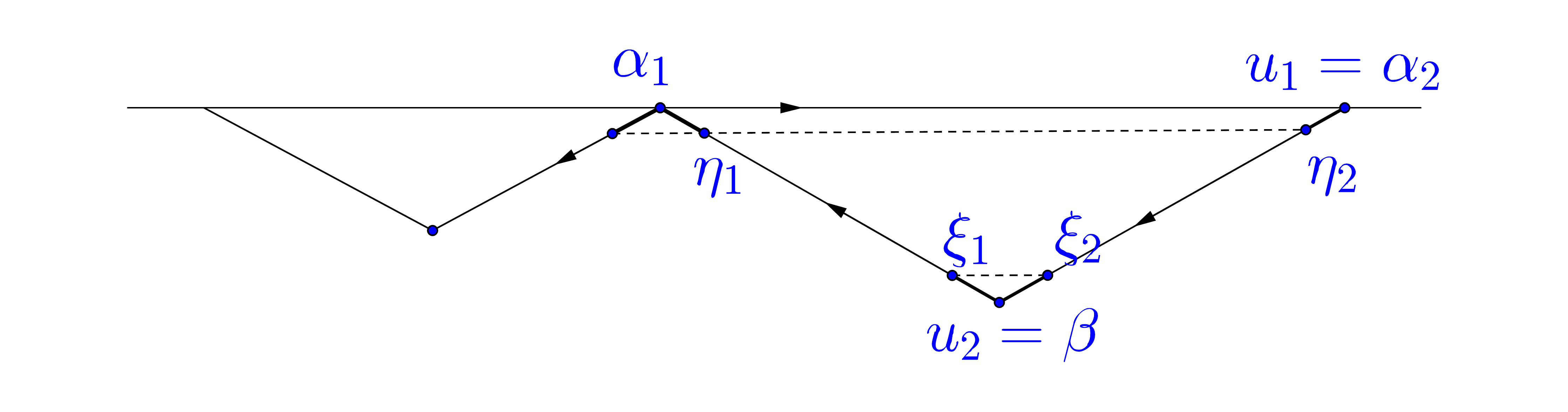}
      \caption{ $(u_1, u_2) = (\alpha_2, \beta)$}
      \label{fig:refinement:02}
    \end{subfigure}
  \caption{ The two cases for refinement. The first case is $\pair = (\beta, \alpha_1)$ in which the first point is an apex point. We mark the piece incident on $\alpha_1$, i.e., the piece $\alpha_1 \eta_1$. The second case is $\pair = (\alpha_2, \beta)$, in which the first point is a partition point. We mark the two pieces incident on $\alpha_2$ and $\beta$ respectively, i.e., the pieces $\alpha_2 \eta_2$ and $\xi_2 \beta$.  Note that after refinement, $|\beta\xi_1| =  |\beta\xi_2|$ and $|\alpha_1 \eta_1| = |\alpha_2 \eta_2|$ since there is no point added on the marked pieces after refinement. }
  \label{fig:refinement}
\end{figure}

\topic{Deciding Emptiness, Leaf-Pairs and Internal-Pairs}
Next, we discuss the principle to decide whether  a piece is empty or non-empty. 
See Figure~\ref{fig:shallow} for an illustration.
 Consider a pair $\pair $ whose  apex point is $\beta $ and partition point is $\alpha$.\footnote{Note that the
  first point of a pair can be either apex point or partition point. Here, $\pair = (\alpha, \beta)$
  or $\pair = (\beta, \alpha)$ depending on whether  first point of $\pair$ is apex point or not.
} We let  the piece incident on the apex point $\beta$ and
the short pieces be \ttype\ and the other pieces be non-\ttype.

For each non-\ttype\ piece, we generate one apex point. The apex set
$\calB_{\pair}$ induces the set $\Phi$ of child-pairs of $\pair$. The types of these pairs are determined as follows.
 Let the three pairs closest to $\alpha$ and two pairs closest to $\beta$ be \emph{leaf-pairs}. We do not further expand the tree from the leaf-pairs. Let the 
other pairs be the \emph{internal-pairs}. 
Naturally, there is no leaf-pair between
any two internal-pairs among pairs in $\Phi$. Hence, Property~\ref{prop:leaf-internal} maintains in the construction. Further, we can see that each other normal point belongs to at most two pairs, except for the two points in the root pair.

For convenience, we call the piece generating two
leaf-pairs \emph{a near-empty piece} and generating one leaf-pair and one internal-pair
\emph{a half-empty piece}.   Note that the near-\ttype\ and half-\ttype\ pieces are special non-\ttype\ pieces.  We can see that, except for the \stype\ piece
incident on the partition point in $\pair$ (e.g., $\mu_1\xi_{d_0}$ in Figure~\ref{fig:shallow}), the
maximum length among the non-\ttype\ pieces is at most twice longer than the minimum one according
to refinement.

\begin{figure}[t]
  \centering
  \includegraphics[width = 0.7\textwidth]{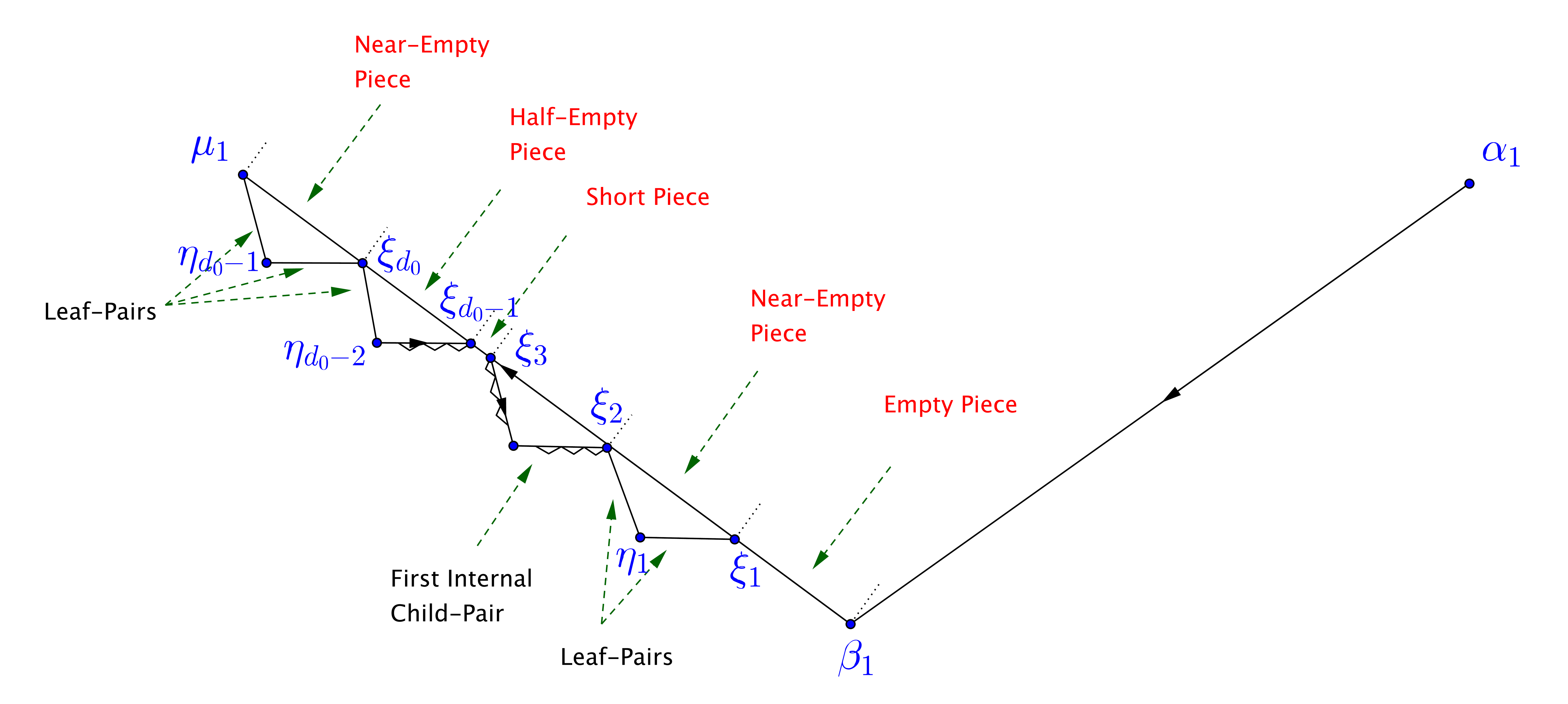}
  \caption{ The figure illustrates the emptiness of each piece.  Consider the pair $(\beta_1, \mu_1)$
    with partition points after refinement.  Segment  $\mu_1 \xi_{d_0}$ and $\xi_1\xi_2$ are two
    \stype\  pieces and $\beta_{1}\xi_1$ is an \ttype\  piece. Pair $( \mu_1, \eta_{d_0-1})$ and $(\eta_{d_0-1},  \xi_{d_0})$, $(\xi_{d_0}, \eta_{d_0-2})$, $(\xi_2, \eta_1),$ $(\eta_1, \xi_1)$  are the five
    leaf-pairs. $\xi_{d_0-1}\xi_{d_0}$  is the half-empty piece. Pair $( \eta_2, \xi_2)$ is the first internal-pair. }
  \label{fig:shallow}
\end{figure}

Overall, after the projection and refinement process, we can generate the gadget for any pair in the tree. We denote this process by
\begin{align}
\label{eqn:gene}
\gadget_{\pair}\leftarrow \projref(\pair).
\end{align}

\begin{figure}[t]
    \centering
    \includegraphics[width = 0.8\textwidth]{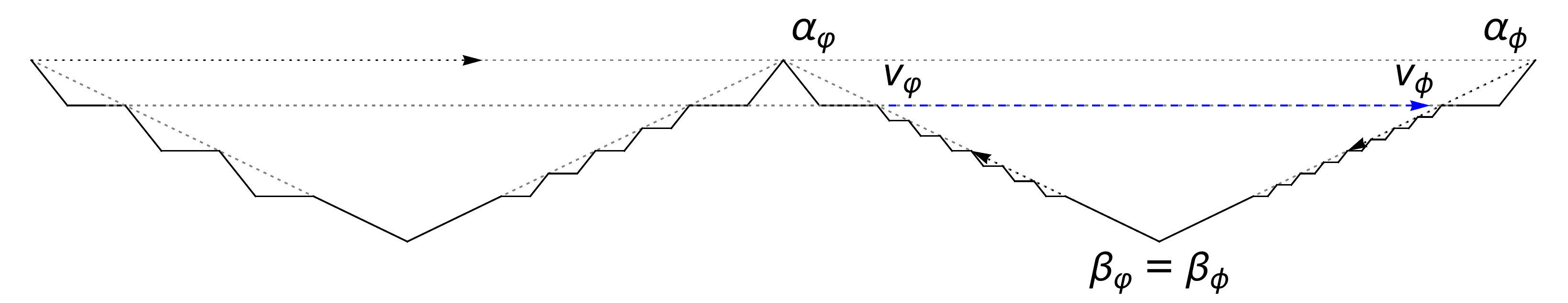}
    \caption{The figure illustrates Property~\ref{prop:p1}. After projection and refinement, $|\alpha_{\pair}v_{\pair}| = |\alpha_{\vpair}v_{\vpair}|$.}
    \label{fig:projequal}
\end{figure}

\begin{property}
\label{prop:p1}
Consider two sibling pairs $\pair$ and $\vpair$. Suppose both of them are internal-pairs and have
 partition point sets $\calA_{\pair}$ and $\calA_{\vpair}$ respectively.  Suppose $\alpha_{\pair} \in \pair$ and $\alpha_{\vpair} \in \vpair$, and both $\alpha_{\pair}$ and $\alpha_{\vpair}$ are partition points.  The  point in   $\calA_{\pair}$ closest to $\alpha_{\pair}$ is
 $v_{\pair}$. Meanwhile,  the point in $\calA_{\vpair}$ closest to $\alpha_{\vpair}$ is
  $v_{\vpair}$. Then $|\alpha_{\pair}v_{\pair}| = |\alpha_{\vpair}v_{\vpair}|$. See Figure~\ref{fig:projequal} for an example.
\end{property}
\begin{proof}
  W.l.o.g., we prove that any two adjacent siblings satisfy the property. Suppose  $\pair$ and $\vpair$ are adjacent siblings. W.l.o.g., assume $\vpair \prec \pair$. $\pair$ has the  candidate partition set  $\tcalA_{\pair}$ after projection. Suppose the
   point in $\tcalA_{\pair}$ closest to $\alpha_{\pair}$  is $\hat{v}_{\pair}$.    According to the
  projection, we know $|\alpha_{\pair}\hat{v}_{\pair}| = |\alpha_{\vpair}v_{\vpair}|$. Since we do not add
  any new point between $\alpha_{\pair}$ $\hat{v}_{\pair}$ after
  refinement,  $\hat{v}_{\pair}$ and $v_{\pair}$ are the same point. Hence, any two adjacent siblings have
  the property.
\end{proof}

\begin{cor}
\label{cor:p1}
  Consider a pair $\pair$  with  partition point set $\calA_{\pair}$. Suppose $\alpha \in \pair$ and $\alpha$ is a partition point. Among all pieces determined by $\calA_{\pair}$, the piece incident on
  $\alpha$ has the maximum length.
\end{cor}
\begin{proof}
It is not difficult to check the root pair holds the property. Then, consider a pair $\hat{\pair}$ with its child pair set $\hat{\Phi}$. We prove that  any pair in $\hat{\Phi}$
holds the property when $\hat{\pair}$ holds the property.   Suppose $\pair $ is the first internal-pair in $\hat{\Phi}$, the corollary is trivially true for $\pair$
  according to the projection process (the first projection case). Otherwise, according to Property~\ref{prop:p1} and the projection process, we know that for any $\pair$ in $\hat{\Phi}$, the piece incident on the partition point in $\pair$ has
  the same length. Thus, any pair in $\hat{\Phi}$ holds the property.
\end{proof}

Now, we prove Property~\ref{prop:p2} that we claimed at the beginning of the construction.

\begin{proofofprop}{\ref{prop:p2}}
  Consider an arbitrary pair with partition point $\alpha$ and the set $\Phi$ of its child-pairs. Consider an internal-pair $\pair \in \Phi$.
  Note that the length of a child-pair  is determined by its corresponding partition piece.
  According to the construction, except for the near-empty piece incident on $\alpha$, the length of any non-empty piece
  is at most twice and  at least half the length of another one.  Thus, except for the sibling pairs generated by the piece incident on $\alpha$, the length of any pair $\vpair \prec \pair$ is at least half of the length of $\pair$.
  Finally, the piece incident on $\alpha$ only induces two leaf-pairs and has the              maximum length among other \ttype\ pieces of
  $\pair$ according to Corollary~\ref{cor:p1}. Hence, we have proven the  property. \hfill  $\blacktriangleleft$ 
\end{proofofprop}

Finally, we summarize the properties of  half-empty, near-empty,  and empty pieces below.

\begin{property}
  \label{prop:piece2}
  Consider an internal-pair $\pair$ with partition point set $\calA_{\pair}$.
 Suppose the length of $\pair$ is $\delta$. The pieces determined by $\calA_{\pair}$ have the
 following properties.
  \begin{itemize}
  \item The sum of  lengths of empty pieces is less than $2\delta/d_{0}$.
  \item  There are two near-empty pieces with sum of lengths  less than $3\delta/d_{0}$.
  \item  There is one half-empty piece  with length less than $\delta/d_0$.
  \item The sum of lengths of empty, near-empty and half-empty pieces  is less than $6\delta/d_0$.
  \end{itemize}

\end{property}

\begin{proof}

  Consider the first property.  Suppose $\beta \in \pair$ and $\beta$ is an apex point.
  There are two kinds of empty pieces.  One is the  short pieces   and the other  is  a
  piece incident on  $\beta$ (denoted by $\xi\beta$). First, the sum of lengths of the short pieces is less than $\delta/d_0$.
  Because the length of each short piece is less than $\delta /n^2$ and there are less than $n$ short
  pieces where $n > d_0$ is the number of partition points in $\tcalA$ after projection and before refinement.  On the other hand,   we prove that  the length of $\xi\beta$  is less   than $\delta/d_{0}$. If  $\beta$ is the first point in
  $\pair$ (refer to $\pair = (\beta, \alpha_1)$ in Figure~\ref{fig:refinement}), according to refinement (the first refinement case), the length of $\xi\beta$  is less
  than $\delta/d_{0}$. Next consider the case that $\beta$ is the second point in $\pair$ $\pair = (\alpha_1, \beta)$ in Figure~\ref{fig:refinement}). Suppose
  $\pair$ shares the point $\beta$ with its sibling $\vpair$. Hence, $\pair$ and $\vpair$ share the point
  $\beta$ and $\beta$ is the first point in $\vpair$.  Denote the  piece of $\vpair$ incident on $\beta$ by   $\eta\beta$. In this case, we
  have that $\pair$ and $\vpair$ have the same length and $|\xi\beta| =   |\eta\beta|$. Since $\beta$ is
  the first point in $\vpair$, we have proven that $|\eta\beta| \leq \delta/d_{0}  $. Thus, $|\xi\beta|  \leq \delta/d_{0} $.

  Consider the second property. Suppose $\alpha$ is a partition point and $\beta$ is an apex point, and $\alpha, \beta \in \pair $.
	First, consider  the \stype\ piece incident on $\alpha$.   If $\pair$ is the first internal-pair of its parent, according to projection, we add a point $\lambda$ on
  the segment of $\pair$ such that $ |\lambda \alpha| =  \delta/d_0$.   Otherwise, according to
  Property~\ref{prop:p2} and~\ref{prop:p1}, we know the piece incident on $\alpha$ is at most
  $2\delta/d_0$. Second, we consider the other \stype\ piece closer to $\beta$. Its length is no
  more than  $\delta/d_0$ based on refinement. Thus, the sum of lengths of \stype\ pieces is less
  than $3\delta/d_{0}$.

  For the third property,  through refinement, the length of  half-empty piece is  less than
  $\delta/d_{0}$.   Above all, we get the fourth property.
\end{proof}


\newcommand{\horder}{\prec_h}

\section{Hinge Set Decomposition of  the Normal Points}
\label{sec:edge}
All points introduced so far are referred to as normal points and  their positions have been defined 
 exactly.  Recall that we denote the set of normal points by $\normp$.  In this section,  decompose $\normp$ into a collection of sets of points
 such that each normal point exactly belongs to one set.  We call these sets \emph{hinge sets}.
  See Figure~\ref{fig:longshort} for an overview of the hinge set decomposition.

 Based on the hinge sets, the edges among normal
points in $\YY{2k+1}(\normp)$ can be organized in the clear way. For convenience, we regard the Yao-Yao graph as a directed graph. Recall the construction of
the directed Yao-Yao graph in Algorithm~\ref{alg:yy}. Note that $C_u(\gamma_1, \gamma_2]$
represents the cone with apex $u$ and consisting of the rays with polar angles in the half-open
interval $(\gamma_1, \gamma_2]$ in counterclockwise. We call the first iteration (line 2 to 5) the \emph{Yao-step} and
call the second iteration (line 6 to 9) the \emph{Reverse-Yao step}.

Then, we define a \emph{total order} among hinge sets. We call an edge in the Yao-Yao graph \emph{a hinge connection} which connects any two points in the same  hinge set or in two adjacent hinge sets
w.r.t. the total order.   Call other edges   \emph{long range connections}. In Section~\ref{sec:aux}, we prove that we can break all long range connections without introducing
new ones by adding some auxiliary points. In Section~\ref{sec:length},  we show that, in the graph
with only hinge connections, the shortest path between the two points of the root pair approaches infinity.

\begin{algorithm}[t]
  \caption{Construct the Yao-Yao graph}
  \label{alg:yy}
  \KwData{A point set $\pset$ and an integer $k \geq 2$}
  \KwResult{$\YY{2k+1}(\pset)$  }
  Initialize: $\theta = 2\pi / (2k+1) $ and two empty graphs $\Yao_{2k+1}$ and $\YY{2k+1}$ \;
  \ForEach{point $u$ in $\pset$}
  {
    \ForEach{$j$ in $[0,2k]$}{
      Select $v $ in $ C_{u}(j\theta, (j+1)\theta]$ such that $|uv|$ is the shortest \;
      Add edge $\ov{uv}$ into $\Yao_{2k+1}$ \;
    }
  }
  \ForEach{point $u$ in $\pset$}
  {
    \ForEach{$j$ in $[0,2k]$}{
      Select $v$ in $ C_{u}(j\theta, (j+1)\theta]$, $\ov{vu}\in \Yao_{2k+1}$ such that $|uv|$ is the shortest \;
      Add edge $\ov{vu}$ into $\YY{2k+1}$ \;
    }
  }
  \KwRet{$\YY{2k+1}$} \;

\end{algorithm}

\begin{figure}[t]
    \centering
    \includegraphics[width = 0.95\textwidth]{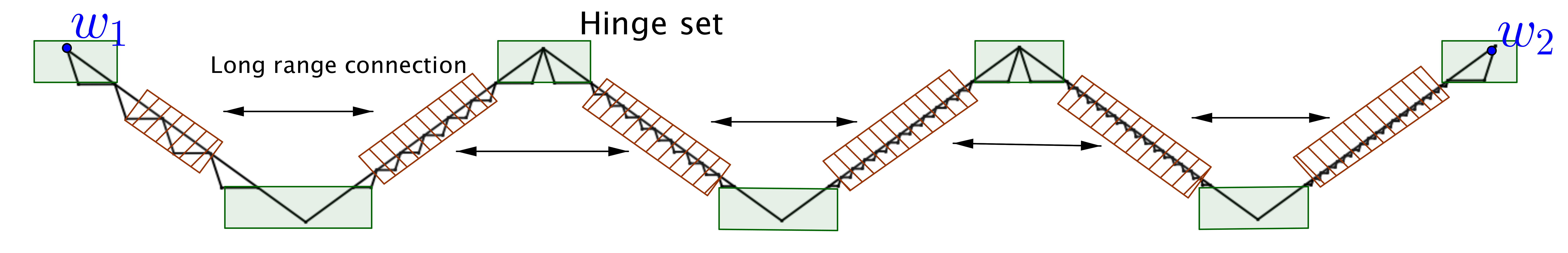}
    \caption{The overview of hinge set decomposition. Roughly speaking, each set of points covered by a green rectangle \drawcover{ color = {rgb,255:red,0; green, 100; blue,0}, fill ={rgb,255:red,0; green,100; blue,0}, fill opacity = 0.1} is a hinge set. Recursively, we can further decompose the points covered by shadowed rectangle \drawcover{ color = {rgb,255:red,153; green,51; blue,0}, pattern=north east lines, pattern color={rgb,255:red,153; green,51; blue,0}} into hinge sets.   The hinge connections are the edges between any two points in a hinge set  or between two adjacent hinge sets.  The other edges in the Yao-Yao graph are long range connections.     }.
    \label{fig:longshort}
\end{figure}

\begin{figure}[t]
\captionsetup[subfigure]{justification=centering}
    \centering
    \begin{subfigure}{0.5\textwidth}
    \includegraphics[width = 1\textwidth]{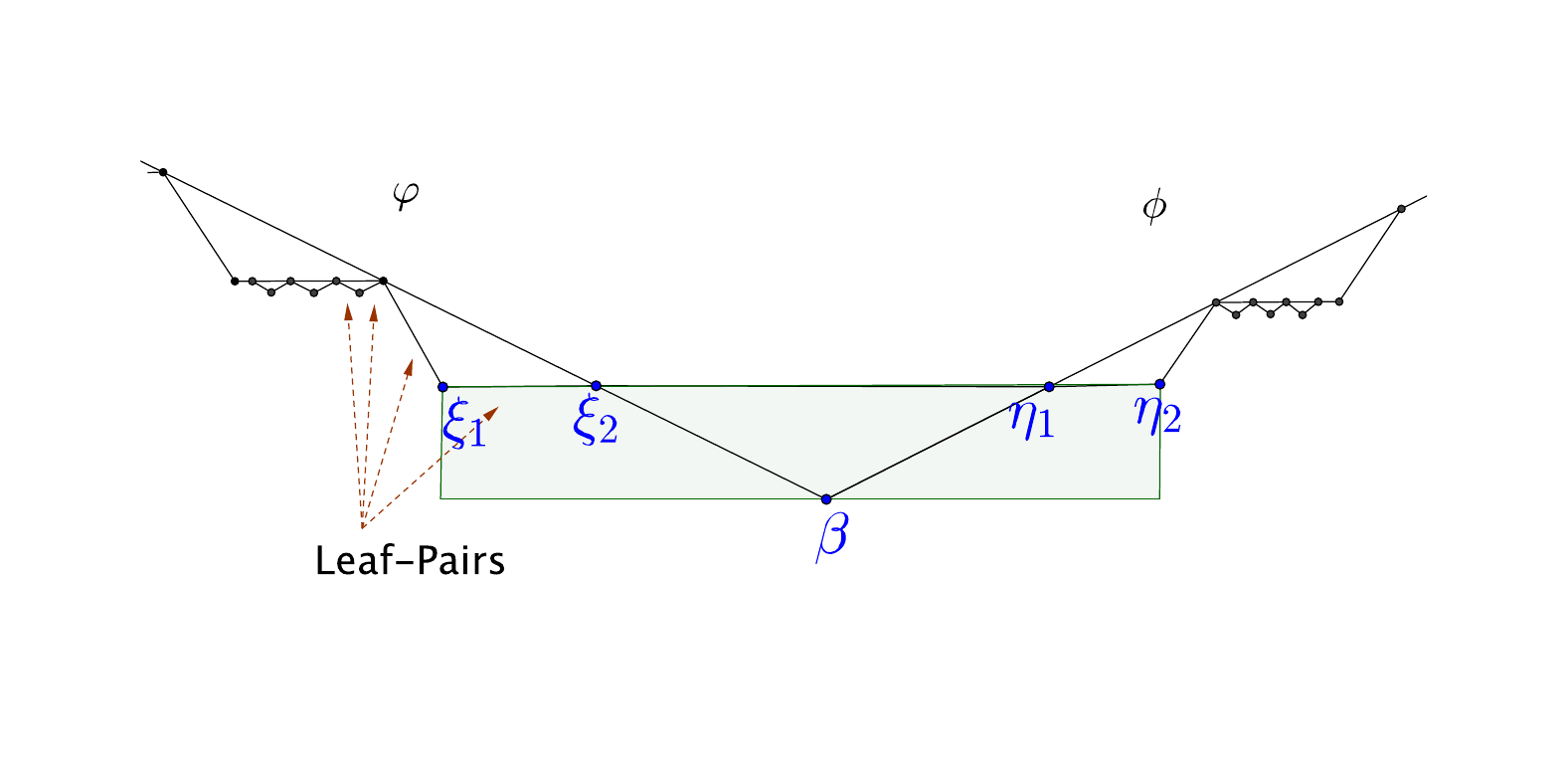}
    \caption{The hinge set centered on  an  apex point}
      \label{fig:hinge:01}
      \end{subfigure}%
    \begin{subfigure}{0.5\textwidth}
      \includegraphics[width = 1\textwidth]{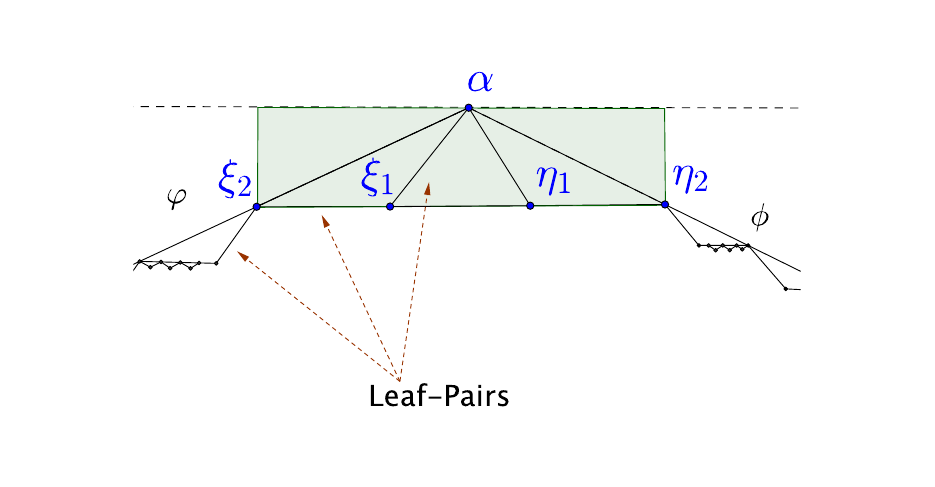}
        \caption{The hinge set centered on  a partition point}
      \label{fig:hinge:02}
    \end{subfigure}

    \begin{subfigure}{0.5\textwidth}
    \includegraphics[width = 1\textwidth]{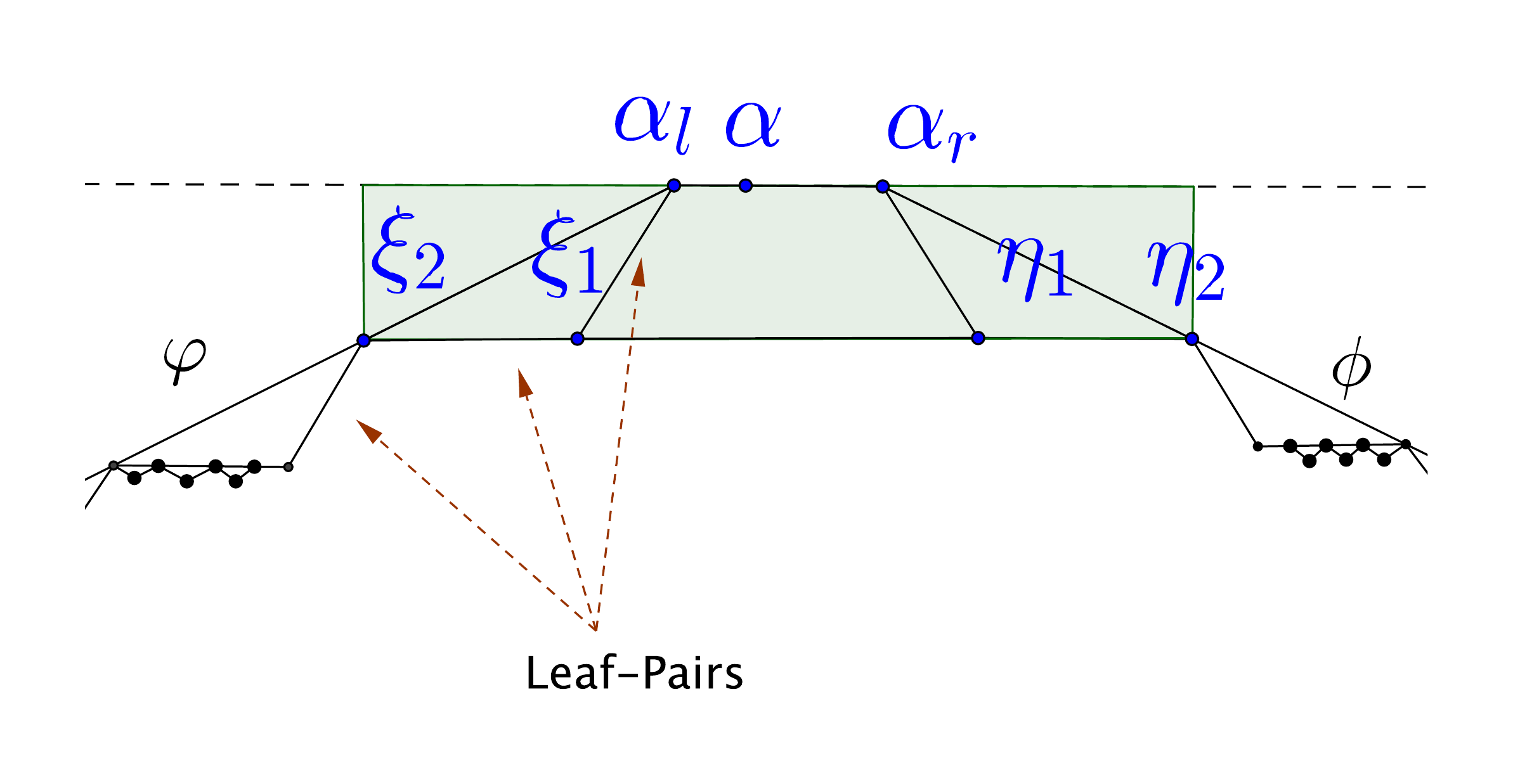}
    \caption{Merge two hinge sets to a new one}
      \label{fig:hinge:03}
    \end{subfigure}%
    \begin{subfigure}{0.5\textwidth}
      \includegraphics[width = 1\textwidth]{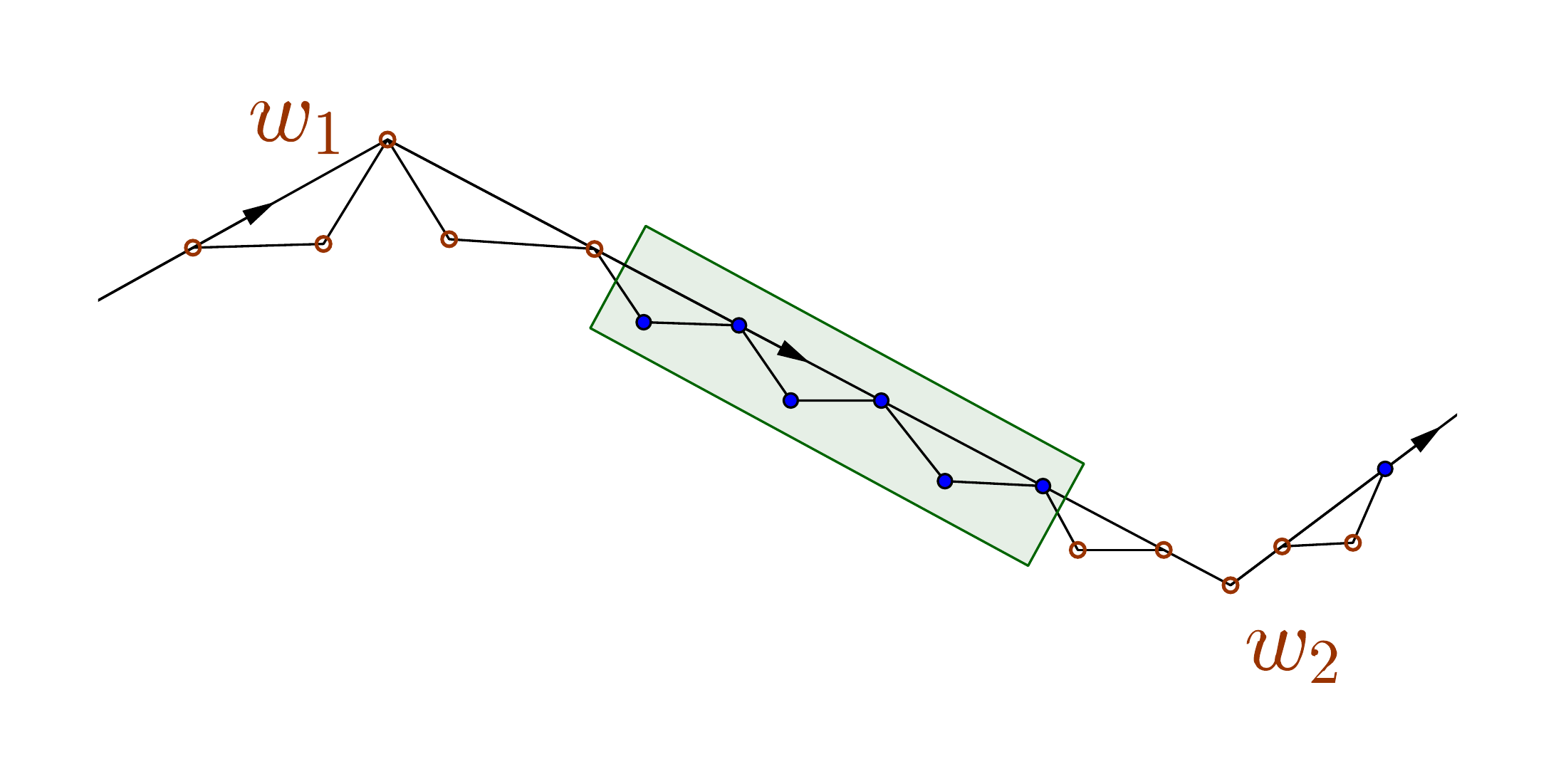}
        \caption{The hinge set consisting of the leaf-pairs at $\level{m}$. $(w_1, w_2)$ is a pair at $\level{(m-1)}$}
      \label{fig:hinge:04}
    \end{subfigure}

    \caption{The hinge sets centered on a point in an internal-pair.}
    \label{fig:hinge}
\end{figure}

\subsection{Hinge Set Decomposition}
\label{subsec:hinge}
 We discuss the process to decompose the set $\normp$ into \emph{hinge
  sets} such that each point in $\normp$ belongs to exactly one hinge set. Briefly speaking, each
hinge set is a set of points which are close  geometrically.

Consider a pair $\hat{\pair}$ at $\level{l}$ ($l < m-1$) with partition point set $\calA_{\hat{\pair}}$ and
apex point set $\calB_{\hat{\pair}}$. Denote the set of the child-pairs of $\hat{\pair}$ by $\hat{\Phi}$.  Recall that we say a point $u$ belongs to $\pair$ (i.e., $u \in \pair$), if the $\pair$ is $(u , \cdot)$ or $(\cdot, u)$. Just for convenient to describe, we call some point \emph{center} of a hinge set and other points \emph{affiliated point}. Formally, the hinge sets are defined as follows. 

\eat{Consider point $u$ in an internal pair $\pair$ on
$\level{t}$, $t \leq m-1$.  According to our construction in Section~\ref{subsec:construc}, the
closest child-pair of $\pair$ to $u$ is a leaf-pair. We group the leaf-pair and $u$ as a hinge
set. Finally, there are remaining some leaf-pairs on $\level{m}$. We group the remaining leaf-pairs with the same 
parent on $\level{(m-1)}$ as a hinge set. Then, we describe the details as follows.
}

\begin{itemize}
\item The hinge set centered on a point $\beta \in \calB_{\hat{\pair}}$ such that $\beta$ belongs
  to one or two internal-pairs in $\hat{\Phi}$:\footnote{
    $\beta$ must belong to two child-pairs of $\hat{\pair}$ since each $\beta$ induces two pairs.  However, $\beta$ may belong to two leaf-pairs (i.e., do not belong to any internal-pair). In this case, $\beta$ is affiliated to a hinge set centered on other point.} 
    We denote the two internal-pairs by $\vpair$ and  $\pair$.\footnote{If one of the two child-pairs is a leaf-pair,  let $\vpair = \emptyset$.}
    See Figure~\ref{fig:hinge:01} for an illustration.
  The hinge set centered at $\beta$ includes: $\beta$ itself, the  child-pair of $\vpair$ closest  to $\beta$
  (denote the pair by $(\xi_1, \xi_2)$) and the  child-pair of $\pair$ closest
  to $\beta$  (denoted the pair by $(\eta_1, \eta_2)$). $ \xi_1, \xi_2, $
  $\eta_1, \eta_2$ are affiliated points. According to the way to determine the leaf-pairs (see  Section~\ref{sec:normal}), they only belong to leaf-pairs. 
  
  \eat{Second, suppose $\beta$ only belongs to
    leaf-pairs,  there is no hinge set centered on $\beta$, i.e., $\beta$ is affiliated to a hinge set centered on another point.
   Actually, in this case, $\beta$ belongs to hinge set
    centered on the points of its parent pair (i.e., $\hat{\pair}$).}

\item The hinge set centered on a point  $\alpha \in \calA_{\hat{\pair}}$ such that $\alpha$ belongs to one or two internal-pairs in $\hat{\Phi}$, or $\alpha$ is an isolated partition point:
    \begin{itemize}
      \item First, suppose  $\alpha$ belongs to one
  or two internal-pairs in $\hat{\Phi}$, which we denote  as  $\vpair$ and $\pair$.\footnote{
    If $\alpha$ belongs to only one internal-pair of $\hat{\Phi}$, let $\vpair = \emptyset$.
} See Figure~\ref{fig:hinge:02}.
  The  hinge set centered on  $\alpha$ includes: $\alpha$ itself, the two
  child-pairs  closest to $\alpha$ of $\vpair$ and $\pair$  (denote the pairs by $(\xi_2, \xi_1)$ and  $(\eta_1, \eta_2)$) respectively.  $ \xi_1, \xi_2, $
  $\eta_1, \eta_2$ are affiliated points which only belong to leaf-pairs.
      \item Second, $\alpha$ is an isolated point in $\calA_{\ntree}$, i.e.,  $\alpha$ is an end point of a short
  piece and  does not belong to any internal-pair in
  $\hat{\Phi}$. See Figure~\ref{fig:hinge:03}. Then, for each direction of segment of
  $\hat{\pair}$,  we find the closest non-isolated  point in $\calA_{\hat{\pair}}$. Denote them by $\alpha_{l}$ and $\alpha_{r}$. Merge the two hinge
  sets centered on $\alpha_{l}$ and $\alpha_{r}$ as a new one and add $\alpha$ to the new hinge
  set.
    \end{itemize}

\end{itemize}

W.l.o.g., we process the points $\mu_1$ and $\mu_2$ in the root pair in the same way as the partition points in $\calA_{(\mu_1, \mu_2)}$.
So far, some points at $\level{m}$ still do not belong to any hinge set.
\begin{itemize}
\item The hinge set consisting of the leaf-pairs at $\level{m}$: Consider any pair $\pair =
  (w_1,w_2)$ at $\level{(m-1)}$. Define the set difference of $\calA_{\pair}\cup \calB_{\pair}$  and the hinge sets centered on  $w_1$ and $w_2$ as a hinge set. \footnote{Although these points form the leaf-pairs at $\level{m}$, these  leaf-pairs are the ``candidate internal-pairs'' to generate  the points at $\level{(m+1)}$. }  See Figure~\ref{fig:hinge:04}.
\end{itemize}
\eat{
  For each maximal set consisting of the points which belong to $\calA_{\pair}\cup \calB_{\pair}$ but do not belong to the hinge sets centered on  $w_1$ and $w_2$,
  we define it as a hinge set.
  }

Overall, we decompose the points  $\normp$ into a collection of hinge sets.
\begin{lemma1}
  Each point $p$ in $\normp$ belongs to exactly one hinge set.
\end{lemma1}

\begin{proof}

  First, we prove that  any two hinge sets are  not overlapping. It means that
 any point in a hinge set does not belong to any other hinge set.  First, consider a point $\lambda$ which only belongs to a leaf-pairs i.e., an affiliated point in the first two type hinge sets (see $\xi_1, \xi_2, \eta_1, \eta_2$ in Figure~\ref{fig:hinge:01}~\ref{fig:hinge:02}~\ref{fig:hinge:03}) or a point in a third type hinge set. It has unique parent-pair $\pair$ such that $\lambda \in \gadget_{\pair}$.  Let $\pair = (\alpha, \beta)$. If $\vpair$ is the closest child-pair to $\alpha$, $\lambda$ belongs to the hinge set centered on $\alpha$. Or if $\vpair$ is the closest child-pair to $\beta$, $\lambda$ belongs to the hinge set centered on $\beta$. Otherwise, $\lambda$ belongs to a third type hinge set. It is not difficult to check that the three cases do non overlap. Thus, point $\lambda$ belongs to at most one hinge set. Next, consider  a point $\lambda$ which belongs to some internal-pair or is an isolated partition point.  Then, $\lambda$  can only belong to the first two type hinge sets. $\lambda$ cannot be a an affiliated point for any hinge set since affiliated point only belongs to leaf-pairs. Besides, according to the definition, the hinge set centered on $\lambda$ is unique. Therefore, $\lambda$ belongs to at most one hinge set.

  On the other hand, we prove that each point in $\normp$ belongs to at least one hinge set.
  First, any point of $\level{m}$ belongs to a hinge set according to the third case. Second, consider a point $\lambda$ on $\level{l}, l<m$.
  If $\lambda$ belong to any internal-pair, it should be a center of a  hinge set. If $\lambda$ is an isolated partition point,
  it merges two hinge sets and  belongs to the new hinge set. If $\lambda$ only belongs to a leaf-pair $\vpair$ and $\vpair $ is a child-pair of $\pair = (\alpha, \beta)$, then, based on the way to determine the leaf-pairs,  $\lambda$ belongs to hinge set centered on $\alpha$ or $\beta$.

  Overall, each point in $\normp$ belongs to exactly one hinge set.
\end{proof}

\begin{figure}[t]
    \centering
    \includegraphics[width = 0.8\textwidth]{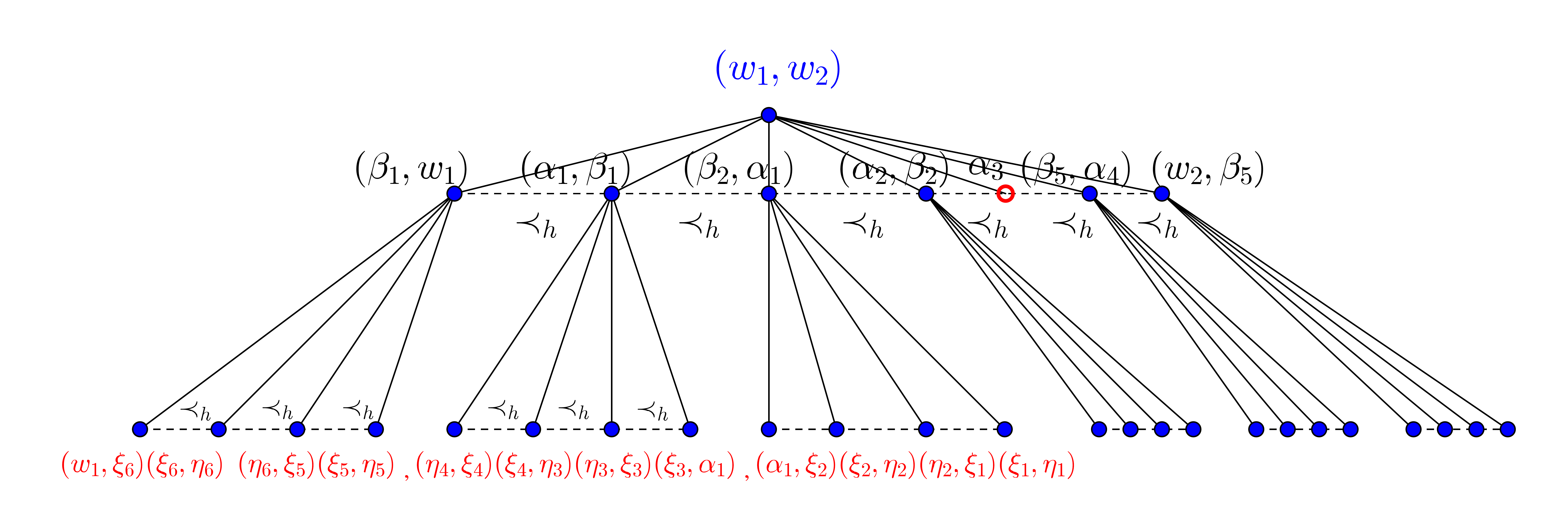}
    \caption{ The illustration for order $\horder$. In the example, notice that the order of the \level{2} in $\rtree$ is
     different with the order in $\ntree$ (see Figure~\ref{fig:tree}).\protect\footnotemark }
    \label{fig:rtree}
\end{figure}
\footnotetext{Note that $\rtree$ only contains the internal nodes of $\ntree$. Consider that \level{2} nodes still have their child nodes. Thus, we do not remove
     the pairs on \level{2} in the example.}

\topic{Order of the hinge sets}  We define the total order of all hinge sets. We denote the order by
``$\horder$'', which is different from the previous order ``$\prec$''. The $\horder$ is in
fact consistent with the order of traversing the fractal path from $\mu_1$ to $\mu_2$.
Rigorously, we define $\horder$ below. For comparison, in Figure~\ref{fig:rtree}, we reorganize the tree in Figure~\ref{fig:tree} according to the order $\horder$.

First, consider the root pair $(\mu_1, \mu_2)$. We denote the hinge set centered on $\mu_1$  by $\hinge{\mu_1}$ and denote the hinge set centered on $\mu_2$  by $\hinge{\mu_2}$.
 Define $\hinge{\mu_1}$ as the first hinge set  and $\hinge{\mu_2}$ as the last hinge set w.r.t. $\horder$. Then, $\hinge{\mu_1} \horder \hinge{\mu_2}$.

Second, we define the orders of other hinge sets. Consider an internal-pair $\pair$ with parent pair $(w_1, w_2)$ (or $(w_2, w_1)$) and $\hinge{w_1} \horder \hinge{w_2}$. Note that  there are two hinge sets centered on the points in $\pair$ respectively. We call the one closer to (in Euclidean distance) $\hinge{w_1}$ the \emph{former hinge set} of $\pair$, denoted by $\lhinge{\pair}$.  Call the other the \emph{latter hinge set} of $\pair$,
denoted by $\rhinge{\pair}$. Let $\hinge{w_1} \horder \lhinge{\pair}\horder\rhinge{\pair} \horder \hinge{w_2}$. Besides, recall that for any internal-pair $\pair$ at $\level{(m-1)}$, the points in $\calA_{\pair} \cup \calB_{\pair}$
 but not in $ \lhinge{\pair}\cup \rhinge{\pair}$ also form a hinge set. We denote it by $\hinge{\pair}$ and define $\lhinge{\pair}\horder  \hinge{\pair} \horder \rhinge{\pair}$.

Note that we have organized all pairs in the recursion tree $\ntree$.
We can transform the tree consisting of all \textbf{internal nodes} of $\ntree$ to a
topological equivalent tree $\rtree$ which has a different ordering of the nodes. The order of the
sibling pairs in $\rtree$ is determined by  their Euclidean distances to the former hinge set of
their parent. Overall, the ordering $\horder$ of the hinge sets can
be defined by a DFS traversing of $\rtree$. When we reach a pair $\pair $ at $\level{l}$ ($l < m-1$)
 for the first time\footnote{$\level{(m-1)}$ is the second to last level of $\ntree$ and the last level of $\rtree$.}, we visit its former hinge set $\lhinge{\pair}$. Next, we recursively traverse its child-pairs in the order we just defined. Then we return to the pair and visit its latter hinge set $\rhinge{\pair}$. When we reach a pair
$\pair$ at $\level{(m-1)}$, we visit  $\lhinge{\pair}, \hinge{\pair},  \rhinge{\pair}$ in order and return.\footnote{
Note that two adjacent sibling pairs share the same hinge set.
So the same hinge set may be visited twice, and the two visits are adjacent in the total order. So it does not affect the order between
two distinct hinge sets.} We denote the procedure by
$\mathsf{TravelHinge(\pair)}$ and the pseudocode can be found in Algorithm~\ref{alg:hinge}.

\subsection{Long Range Connection}
\label{subsec:long}
We call the edges connecting two non-adjacent hinge sets \emph{long range connections}.
\begin{defn}[Long range connection]
  A \emph{long range connection} is an edge connecting two points in
  two non-adjacent hinge sets.
\end{defn}

\begin{algorithm}[t]
	\caption{$\mathsf{TravelHinge}(\pair)$: Travel the hinge sets in the tree $\rtree_{\pair}$}
	\label{alg:hinge}
      $\mathsf{Visit}(\lhinge{\pair})$ \;
      \uIf{ $\pair$ is at $\level{l} $ ($l < m-1$)}{
      \ForEach{child-pair $\vpair$ of $\pair$ in $\rtree_{\pair}$}{
		$\mathsf{TravelHinge}(\vpair)$ \;
      }
      }
      \Else{
        $\mathsf{Visit}(\hinge{\pair}) $
      }
      $\mathsf{Visit}(\rhinge{\pair})$ \;
\end{algorithm}

If there is no long range connection, the total order of the hinge sets corresponds to the ordering of the shortest path from $\mu_1$
 to $\mu_2$ in the final construction.  It means that  each hinge set has at least one point on the
shortest path between $\mu_1$ and $\mu_2$ and the order of these points is consistent with $\horder$.
However, there indeed exist long range connections among normal points. In order to achieve the above purpose, we should
break the edges connecting two non-adjacent hinge sets. Fortunately, the long range connections in $\normp$ have relatively simple form. We claim that after  introducing some auxiliary points (in Section~\ref{sec:aux}),
we can cut the long range connections  without introducing any  new long range connections. Hence, only adjacent hinge sets in the above order $\horder$
have edges in the Yao-Yao graph.

Now, we examine the long range connections in $\YY{2k+1}(\normp)$.
  First, we show that we only need to consider the long
range connections between the points in $\ntree_{\pair}$ and $\ntree_{\vpair}$ for any two sibling pairs
$\pair$ and $\vpair$. Recall that $\ntree_{\pair}$ denotes the subtree rooted at $\pair$ (including $\pair$).
If there exist two points $p \in \ntree_{\pair} - \ntree_{\pair} \cap \ntree_{\vpair}$
and $q \in \ntree_{\vpair} - \ntree_{\pair} \cap \ntree_{\vpair}$ such that $pq$ is a long range connection, we say there is a long range connection between $\ntree_{\pair}$ and $\ntree_{\vpair}$.  

\begin{claim}
\label{claim:nolong}
 Suppose that  for any two sibling pairs $\pair$ and $\vpair$ in $\ntree$ at $\level{l}$ for $l \leq m-1$, there is
no long range connection between the points in $\ntree_{\pair}$ and $\ntree_{\vpair}$.
 Then, there
is no long range connection.
\end{claim}

\begin{figure}[t]
\captionsetup[subfigure]{justification=centering}
    \centering
    \begin{subfigure}{0.5\textwidth}
    \includegraphics[width = 1\textwidth]{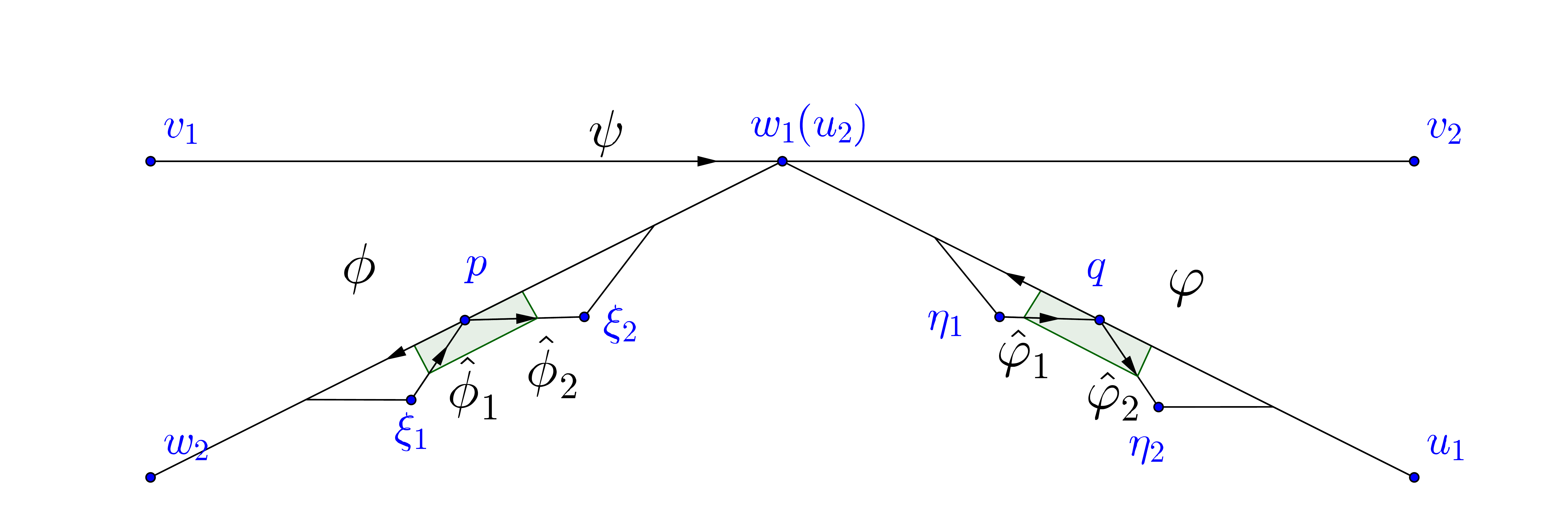}
    \caption{$\psi = (v_1,v_2), \pair = (w_1, w_2), \vpair  = (u_1, u_2),   $
  $  \hat{\vpair}_{1} = ( \eta_1,q), \hat{\vpair}_{2} = (q, \eta_2),$  $ \hat{\phi}_{1} = (\xi_1, p), $
    $\hat{\phi}_{2} = ( p, \xi_2)$, $w_1$ may equal to $u_2$ }
      \label{fig:hinge2treea}
      \end{subfigure}%
    \begin{subfigure}{0.5\textwidth}
      \includegraphics[width = 1\textwidth]{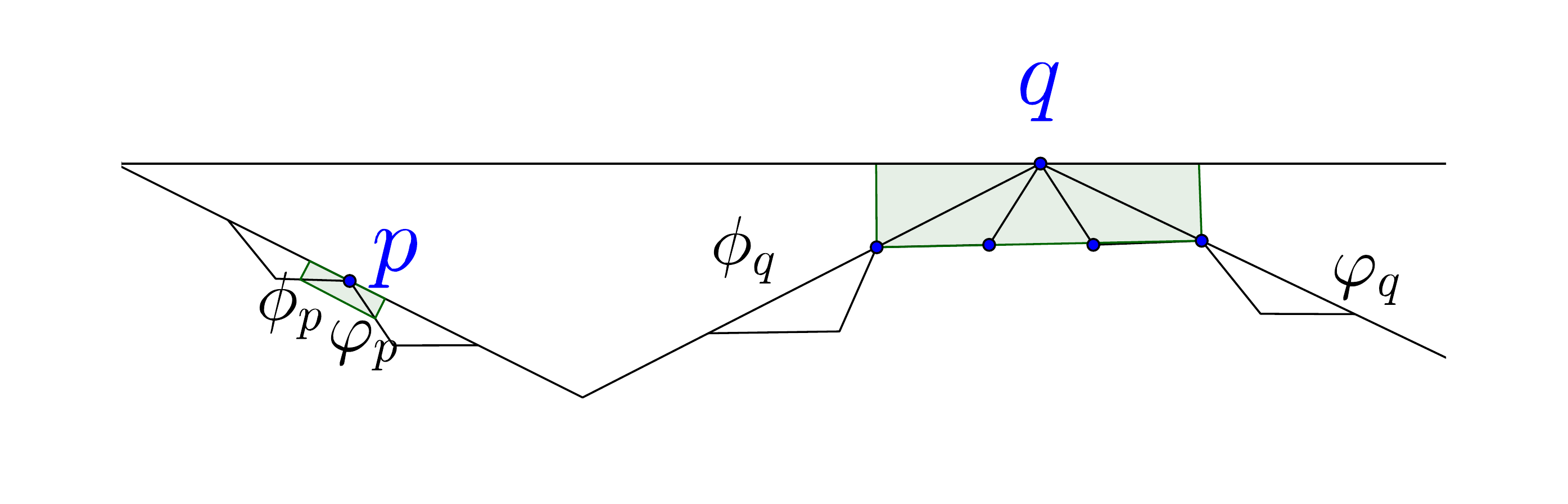}
    \caption{ $  \psi = \vpair  = (v_1, v_2), \pair = (w_1, w_2),  $   $  \hat{\vpair}_{1} = ( q, \eta_1), $
     $\hat{\vpair}_{2} = (\eta_2, q),$  $ \hat{\phi}_{1} = (\xi_1, p), $   $\hat{\phi}_{2} = ( p, \xi_2)$,
       $w_1$ may equal to $\eta_1$}
      \label{fig:hinge2treea2}
    \end{subfigure}

    \begin{subfigure}{0.5\textwidth}
          \includegraphics[width = 1\textwidth]{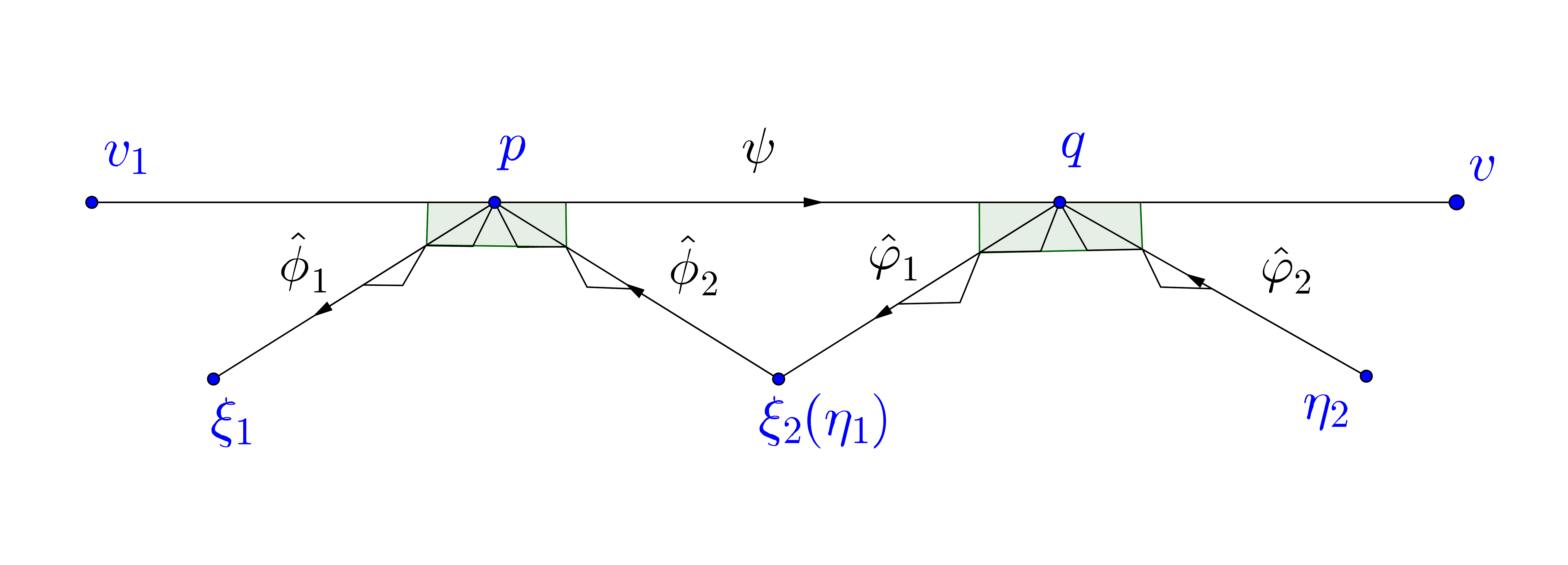}
    \caption{$\psi =\pair = \vpair = (v_1,v_2), \hat{\vpair}_{1} = (q, \eta_1),\hat{\vpair}_{2} = (\eta_2, q),$
    $ \hat{\phi}_{1} = (p, \xi_1), $  $\hat{\phi}_{2} = (\xi_2, p)$ , $\xi_2$ may equal to $\eta_1$}
      \label{fig:hinge2treeb}
    \end{subfigure}%
    \begin{subfigure}{0.5\textwidth}
    \includegraphics[width = 1\textwidth]{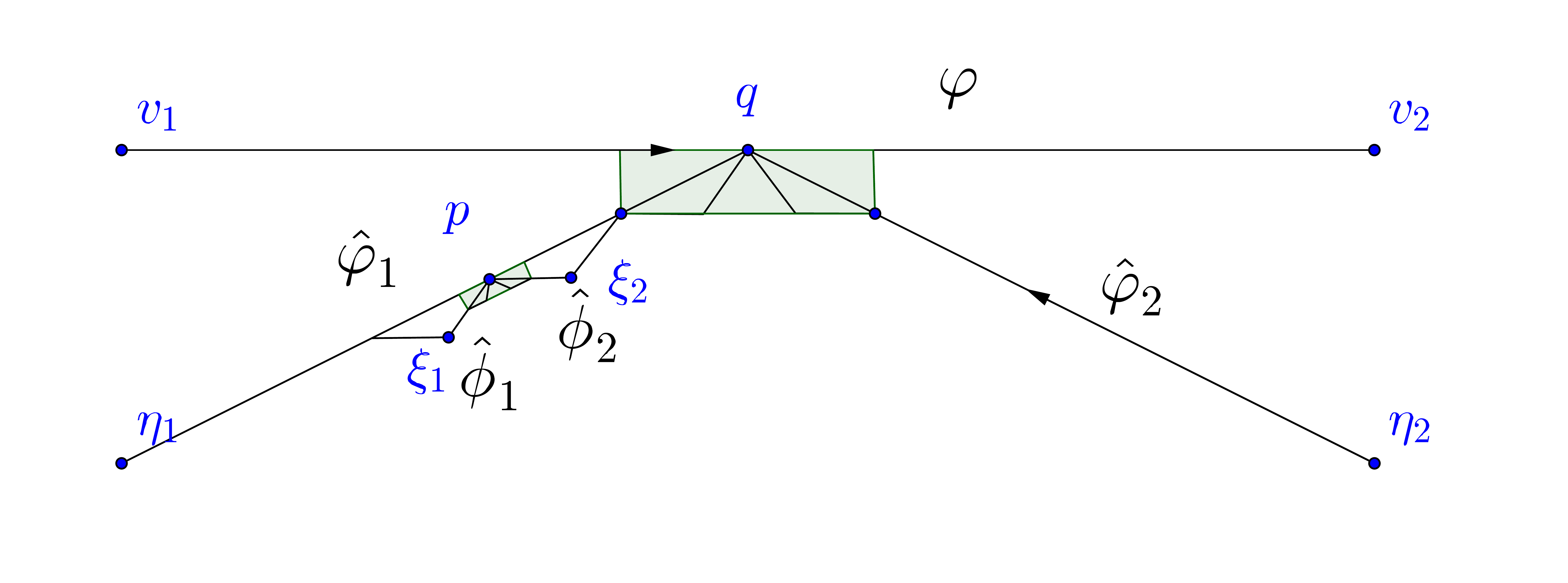}
    \caption{$\vpair = (v_1,v_2), \hat{\vpair}_{1} = (q, \eta_1),\hat{\vpair}_{2} = (\eta_2, q),$  \newline  $ \hat{\phi}_{1} = (\xi_1, p), $ $\hat{\phi}_{2} = ( p, \xi_2)$ }
      \label{fig:hinge2treec}
    \end{subfigure}
    \caption{The possible relative positions of $p$ and $q$. Although, in the figure, $p$ and $q$ are partition points, it  does not induce any new case when $p$ or $q$ is  an apex point according to our divided condition in the proof.}
    \label{fig:hinge2tree}
\end{figure}

\begin{proof}

Consider two non-adjacent hinge sets $\hinge{p}$ and $\hinge{q}$ (if two hinge sets are adjacent, the edges between them are hinge connections) and $\lambda_1 \in \hinge{p}$ and $\lambda_2 \in \hinge{q}$.  We prove that we can find two sibling pairs $\pair$ and $\vpair$ such that $\ntree_{\pair}$ and $\ntree_{\vpair}$ contains $\lambda_1$ and $\lambda_2$ respectively.
Then, we consider the possible cases about the two non-adjacent hinge sets.

First, we consider the case that  each of the two hinge sets is centered on a point
of some internal-pair. Denote the two center points by $p$ and $q$ and the two hinge sets by $\hinge{p}$ and
$\hinge{q}$.  $p$ belongs to one or two adjacent internal-pairs. W.l.o.g., suppose they are $\hat{\pair}_1$ and
$\hat{\pair}_2$ ($\hat{\pair}_2 = \emptyset$ if $p$ only belongs to one internal-pair).  Meanwhile, $q$
belongs to one or two adjacent internal-pairs. Suppose they are $\hat{\vpair}_1$ and $\hat{\vpair}_2$.
 W.l.o.g., suppose the level of $\hat{\vpair}_1$ and $\hat{\vpair}_2$ is no more than the level of $\hat{\pair}_1$ and
$\hat{\pair}_2$. Then, we distinguish  two cases. In the first one,  none of
$\hat{\vpair}_1$ and $\hat{\vpair}_2$ is an ancestor of   $\hat{\pair}_1$ and $\hat{\pair}_2$. Otherwise, it is the second case.

Consider the first case. Suppose the closest common ancestor of $\hat{\vpair}_1$, $\hat{\vpair}_2$, $\hat{\pair}_1$ and
$\hat{\pair}_2$ is pair $\psi$ in $\ntree$.  If $p$ and $q$ do not belong to $\calA_{\psi} \cup \calB_{\psi}$,
there are  two different child-pairs $\pair$ and $\vpair$ of $\psi$ (see Figure~\ref{fig:hinge2treea}),
such that $\hinge{p}$  belongs to $\ntree_{\pair}$ and $\hinge{q}$ belongs to $\ntree_{\vpair}$. Since
points between $\ntree_{\pair}$ and $\ntree_{\vpair}$ have no long range connection according to the assumption,
 points between $\hinge{p}$ and $\hinge{q}$ have no long range connection.
Then consider  the case that $q$  belongs to $\calA_{\psi} \cup \calB_{\psi}$ (see Figure~\ref{fig:hinge2treea2}). Note that $\hinge{q}$ is a  subset of $\ntree_{\hat{\vpair}_1} \cup
 \ntree_{\hat{\vpair}_2}$. Because there is no long range connections for the points  between $\ntree_{\hat{\vpair}_1}, \ntree_{\hat{\vpair}_2}$ and
 $\ntree_{\pair}$, $\hinge{p}$ and $\hinge{q}$ have no long
 range connection. Finally, if both $p$ and $q$ belong to $\calA_{\psi} \cup \calB_{\psi}$ (see
 Figure~\ref{fig:hinge2treeb}), since there is no long range connection for points  between  $\ntree_{\hat{\pair}_1},
 \ntree_{\hat{\pair}_2}$ and  $\ntree_{\hat{\vpair}_1},  \ntree_{\hat{\vpair}_2}$,  $\hinge{p}$ and $\hinge{q}$ have no long
 range connection.

Consider the second case.  See Figure~\ref{fig:hinge2treec}. W.l.o.g.,  suppose $\hat{\pair}_1$ and
$\hat{\pair}_2$ are in  the subtree of $\ntree_{\hat{\vpair}_1}$.
$\hinge{q}$ is a subset of    $\ntree_{\hat{\vpair}_1} \cup \ntree_{\hat{\vpair}_2}$. Moreover, according to the assumption,
the points between $\ntree_{\hat{\vpair}_1} $ and $ \ntree_{\hat{\vpair}_2}$ have no long rang connections.
Since $\hat{\pair}_1$ and $\hat{\pair}_2$ are in  the subtree of
$\ntree_{\hat{\vpair}_1}$, we know that the points in  $\hinge{q} \cap \ntree_{\hat{\vpair}_2}$ have no long range
connections to $\hinge{p}$. Then we consider the long range connection between  $\hinge{q} \cap
\ntree_{\hat{\vpair}_1}$ and $\hinge{p}$. Actually, it is reduced to the first case, thus  they have no long range connection.

Above all, we have discussed the case that  each of the two hinge sets  is centered on a point of some
internal-pair.  Next, suppose that at least one of the hinge sets  is a third type hinge set which contains only leaf-pairs at
$\level{m}$.  If both of them are the third type hinge sets,
denoted by $\hinge{\pair}$ and $\hinge{\vpair}$, there must exist two sibling pairs $\hat{\pair}$
and $\hat{\vpair}$ such that $\hinge{\pair} \in \ntree_{\hat{\pair}}$ and $\hinge{\vpair} \in
\ntree_{\hat{\vpair}}$. According to the hypothesis that there is
no long range connection between the points in $\ntree_{\hat{\pair}}$ and $\ntree_{\hat{\vpair}}$, there is no long range connection
between $\hinge{\pair}$ and $\hinge{\vpair}$. Finally, consider the case that  there is only one
third type hinge set, denoted by $\hinge{\pair}$ and the other is centered on $q$, denoted by
$\hinge{q}$. Suppose $q$ is the shared point of $\hat{\vpair}_2$ and $\hat{\vpair}_1$. We distinguish two
cases according to whether $\pair \in \ntree_{\hat{\vpair}_1} \cup \ntree_{\hat{\vpair}_2} $ or not. As we
have discussed above, we can prove that there is no long range connection between $\hinge{q}$ and $\hinge{\pair}$.
\end{proof}

\begin{figure}[t]
    \centering
    \includegraphics[width = 0.85\textwidth]{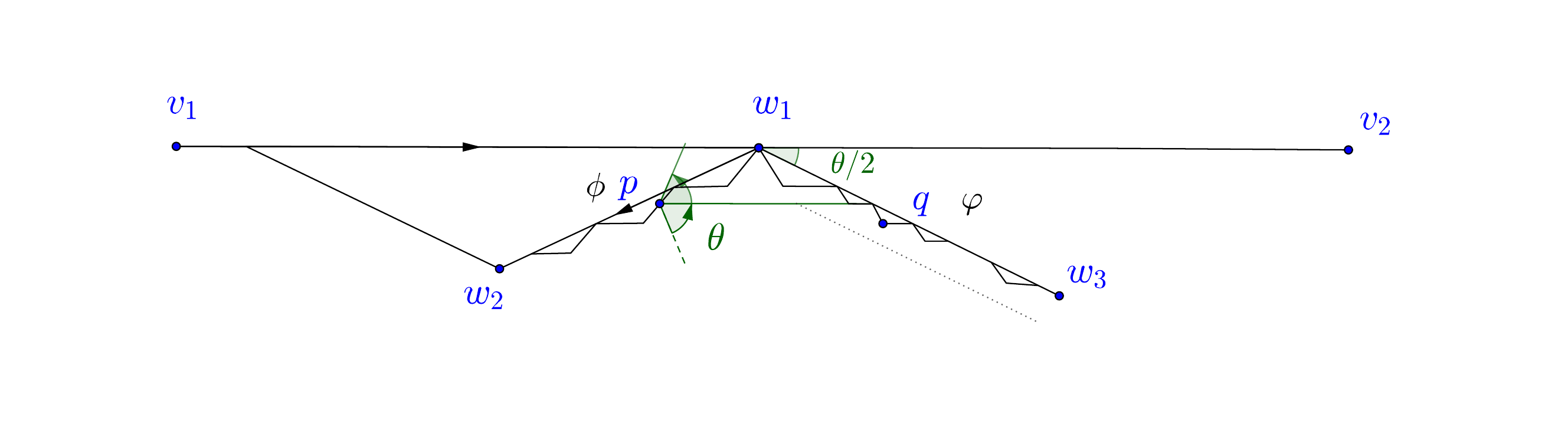}
    \caption{The points of $\ntree_{\vpair}$ locate in  at most two cones of $p$. }
    \label{fig:twocone}
\end{figure}

According to Claim~\ref{claim:nolong}, next, we discuss the possible long range connections between $\ntree_{\pair}$ and $\ntree_{\vpair}$
for two sibling pairs $\pair$ and $\vpair$.  Suppose $p$  belongs to $\ntree_{\pair}$ and $q$
belongs to $\ntree_{\vpair}$.  In the following, we prove that if  the directed edge $\ov{pq}
$ is an edge in $\YY{2k+1}(\normp)$, then  $\pair \prec \vpair$. Moreover, note that the points of
$\ntree_{\vpair}$ locate in  at most two cones of $p$. See Figure~\ref{fig:twocone} for an illustration. $p \in \ntree_{\pair}$ and $\ntree_{\vpair}$ locates in two cones of $p$ according to the angular relation. We prove that for each point $p$, only one of
the two cones may contain a long range connection.  Intuitively, these properties (Observation~\ref{obs:varphi}, Lemma~\ref{lm:left} and~\ref{lm:right}) result from Property~\ref{prop:cone} which does not hold for even Yao-Yao graphs. 

We prove the properties formally below. Consider two sibling pairs $\pair \prec \vpair$. First, suppose $\vpair$ is a leaf-pairs (see Observation~\ref{obs:varphi}). Second, we consider that $\vpair$ is an internal-pair (see Lemma~\ref{lm:left} and~\ref{lm:right}).

\begin{obs}
  \label{obs:varphi}
  Consider two sibling pairs $\pair$ and $\vpair$ such that $\pair \prec \vpair$. If $\vpair$ is  a leaf-pair, there is no long range connection between $\ntree_{\pair}$ and $\ntree_{\vpair}$ (i.e., $\vpair$ itself).
\end{obs}
\begin{proof}
  See Figure~\ref{fig:leftleaf} for an illustration. Suppose $\vpair =(w_3, w_1)$. First, we consider the case that $\pair$ and $\vpair$ share one point. Let  $\pair = (w_1, w_2)$. Then, $\ntree_{\pair} \cap \ntree_{\vpair} = w_1$. We should prove that for any $p \in \ntree_{\pair} - w_1$ (i.e., $p = p^{(1)}$ in Figure~\ref{fig:leftleaf}), there is no edge $pw_3$ in Yao-Yao graph.  Let $(\eta_1, \eta_2)$ be the pair in $\gadget_{\pair}$  closest to $w_1$. There is no edge $\ov{w_3p}$ in the Yao-step since $\eta_2$ and $p$ are in the same cone of $w_3$ and $|\eta w_3| < |p w_3|$.  Note that $\eta_2 w_3$ is a hinge connection. If directed edge $\ov{pw_3}$ is accepted in the Yao-step, $\ov{p w_3}$ cannot be accepted in the reverse-Yao step since $\ov{\eta_2 w_3}$ exists in the Yao-step, and $\eta_2$ and $p$ are in the same cone of $w_3$ and $|\eta_2 w_3| < |p w_3|$. Thus, there is no long range connection between $\ntree_{\pair}$ and $\vpair$.
  
  Second, we consider the that $\pair$ and $\vpair$ do not overlap, i.e., Then, $\ntree_{\pair} \cap \ntree_{\vpair} = \emptyset$.   W.l.o.g., let $\pair = (w_2, w_4)$ and $p \in \pair$ (i.e., $p = p^{(2)}$). Similar to the first case, we can prove that there is no long range connection $pw_3$. Then we consider the point $w_1$. There is edge from $w_1$ to $p$ in the Yao-step since $\eta_1$ and $p$ are in the same cone of $w_1$ and $|w_1\eta_1| < |w_1p|$. Besides, if directed edge $\ov{pw_1}$ is accepted in the Yao-step, $\ov{p w_1}$ cannot be accepted in the reverse-Yao step since $\ov{\eta_1 w_1}$ exists in the Yao-step and $|\eta_1 w_1| < |pw_1|$.
    \end{proof} 
  \eat{ 
   W.l.o.g., suppose $ \pair = (w_1, w_2)$ and $\vpair =(w_3, w_1)$. Then, $\ntree_{\pair} \cap \ntree_{\vpair} = w_1$. We should prove that for any $p \in \ntree_{\pair} - w_1$, there is no edge $pw_3$ in Yao-Yao graph.  Let $\eta$ be the point in $\gadget_{\pair}$  closest to $w_1$. There is no edge $\ov{w_3p}$ in the Yao-step since $\eta$ and $p$ are in the same cone of $w_3$ and $|\eta w_3| < |p w_3|$.  Note that $\eta w_3$ is a hinge connection. If directed edge $\ov{pw_3}$ is accepted in the Yao-step, $\ov{p w_3}$ cannot be accepted in the reverse-Yao step since $\eta w_3$ exists in the Yao-step, and $\eta$ and $p$ are in the same cone of $w_3$ and $|\eta w_3| < |p w_3|$. Thus, there is no long range connection between $\ntree_{\pair}$ and $\vpair$.
   }

\begin{figure}[t]
  \centering
    \label{fig:case1leaf}
    \includegraphics[width = 0.75\textwidth]{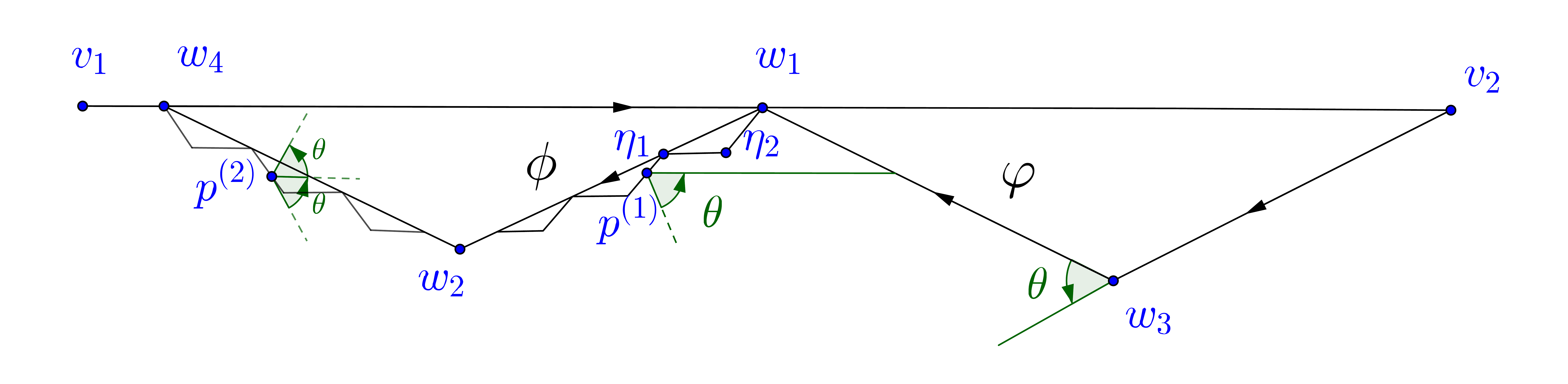}
  \caption{$\pair$ and $\vpair$ are sibling pairs such that $\pair \prec \vpair$. $\vpair$ is  a leaf-pair. There is no long range connection between $\ntree_{\pair}$ and the points of $\ntree_{\vpair}$ (i.e., $\vpair$ itself).}
  \label{fig:leftleaf}
\end{figure}

Given a pair $(v_1, v_2)$ with child-pair set $\sset$, consider two sibling pairs $\pair$
and $\vpair$ in $\sset$ where $\pair = (w_1, w_2)$.  For convenience, let $\angle
u_1u_2$  be the polar angle of vector $u_1u_2$. Let  $\vangle(u_1u_2, v_1v_2)$  be $\angle
v_1v_2 - \angle u_1u_2 $, i.e., the angle from $u_1u_2$ to $v_1v_2$ in the counterclockwise direction.

Recall that there are two kinds of normal points according to the definition of gadget: partition
points and apex points. According to the type of point $w_1$ and the relative
position between $\pair = (w_1, w_2)$ and $(v_1,v_2)$, there are four cases: (1) $w_1$ is a \onseg\ point and
$\pair$ is on the right side of $v_1v_2$, (2) $w_1$ is an \outseg\ point and $\pair$ is on the left
side of $v_1v_2$,  (3) $w_1$ is a \onseg\ point and $\pair$ is on the left side of
$v_1v_2$, (4) $w_1$ is an \outseg\ point and $\pair$ is on the right side of $v_1v_2$.  See
Figure~\ref{fig:left} and~\ref{fig:right} for illustrations.

We prove the possible long rang connections between the points of $\ntree_{\pair}$ and
$\ntree_{\vpair}$ case by case.
Lemma~\ref{lm:left} covers  case (1) and case (2) which satisfy the condition
$\vangle(v_1v_2, w_2w_1) =  \theta/2 $. Lemma~\ref{lm:right} covers case (3) and case (4) which
satisfy the condition $\vangle(v_1v_2, w_2w_1) = -\theta/2$.

\begin{lemma1}
  \label{lm:left}
  Given a pair $(v_1, v_2)$ with child-pair set $\sset$, consider two sibling pairs $\pair$
  and $\vpair$ in $\sset$ where $\pair = (w_1, w_2)$. $\pair$ and $\vpair$ are at
  $\level{l}$ for $l \leq m-1$. Suppose point $p$ belongs to $\ntree_{\pair}$ and $q$ belongs to
  $\ntree_{\vpair}$.  If  $\vangle(v_1v_2, w_2w_1) = \theta/2$ and there is a directed edge from $p $ to $q$ in   $\YY{2k+1}(\normp)$, then
  $\vangle(v_1v_2, pq) = 0$ (i.e., $pq$ is parallel to $v_1v_2$), and $q$ is a point in the gadget
  $\gadget_{\vpair}$ generated by $\vpair$.
\end{lemma1}

\begin{proof}

\begin{figure}[t]
\captionsetup[subfigure]{justification=centering}
  \centering
  \begin{subfigure}{0.8\textwidth}
      \includegraphics[width = 1\textwidth]{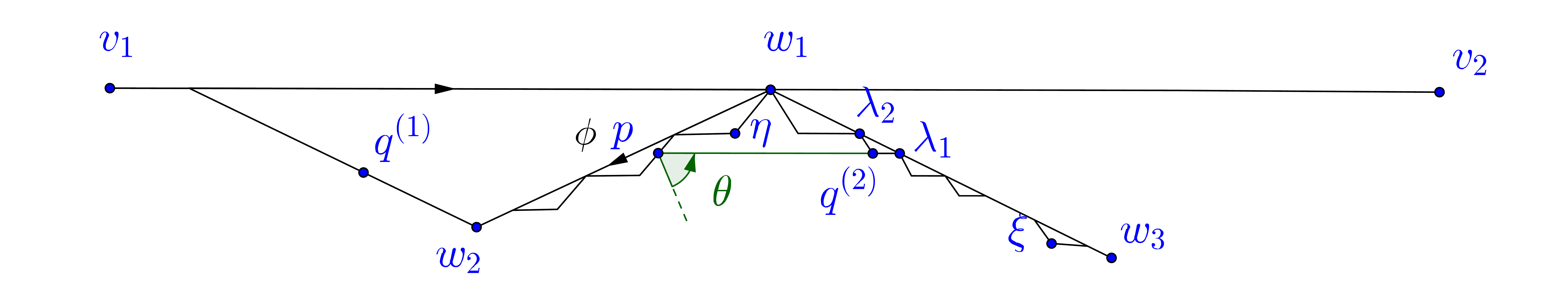}
  \caption{Case 1: $w_1$ is the \onseg\  point and $\pair$ is on the right side of $v_1v_2$}
    \label{fig:case1}
    \end{subfigure}

    \begin{subfigure}{0.8\textwidth}
    \includegraphics[width = 1\textwidth]{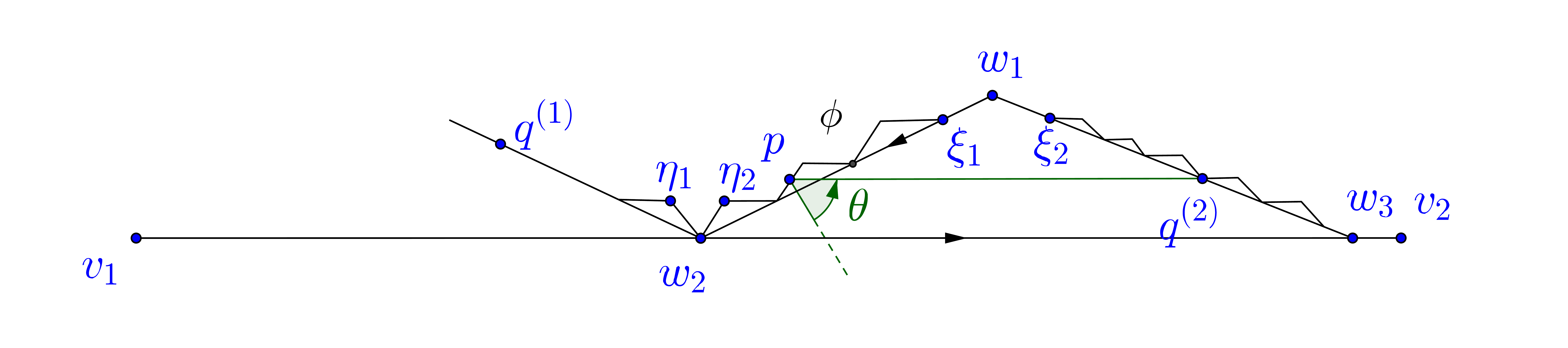}
  \caption{Case 2: $w_1$ is the \outseg\ point  and $\pair$ is on the left side of $v_1v_2$}
    \label{fig:case4}
    \end{subfigure}
  \caption{The two cases about $\vangle(v_1v_2, w_2w_1) = \theta/2$. Here $\pair = (w_1, w_2)$ and  $p \in \ntree_{\pair}$. }
  \label{fig:left}
\end{figure}
 As we have discussed  above, there are two cases under the conditions.   Consider case (1). See
 Figure~\ref{fig:case1}. First, we prove that $\vangle(v_1v_2, pq) $
  should belong to $(-\theta/2, 0]$.   If $q$ belongs to $\ntree_{\vpair}$  and $\vpair \prec
  \pair$ (i.e., $q = q^{(1)} $ in Figure~\ref{fig:case1}),  $w_2$ and $q$  are in the same cone of
  $p$. There is   no edge   from $p$ to $q$  since $|pw_2| < |pq|$ and the edge
  $pq$ is rejected in the Yao-step.  Then
  consider that $q$ belongs to $\ntree_{\vpair}$ and $\vpair \succ  \pair$ (i.e., $q = q^{(2)} $ in
  Figure~\ref{fig:case1}).   According to Observation~\ref{obs:varphi}, we safely assume that $\vpair$ is an
  internal-pair. Denote the point in $\gadget_{\pair}$ closest to $w_1$  by $\eta$. If
  $\vangle(v_1v_2, pq) > 0$, $\eta$ and $q$ are in the same  cone of $p$ since   $w_1\eta$ has the
  maximum length among its  sibling pairs according to Corollary~\ref{cor:p1}. Thus,  there is no
  edge from $p$ to $q$ in the Yao-step since $|p\eta| < |pq|$. Thus, $\vangle(v_1v_2, pq) \in
  (-\theta/2, 0]$. Then, we prove that $\vangle(v_1v_2, pq) = 0$. Suppose the projection point of
  $p$ to pair $\vpair$  is $\lambda_1$ (the $\lambda_1$ must exist according to the
  projection process) and $q^{(2)}$ is an apex point of the piece  $\lambda_1\lambda_2$.   Note that $\theta <
  \pi/3$ for $k \geq 3$ and the maximum length among child internal-pairs  of $\vpair$  is at most twice
  longer than the minimum one (see Property~\ref{prop:p2}).  It is not difficult to check that the point closest to $p$ in cone $C_{p}(-\theta, 0]$
  is $q^{(2)}$. Thus, $q = q^{(2)}$ and $pq$ is parallel to   $v_1v_2$.

  Note that there is a degenerated case in which the projection $\lambda_1$ is an end point of
  an empty  piece. Thus, we do not generate the corresponding apex point $q$.  See
  Figure~\ref{fig:dege1} for an  illustration. $(\xi_2,\xi_1)$ is the pair closest to $\lambda_1$.
  Note that for $k \geq 3$, the angle $ \angle p \lambda_1\xi_2 > \pi/2$ and $(\xi_1,
  \xi_2) $ is a leaf. Thus, in this degenerated case, the  point  closest to $p$ in cone
  $C_{p}(-\theta, 0]$ is  $\lambda_1$. $p\lambda_1$ is also parallel to $v_1v_2$.  Thus,  the lemma is still true.
  We can process  the degenerated case in the same framework in the following and do not distinguish the degenerated
  case particularly.

\begin{figure}[t]
\captionsetup[subfigure]{justification=centering}
  \centering
  \begin{subfigure}{0.5\textwidth}
  \includegraphics[width = 1\textwidth]{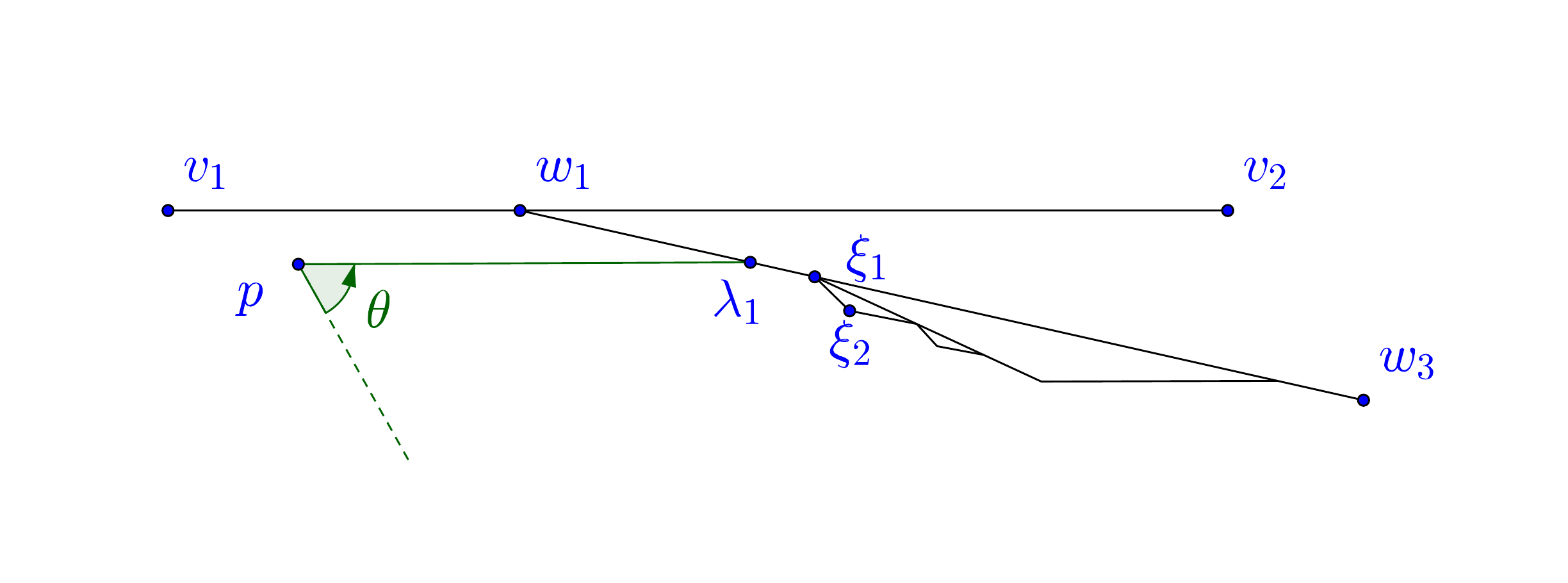}
  \caption{The degenerated case of case 1 in Lemma~\ref{lm:left}}
    \label{fig:dege1}
    \end{subfigure}%
  \begin{subfigure}{0.5\textwidth}
   \includegraphics[width = 1\textwidth]{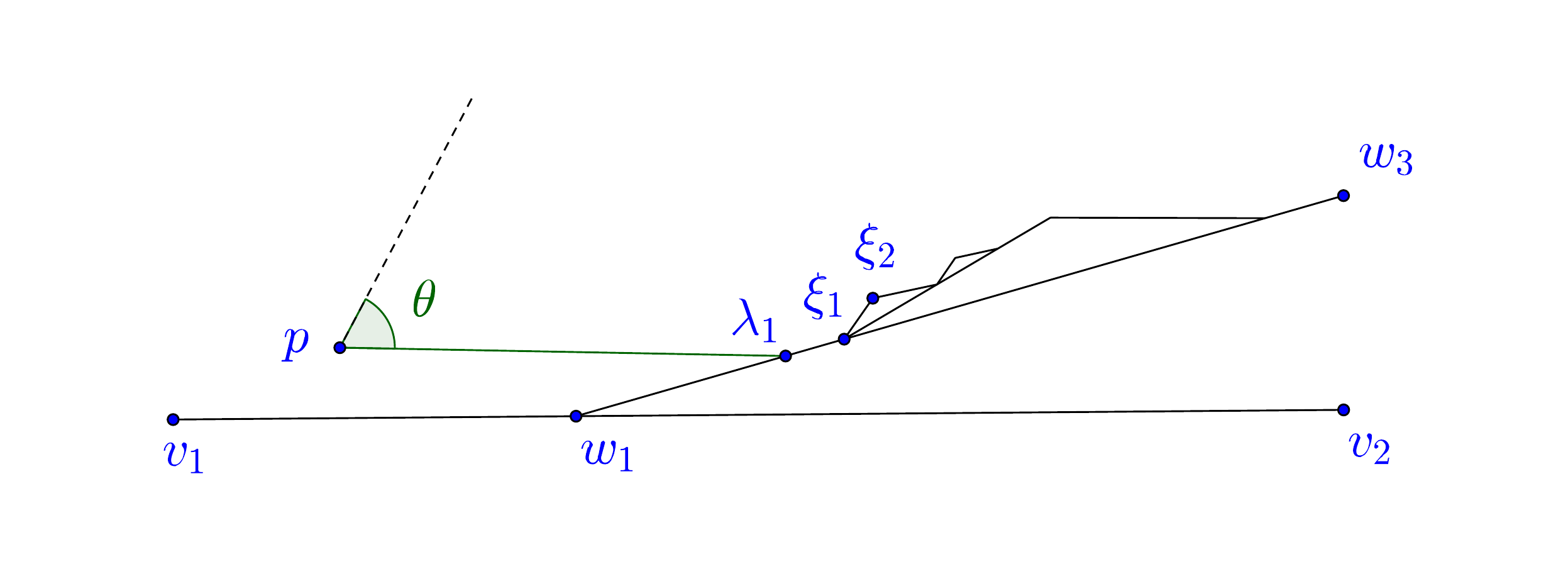}
  \caption{The degenerated case of case 3 in Lemma~\ref{lm:right}}
    \label{fig:dege2}
   \end{subfigure}
  \caption{The degenerated cases in which the projection point of $p$ is an isolated partition point, i.e., $\lambda_1$ in the figure is an isolated partition point which is incident on a short piece. }
  \label{fig:dege}
\end{figure}

  Consider  case (2). See Figure~\ref{fig:case4}. Suppose $\eta_2$ is the apex point of the near-empty
  piece of $\pair$ incident on $w_2$. Note that $w_2\eta_2$ has the maximum length among its sibling
  pairs.  If $q$ belongs to $\ntree_{\vpair}$ and   $\vpair \prec \pair$ (i.e., $q =
  q^{(1)}$ in Figure~\ref{fig:case4}), $\eta_2$ and   $q$ are in the   same cone of $p$ and
  $|\eta_2p| < |qp|$.  Thus, there is no edge from $p$ to $q$ in the Yao-step. Then consider that  $q$
  belongs to  $\ntree_{\vpair}$ and $\vpair \succ \pair$ (i.e., $q =  q^{(2)}$ in
  Figure~\ref{fig:case4}).   According to Observation~\ref{obs:varphi}, we assume that $\vpair$ is an
  internal-pair.  If $\vpair \succ \pair$,
  the polar angle of $pq$ should belong to $( -\theta/2,0]$. If not, $w_1$ and $q$ are in the same
  cone. Thus, there is no edge from $p$ to $q$ in the Yao-step since $|pq| > |pw_1|$.  Then we prove that  $pq$
  is parallel to $v_1v_2$. The point closest to $p$ in the cone $C_{p}(-\theta, 0]$  is the
  projection point of $p$ ($p$ must exist because of the projection). Thus, $pq$ is parallel to
  $v_1v_2$.
\end{proof}

  \begin{lemma1}
    \label{lm:right}
    Given a pair $(v_1, v_2)$ with child-pair set $\sset$, consider two sibling pairs $\pair$
    and $\vpair$ in $\calS$ where $\pair = (w_1, w_2)$. Suppose $\pair$ and $\vpair$ are at
    $\level{l}$ for $l \leq m-1$. Suppose point $p$ belongs to  $\ntree_{\pair} $, and $q$
    belongs to  $\ntree_{\vpair}$. If $\vangle(v_1v_2, w_2w_1) =  - \theta/2$ and there is a directed
    edge from $p $ to $q$ in   $\YY{2k+1}(\normp)$ , then $\vangle(v_1v_2, pq) \in
    (0,\theta/2)$. Moreover, there exists a point $r$ in $\ntree_{\vpair}$ such that $pr$ is
    parallel to $v_1v_2$ and $|pr| < |pq|$. Moreover,  $r$ is a point in the gadget $\gadget_{\vpair}$
    generated by $\vpair$.
  \end{lemma1}

  \begin{figure}[t]
  \captionsetup[subfigure]{justification=centering}
    \centering
     \begin{subfigure}{0.8\textwidth}
      \includegraphics[width = 1\textwidth]{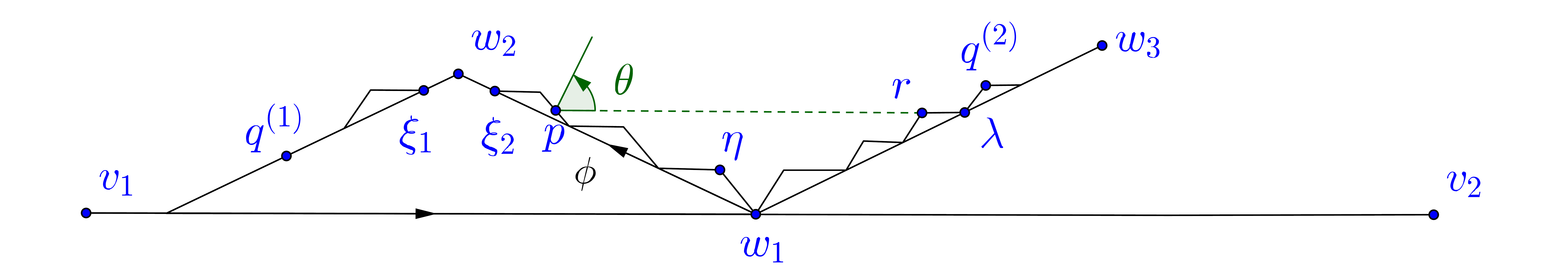}
        \caption{Case 3: $w_1$ is the \onseg\ point and $\pair$ is on the left side of $v_1v_2$}
      \label{fig:case2}
      \end{subfigure}

      \begin{subfigure}{0.8\textwidth}
      \includegraphics[width = 1\textwidth]{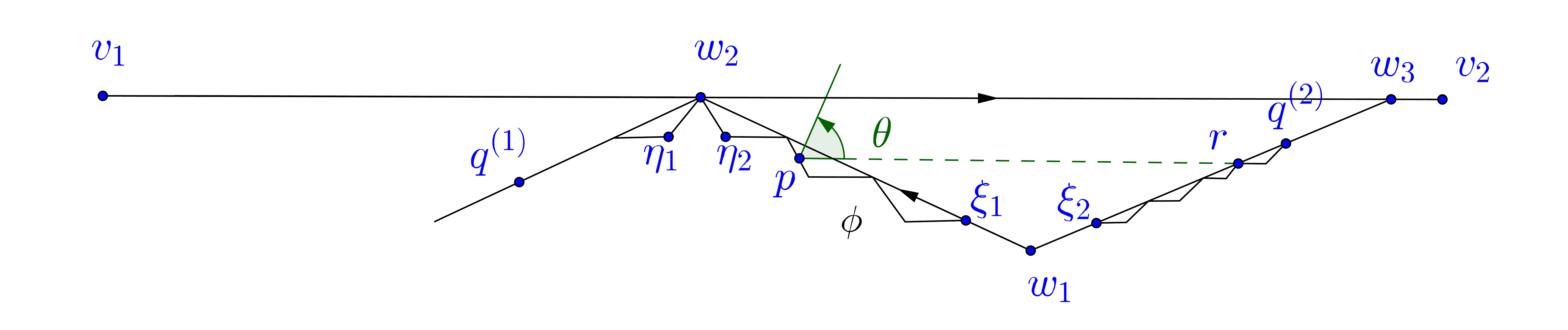}
        \caption{Case 4: $w_1$ is the \outseg\ point and $\pair$ is on the right side of $v_1v_2$}
      \label{fig:case3}
       \end{subfigure}
    \caption{The two cases about $\vangle(v_1v_2, w_2w_1) = -\theta/2$. Here $\pair = (w_1, w_2)$ and $p \in \ntree_{\pair}$. }
    \label{fig:right}
  \end{figure}

  \begin{proof}
    As we have discussed above,   case (3) and (4) satisfy the condition
    $\vangle(v_1v_2,  $ $w_2w_1) = - \theta/2$.  Suppose $q$ belongs to
    $\ntree_{\vpair}$.   Consider case (3). See Figure~\ref{fig:case2}. Suppose $w_2\xi_1$ and
    $w_2\xi_2$ are the two empty pieces incident on $w_2$.  If $q$  is in $\ntree_{\vpair}$ and $\vpair \prec
    \pair$ (i.e., $q = q^{(1)}$ in Figure~\ref{fig:case2}),  $\xi_1$  and $q$  are in the same cone of
    $p$.  If $q$ is not $\xi_1$, there is no edge from $p$ to $q$ even in the Yao-step  since $|p\xi_1| <
    |pq|$. If $q$ is  $\xi_1$, $p\xi_1$ would not be accepted by $\xi_1$ in the reverse-Yao step, since
    there is an edge   from $\xi_2$ to
    $\xi_1$ and $|\xi_1\xi_2| < |p\xi_1|$. Then consider that $q$ (i.e., $q = q^{(2)}$ in
    Figure~\ref{fig:case2}) is in  $\ntree_{\vpair}$ and   $\vpair  \succ \pair$.
     According to Observation~\ref{obs:varphi}, we safely assume that $\vpair$ is an
     internal-pair.   Thus,  $\vangle(v_1v_2, pq) \in  [-\theta/2,\theta/2)$. If $\vangle(v_1v_2, pq) \in  [-\theta/2,0]$,
    $pq$ is not a directed edge in Yao-step since $w_1$ and $q$ are in the same cone and $|w_1p| < |pq|$.
    Finally, consider the projection point $\lambda$ ($\lambda$ exists because of the projection) of
    $p$ to pair $\vpair$. $r$ is the apex point
    related to $\lambda$ and on the segment $p\lambda$.  It is not difficult to check that $|pr| <
    |pq|$ since $\theta/2 \leq \pi/2$ and  the maximum length among the  non-\ttype\ pieces  of
    $\vpair$  is at most twice longer the minimum one (according to refinement).  Similar to
    case 1 in Lemma~\ref{lm:left}, these is a degenerated case that $\lambda$ is the end point of an empty  piece.  See
    Figure~\ref{fig:dege2}. In this case, it is not difficult to check  $|p\lambda| < |pq|$.


    Consider case (4). See Figure~\ref{fig:case3}. Suppose $\eta_1$ and $\eta_2$ are the apex points of
    the \stype\ pieces incident on $w_2$.  If $q$ is in $\ntree_{\vpair}$ and  $\vpair \prec
    \pair$ (i.e., $q = q^{(1)}$ in Figure~\ref{fig:case3}), $\eta_1$ and  $q$ are in the same cone of
    $p$. If $q$ is not $\eta_1$, there is no edge from $p$ to $q$ in the Yao-step since
    $|p\eta_1| < |pq|$. If $q$  is $\eta_1$, $p\eta_1$ would not be accepted by $\eta_1$ in the
    reverse-Yao step since there is an edge
    from $\eta_2$ to $\eta_1$ and $|\eta_1\eta_2| < |p\eta_1|$. Then consider $q$ is in $\ntree_{\vpair}$
    and $ \vpair \succ \pair$ (i.e. $q = q^{(2)}$ in Figure~\ref{fig:case3}). Based on
    Observation~\ref{obs:varphi}, we assume that $\vpair$ is an   internal-pair.
    The polar angle of $pq$ should belong to $(0, \theta/2)$. If not, $w_1$ and $q$ are in the same
    cone. Thus, there is no  edge from $p$ to $q$ since $|pq| > |pw_1|$. Finally, consider the
    projection point $r$ of $p$   ($r$ must exist because of the projection) to pair $\vpair$. $|pr|
    < |pq|$ and $pr$ is parallel to $v_1v_2$ since $\theta/2 \leq \pi/2$.
  \end{proof}

  In the next section, we  discuss how to cut  such long range connections. Roughly speaking, under
  the condition of Lemma~\ref{lm:left}, we can cut the long range connection $pq$ through adding two
  auxiliary points close to $q$. Under the condition of Lemma~\ref{lm:right}, we can  cut the long
  range connection $pq$ through adding  two auxiliary points close to $r$.

  Based on Lemma~\ref{lm:left} and~\ref{lm:right}, we have the following corollary.
  \begin{cor}
    \label{cor:nocycle}
    Consider two sibling pairs $\pair$ and $\vpair$ with subtrees $\ntree_{\pair}$ and
    $\ntree_{\vpair}$ respectively. Suppose $p$ belongs to $\ntree_{\pair}$ and $q$ belongs to
    $\ntree_{\vpair}$. If directed edge $\ov{pq}$ is in $\YY{2k+1}(\normp)$, then $\pair \prec \vpair$.
  \end{cor}


\section{The Positions of Auxiliary Points}
\label{sec:aux}

We discuss how to use the auxiliary points to cut the long range connections in the Yao-Yao
graph $\YY{2k+1}(\normp)$. According to Claim~\ref{claim:nolong}, it is sufficient by cutting all
long range connections between siblings.  Denote the set of auxiliary points  by $\auxp$. Let
$\calP_m = \normp \cup \auxp$.

First, we consider a simple example to see how auxiliary points work. Consider three points $u$,
$v$ and $w$. Line $uv$ is horizontal, and $\angle wvu = \angle wuv
=\theta/2$. The point $\xi_1$ and $\xi_2$ are two points on segment $uw$ and $vw$
respectively. $\xi_1\xi_2$ is horizontal.  See Figure~\ref{fig:longtoy}.  Note that the polar angles
of a cone in the Yao-Yao graph belong to a half-open
interval  in the counterclockwise direction. Thus,  $uv$ is  in the $\YY{2k+1}$ graph, which is the shortest path
between $u$ and $v$.  However, we can add an auxiliary point $r$ close to  $ v$  and $\angle rvu
< \theta/2$. Then according to the definition of Yao-Yao graphs, the point $v$ rejects the edge $uv$
in the reverse-Yao step since $rv$ exists in the Yao-step, and point $r$ and $u$ are in
the same cone of $v$ and $|rv| < |vu|$. Then, consider $ur$ and $r \xi_1$. The directed edge $ur$ is not in Yao graph since $\xi_1$ and $r$ are in the
same cone of $u$ and $|\xi_1u| < |ur|$. The directed edge $ru$ is not in  Yao graph since $\xi_1$
and $u$ are in the same cone of $r$ and $|\xi_1r| < |ur|$. Besides, directed edge $\xi_1r$ is not
in the Yao graph  since $r$ and $\xi_2$ are in the same cone of $\xi_1$ and $|\xi_1\xi_2| <
|\xi_1r|$. Finally, directed edge $r\xi_1$  is not accepted by $\xi_1$ in the reverse-Yao step since there is an edge $\xi_2\xi_1$ in the same cone of $r$ and $|\xi_2\xi_1| < |r\xi_1|$. Overall,
the shortest path between $uv$ becomes $u\xi_1\xi_2rv$.

\begin{figure}[t]
  \centering
  \includegraphics[width = 0.6\textwidth]{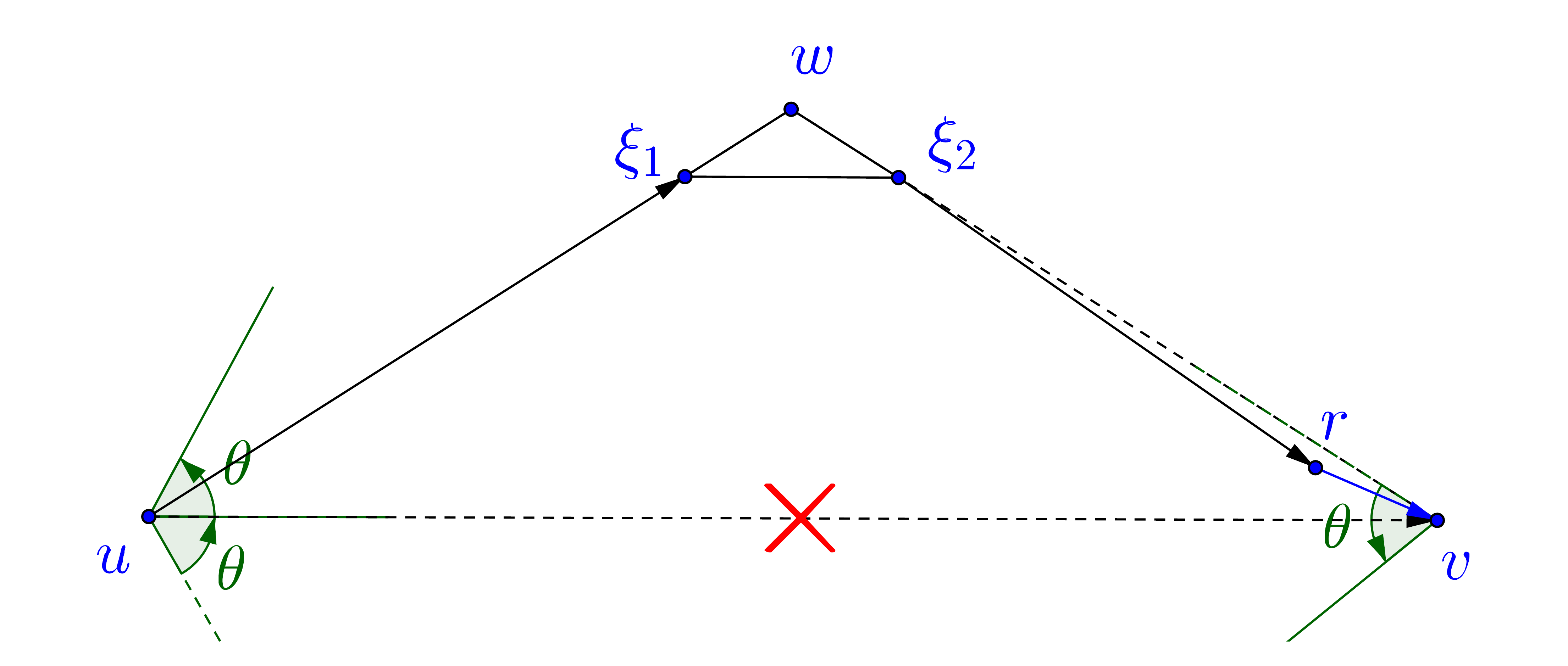}
  \caption{A simple example to explain how an auxiliary point cuts a long range connection.}
  \label{fig:longtoy}
\end{figure}

\topic{The positions of the auxiliary points}  
Inspired by the example in Figure~\ref{fig:longtoy}, we call the normal point closest to an auxiliary point the \emph{center} of the auxiliary point. Then, we find \emph{candidate centers} to add auxiliary points.

\begin{figure}[t]
\captionsetup[subfigure]{justification=centering}
  \centering
  \begin{subfigure}{0.5\textwidth}
    \includegraphics[width = 1\textwidth]{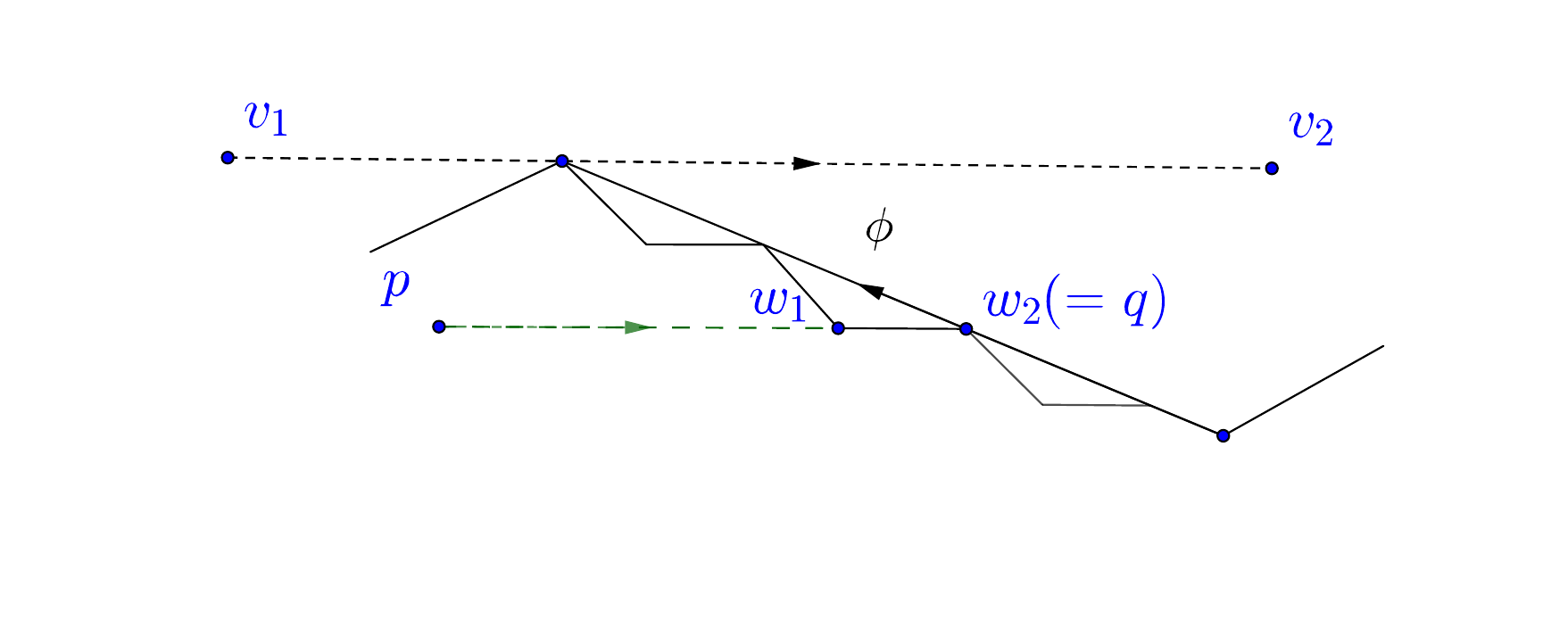}
    \caption{Case 1: the projected point $q$ is $w_2$.}
    \label{fig:center1}
  \end{subfigure}%
  \begin{subfigure}{0.5\textwidth}
    \includegraphics[width = 1\textwidth]{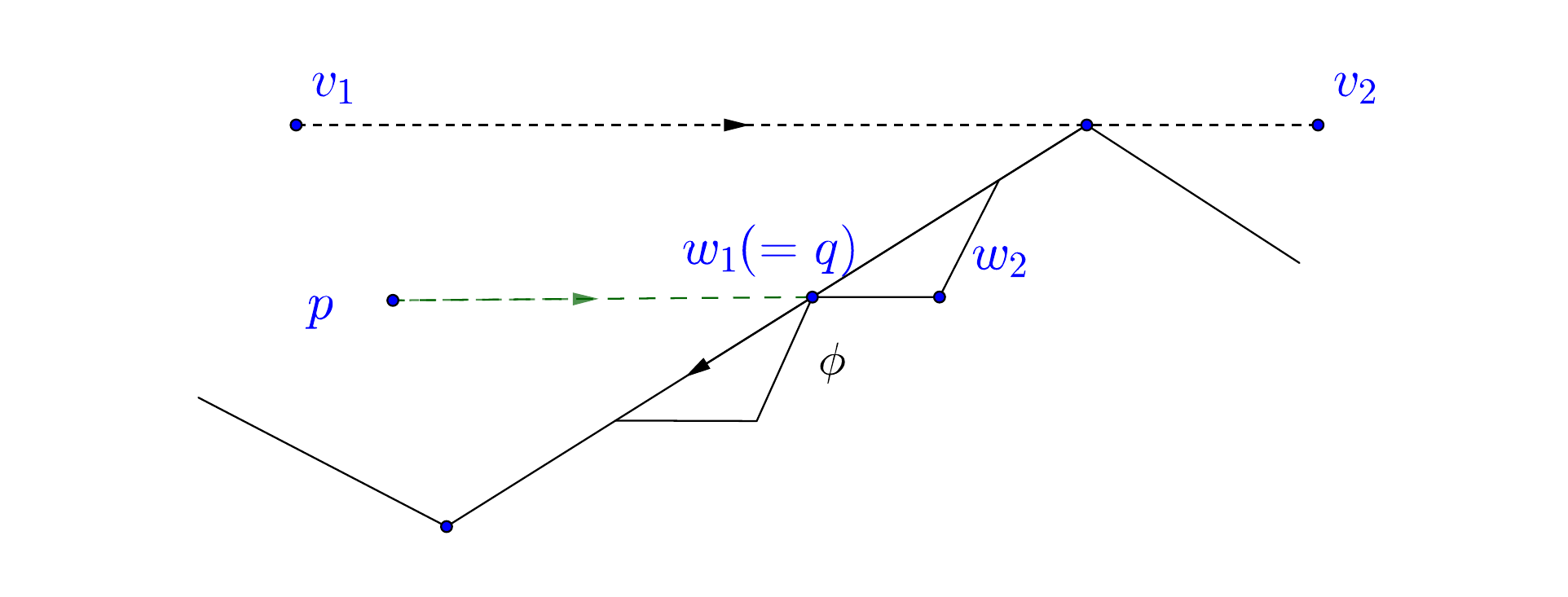}
    \caption{Case 2: the projected point $q$ is $w_1$.}
    \label{fig:center2}
  \end{subfigure}
\caption{The illustration for candidate center of $p$.}
\label{fig:center}
\end{figure}

\begin{lemma1}[Candidate center]
\label{lm:projexist}
Given a pair $(v_1,v_2)$ with its child-pair set $\sset$, consider two sibling pairs  $\vpair, \pair \in \sset$ and $\vpair \prec \pair$. Suppose $p$ is a point in $\ntree_{\vpair}$ and its projected point (denoted by $q$) on the segment of $\pair$ along direction $\ov{v_1v_2}$.  Then, there exists a nonempty subset   $\mathcal{S} \subseteq \gadget_{\pair}$  such that for any $u\in \mathcal{S}$, $pu$ is parallel to $v_1v_2$. See Figure~\ref{fig:center} for an illustration.

We call the point $u := \arg \min_{u \in \mathcal{S} }|pu| $ a \emph{candidate center} of $\pair$. Note that the candidate center may not be the projected point $q$.
\end{lemma1}
\begin{proof}
The correctness  directly results from the projection process.  In the first case (see Figure~\ref{fig:center1}) where the apex points of $\gadget_{\pair}$ and $p$ are in the same side of segment of $\pair$, we will generate a pair $(w_1, w_2)$ such that $q = w_2$ and  $pw_1$ and $pw_2$ are parallel to $v_1v_2$. Thus, $\mathcal{S} = \{w_1, w_2\}$ and we call the point $w_1$ a candidate center of $\pair$. Note that $w_1$ is not the projected point of $p$. In the second case (see Figure~\ref{fig:center2}) where  the apex points of $\gadget_{\pair}$ and $p$ are on the different sides of segment of $\pair$, we will generate a pair $(w_1, w_2)$ such that $q = w_1$ and  $pw_1$ and $pw_2$ are parallel to $v_1v_2$. Thus, $\mathcal{S} = \{w_1, w_2\}$ and we call the point $w_1$ a candidate center of $\pair$.  Beside, $q$ may be an isolated partition point or a point in $\pair$, then $\mathcal{S}$ and the candidate center is  point $q$ itself.
\end{proof}

\begin{defn}[Candidate center set of $\pair$]
Consider a pair $\pair$ with parent pair $(v_1, v_2)$. Let $\Phi$ be the set of child-pairs of $(v_1, v_2)$. In the projection process,
we project all points $p \in \bigcup_{\vpair \prec \pair, \vpair \in \Phi} \ntree_{\vpair}$ to the segment of $\pair$ along the direction $v_1v_2$. Each such point $p$ whose projected point falls inside the segment of $\pair$ corresponds to a candidate center of $\pair$ defined in Lemma~\ref{lm:projexist}. We call the set consisting of all these candidate centers the \emph{candidate center set} of $\pair$. Note that candidate center set is a subset of $\gadget_{\pair}$.
\end{defn}

\eat{

\begin{defn}[Candidate center set of $\pair$]
Consider a pair $\pair$ with parent pair $(v_1, v_2)$. In the projection process,
we project all points $p \in \bigcup_{\vpair \prec \pair} \ntree_{\vpair}$ to the segment of $\pair$ along the direction $v_1v_2$. Each such point $p$ whose projected point falls inside the segment of $\pair$ corresponds to a candidate center of $\pair$ defined in Lemma~\ref{lm:projexist}. We call the set consisting of all these candidate centers the \emph{candidate center set} of $\pair$. Note that candidate center set is a subset of $\gadget_{\pair}$.
\end{defn}

\begin{defn}[Candidate center]
\label{defn:qc}
Consider a pair $(v_1,v_2)$ and  the set $\sset$ of its child-pairs. Suppose $\pair, \vpair \in \sset$ and $\vpair \prec \pair$. $p$  is a point in $\ntree_{\vpair}$. Suppose one projection point of $p$
is in  $\ntree_{\pair}$. Thus, there exists a point $q$ closest to $p$ such that $q \in
\ntree_{\pair}$ and $pq$ is parallel to $v_1v_2$.   We call such a  point $q$ a \emph{candidate center} of auxiliary points.
\end{defn}

Note that if there is a point $p \in \ntree_{\vpair}$  which has a projection point in $\ntree_{\pair}$, then $\vpair \prec \pair$
according to Corollary~\ref{cor:nocycle}. Moreover, according to  Lemma~\ref{lm:left} and~\ref{lm:right},  each candidate center in $\ntree_{\pair}$ is a point in gadget $\gadget_{\pair}$.
Thus, 
}

We add some
auxiliary points centered on these candidate centers to break long range connection. 
For convenience, we define some parameters first.  Let $\Delta$ be the minimum distance between any two normal points and $n$ be the number of the normal points. Recall that we partition the root pair $\mu_1, \mu_2$ into
$d_0$ equidistant pieces.   Let $\gamma$ be a very small angle,
such as $\gamma = \theta d^{-1}_0$.  Let $ \sigma = \max \{ \sin (\theta/2 -\gamma) / \sin \gamma, \sin^{-1}(\theta/2 - \gamma) \}+ \epsilon$ for some small $\e > 0$. Let $\chi = d_0\sigma^{n} \Delta^{-1}$. Roughly speaking, $ \chi \gg d_0 >  \sigma >   1$.

We traverse $\ntree$ in the DFS preorder. Each time we reach a pair
$\pair$, we find all candidate centers  in $\gadget_{\pair}$ and add auxiliary points centered on
them.\footnote{Note that the candidate centers belong to $\gadget_{\pair}$, may not belong to $\pair$ itself. Besides, here we do not need to distinguish whether the candidate center related to a long range connection or not. It may reduce the number of auxiliary points but do not influent the correctness.
} Moreover, let the order of  $\pair$ in the
DFS preorder w.r.t. $\ntree$ be $\kappa$.  The distance between the auxiliary point and its center $q$ depends on $\kappa$.  We use the polar coordinate to
describe the relative location of an auxiliary point to its center.

\eat{
   After processing all points $q$, combining with Lemma~\ref{lm:left} and~\ref{lm:right}
and Claim~\ref{claim:nolong},  there  is no long range connections in $ \calP_{m}$.
}

\begin{figure}[t]
\captionsetup[subfigure]{justification=centering}
  \centering
  \begin{subfigure}{0.5\textwidth}
    \includegraphics[width = 1\textwidth]{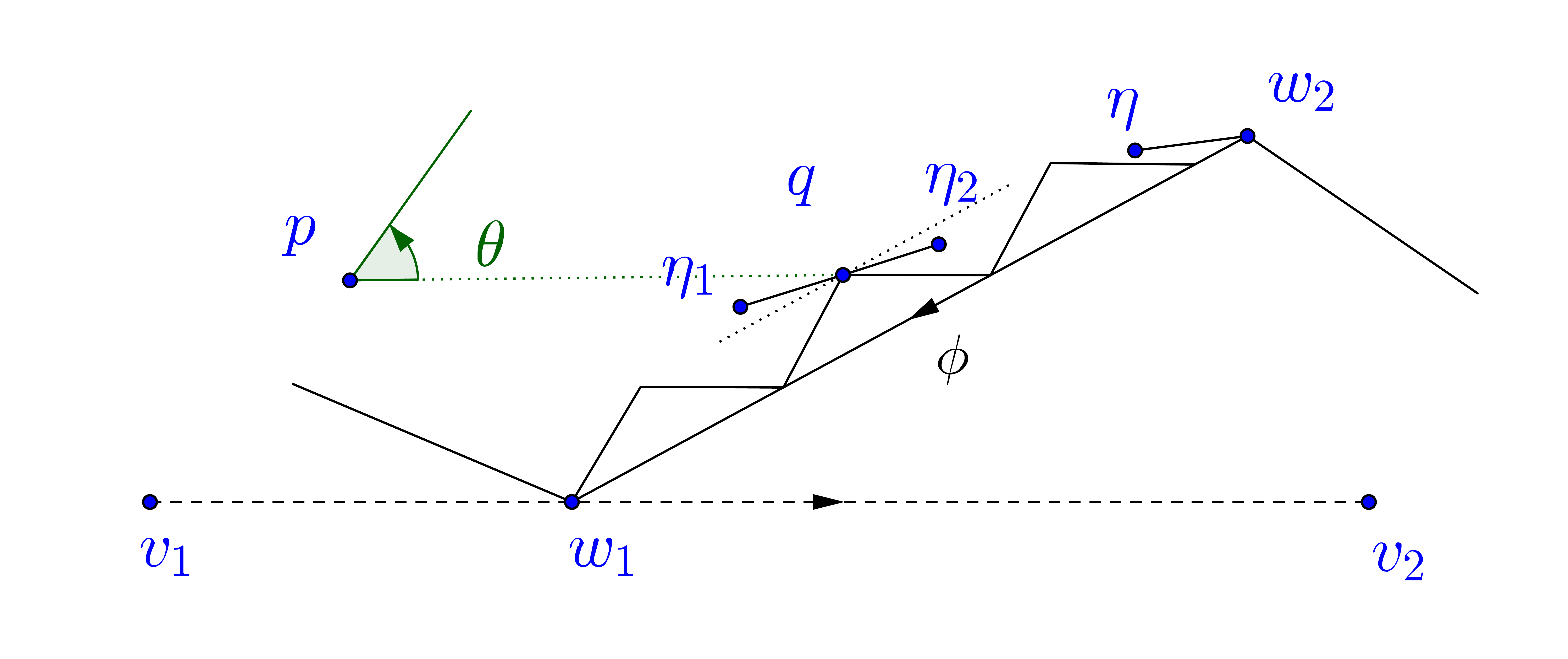}
    \caption{Case 1a: $w_1$ is the \onseg\ point and $\pair$ is on the  left side of
      $v_1v_2$}
    \label{fig:gcase2}
  \end{subfigure}%
  \begin{subfigure}{0.5\textwidth}
    \includegraphics[width = 1\textwidth]{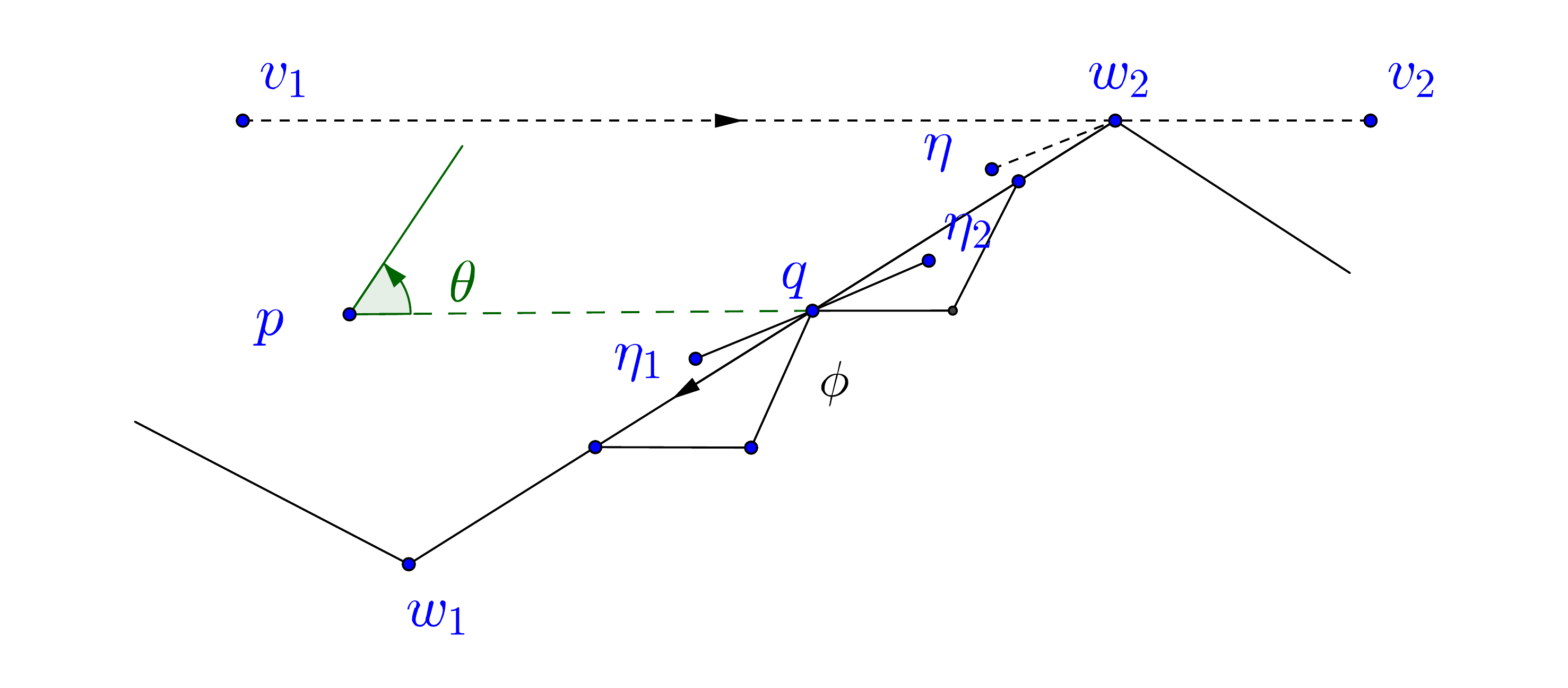}
    \caption{Case 1b: $w_1$ is the \outseg\ point and $\pair$ is on the right side of $v_1v_2$}
    \label{fig:gcase3}
  \end{subfigure}

  \begin{subfigure}{0.5\textwidth}
    \includegraphics[width = 1\textwidth]{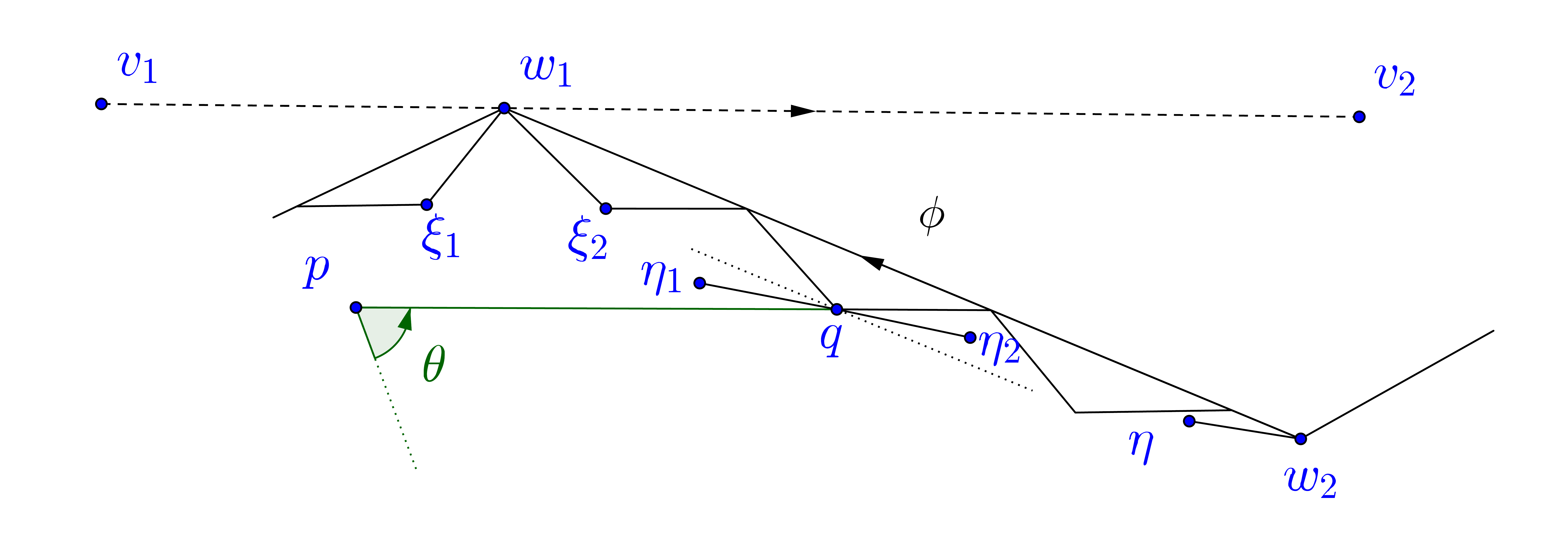}
    \caption{Case 2a: $w_1$ is the \onseg\  point and $\pair$ is on the  right side of
      $v_1v_2$}
        \label{fig:gcase1}
  \end{subfigure}%
  \begin{subfigure}{0.5\textwidth}
    \includegraphics[width = 1\textwidth]{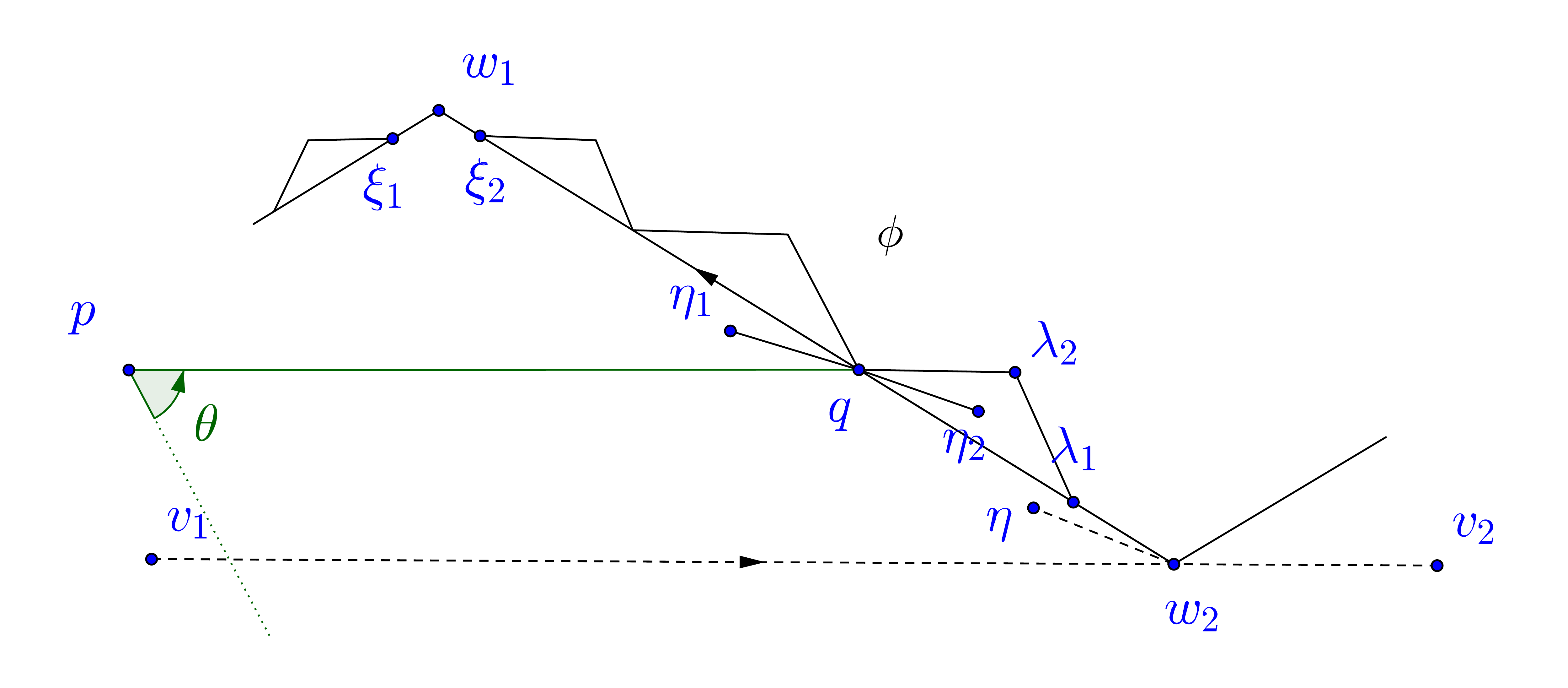}
    \caption{Case 2b: $w_1$ is the \outseg\ point  and $\pair$ is on the left side of $v_1v_2$}
    \label{fig:gcase4}
  \end{subfigure}
  \caption{The auxiliary points for each point. Here $\pair = (w_2, w_1)$ and $q \in \gadget_{\pair}$. $\eta_1$ and $\eta_2$ are two auxiliary points centered on $q$. Note that $|\eta_1 q|$ and $ |\eta_2 q|$ are very small in fact. This is just a diagram to explain the relative positions between $\{\eta_1, \eta_2\}$ and $q$.
   }
  \label{fig:auxp}
\end{figure}

Let $\pair = (w_2, w_1)$ and $(v_1, v_2)$ be the parent-pair of $\pair$. There are two cases according to  $\angle (v_1v_2, w_1w_2) =
\theta/2$ or $-\theta/2$.
\begin{itemize}
\item  $\angle (v_1v_2, w_1w_2) =  \theta/2$ (see Figure~\ref{fig:gcase2}
  and~\ref{fig:gcase3}):
  \begin{itemize}
    \item If $q = w_1$, do not add auxiliary point.
    \item If $q = w_2$, we add the  point $\eta$  such that $\angle  (w_2w_1, w_2 \eta)  = -\gamma$ and $|w_2\eta| =
   \sigma^{\kappa} \chi^{-1}$. 
    \item   Otherwise, we add two points $\eta_1$ and $\eta_2$ centered on $q$ such that   $\angle
  (w_2w_1, q \eta_1) = \angle ( w_2w_1, \eta_2q) = -\gamma$ and $|q\eta_1| =
  |\eta_2q | = \sigma^{\kappa} \chi^{-1}$.
  \end{itemize}

\item $\angle (v_1v_2, w_1w_2) = - \theta/2$ (see Figure~\ref{fig:gcase1} and~\ref{fig:gcase4}): 
\begin{itemize}
  \item If $q = w_1$, do not add auxiliary point.
  \item If $q = w_2$,  we add the point $\eta$  such that $\angle  (w_2w_1, w_2 \eta)  = \gamma$ and $|w_2\eta| =
   \sigma^{\kappa} \chi^{-1}$. 
  \item If $p$ and $q$ are in the same hinge set (i.e., $p, q$ are the points $\xi_1, \xi_2$ in Figure~\ref{fig:gcase1} or~\ref{fig:gcase4}),  we add two points $\eta_1$ and $\eta_2$ centered on $q$ such that $\angle (w_2w_1, q \eta_1) = \angle (w_2w_1, \eta_2 q) = \gamma$  and $|q \eta_1| = |\eta_2 q| = \sigma^{\kappa} \chi^{-1} + \e_{0}$ where $\e_{0}$ is much less than the distance between any two points in $\calP_{m}$.\footnote{ It is slightly different from the first case. We add two auxiliary points with distance slightly larger than $\sigma^{\kappa} \chi^{-1}$ to its center when $p $ and $q$ are in the same hinge set. The reason is that the cone is half-open half-close in the counterclockwise direction.  It will help a lot to unify the proof in the same framework. See the details in the proof of Lemma~\ref{lm:aux2norm}. }
  \item Otherwise, we add two points $\eta_1$ and $\eta_2$ centered on $q$ such that $\angle (w_2w_1, q \eta_1) = \angle (w_2w_1, \eta_2 q) = \gamma$ and $|q \eta_1| = |\eta_2 q| = \sigma^{\kappa} \chi^{-1}$.
\end{itemize}

\end{itemize}

First, we list some useful properties of the auxiliary points below.

\begin{property}
    Properties of auxiliary points:
  \label{prop:aux}
  \begin{itemize}
  \item [P1] The maximum length between an auxiliary point and its center is at most $ d^{-1}_0 \Delta $.
  \item [P2] Any point $q \in \normp$ can become a center for auxiliary points at most twice. Here, for each time that we indeed add some auxiliary points for a candidate center $p$, we say that $p$ becomes a center once.
 \item [P3] There are at most three auxiliary points centered on a normal point. 
  \item [P4] Suppose $q$ is a candidate center because of the projection of $p$ and we add the auxiliary  point $\eta$ centered on $q$.  If there is an auxiliary point $\xi$ centered on $ p$, then $|\xi p | \leq   \sigma^{-1} |\eta q|$. Hence, the perpendicular distance from   $\eta$ to the line $pq$ is  larger than $|\xi p|$.
 \item [P5] If auxiliary points $\eta_1, \eta_2$ and
   $\eta_3$ are centered on $q$ and $|q \eta_1| = |q \eta_2|$,  then $|q \eta_1| \leq \sigma^{-1} |q \eta_3|$,
   $\angle \eta_2 q \eta_3 = (\theta/2 - 2\gamma)$, and $\angle q \eta_3 \eta_2 < \gamma$.
  \end{itemize}
\end{property}
\begin{figure}[t]
  \centering
  \includegraphics[width = 0.6\textwidth]{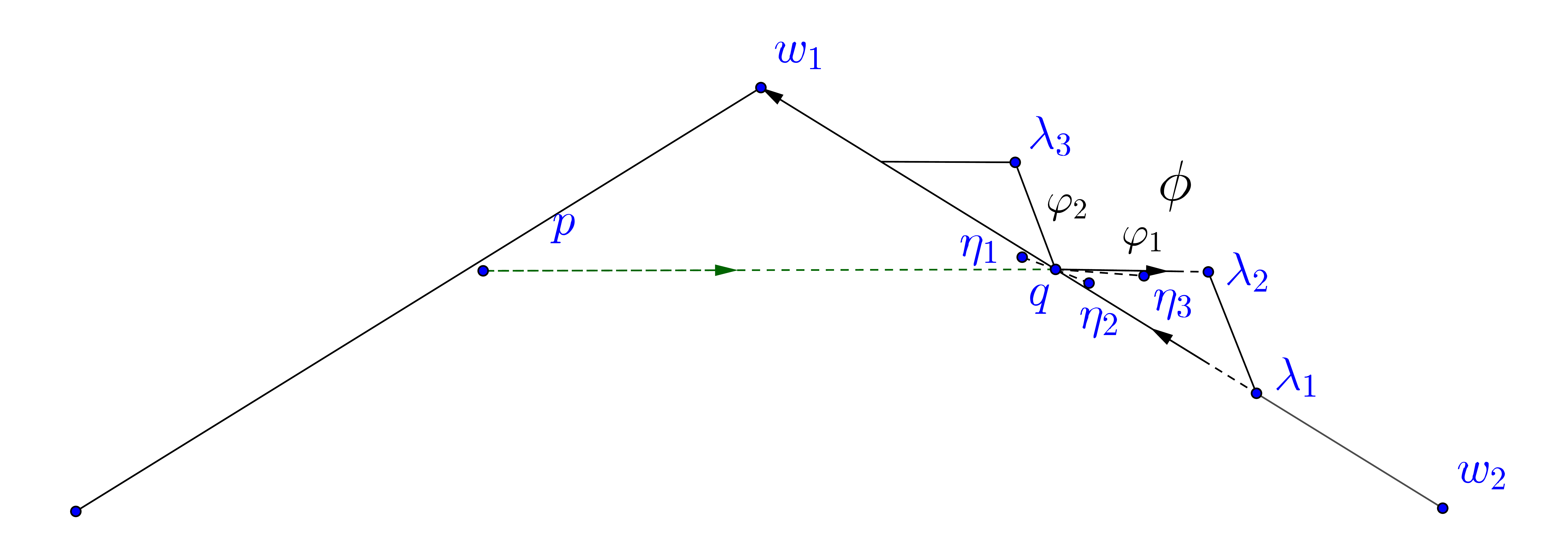}
  \caption{The positions of auxiliary points ($\eta_1, \eta_2, \eta_3$) centered on normal point $q$.}
  \label{fig:three}
\end{figure}

\begin{figure}[t]
  \centering
  \includegraphics[width = 0.6\textwidth]{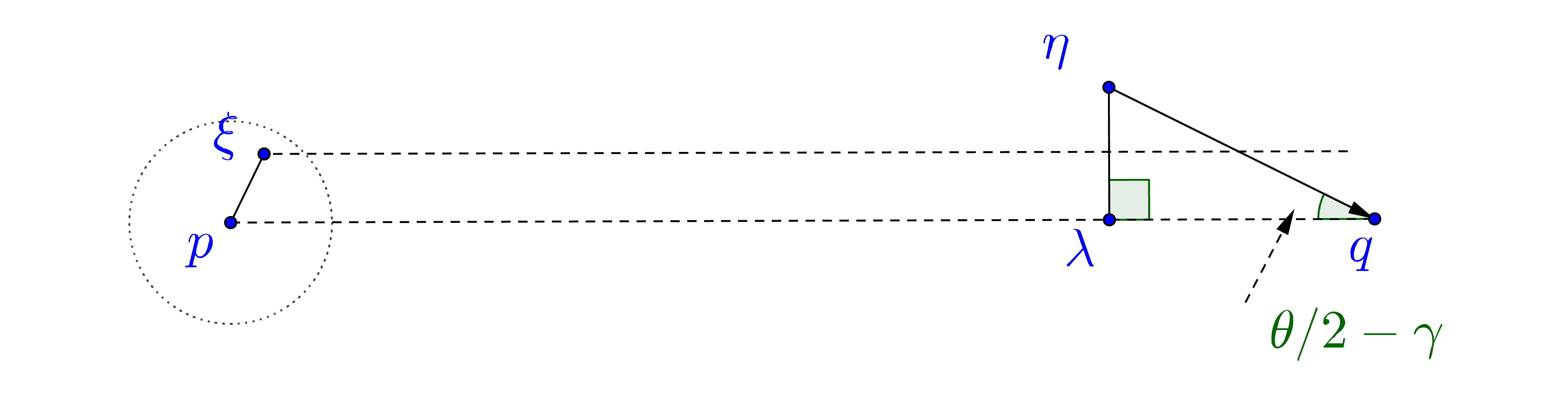}
  \caption{ $|\xi p | \leq \sigma^{-1} |\eta q|$ for the auxiliary point of $p$. Moreover, the perpendicular distance from $\eta$ to $pq$ (i.e., $|\eta \lambda|$ in the figure) is larger than $|\xi p|$. }
  \label{fig:vertical}
\end{figure}

\begin{proof}

  [P1] Note that the largest  $\kappa$  is at  most $n$ since there are at most $n$ pairs in the
  tree. The maximum length between the auxiliary   point and its center is at most $\sigma^{n}  \chi^{-1}  =  d^{-1}_0 \Delta$.

 [P2] Note that each point $q \in \normp$ belongs to at most three gadgets, one pair $\pair$ such that $q \in \gadget_{\pair}$ and two sibling pairs $\vpair_1$ and $\vpair_2$ such that $q \in \vpair_1 \cap \vpair_2$.
 See Figure~\ref{fig:three} for an example.  We visit $\pair$ first and then $\vpair_1$ and $\vpair_2$ in order. Note that $q$ is the shared point of $\vpair_1$ and $\vpair_2$. According to the way to add auxiliary point,  when we visit $\vpair_2$, $q$ corresponds to the point ``$w_1$'' (in Figure~\ref{fig:auxp}) in the rules. Thus,  we do not add auxiliary points for $q$. Hence,  there are only two times that $q$ can become a center for auxiliary points. The first time happens when we visit $\pair$ and the second time happens when we visit $\vpair_1$.

 [P3]  Followed by the proof of [P2],  in the first time, we add two auxiliary points $\eta_1$ and $\eta_2$ centered on
 $q$. In the second time, we  add one auxiliary point $\eta_3$ centered on $q$. Thus, there are at most three auxiliary points centered on a normal point.

 [P4] Suppose $p$ belongs to  subtree $\ntree_{\vpair}$  and $q$ belongs to subtree
  $\ntree_{\pair}$ and $\vpair \prec \pair$. Thus, the auxiliary points are  added for $p$  earlier
  than $q$. It means that $|\xi p | \leq  \sigma^{-1} |\eta q|$. Note that the acute angle between $\eta q$
  and $pq$ is $(\theta/2 - \gamma)$ and  $\sigma >  \sin^{-1}(\theta/2 - \gamma)$.  Thus, the
  perpendicular distance from $\eta$ to the line $pq$ is  larger than $|\xi p|$.  See Figure~\ref{fig:vertical}.

[P5]   According to the proof of [P3] (see Figure~\ref{fig:three}) we add $\eta_1$ and $\eta_2$ earlier than $\eta_3$.  According to the construction, we can add these three auxiliary points for $q$. Checking the four cases in construction,  we can get $\angle(pq, q\eta_3) = -\gamma$ and $\angle(pq, q\eta_2) =  -\theta/2 + \gamma$.  Thus, $\angle \eta_2 q \eta_3 = \theta/2 - 2\gamma$.  Moreover, note that $\sigma >  \sin (\theta/2 -\gamma) / \sin \gamma $ and $|\eta_2 q| < \sigma^{-1}|\eta_3 q|$. According to the law of sines, we get $\angle q \eta_3 \eta_2 < \gamma$.
  \end{proof}

\topic{Extended hinge set} We extend the concept of hinge sets  to the \emph{extended
  hinge set} to include auxiliary points.   The extended
  hinge set consists of the
normal points in the  hinge set and the auxiliary points centered on these normal points.
 Besides, if $p$ belongs to
$\ntree_{\pair}$, then the auxiliary points centered on $p$ belong to \emph{extended} $\ntree_{\pair}$.
Then Claim~\ref{claim:nolong} is still true for $\YY{2k+1}(\calP_m)$ with the same proof.
It means that we only need to consider the long range connections between the descendants of any two
sibling pairs.

Moreover,  we can get similar  properties
as  Lemma~\ref{lm:left} and~\ref{lm:right}  for the auxiliary points.
 Suppose $\pair$ and $\vpair$
are two sibling pairs. If  $p \in \ntree_{\pair}$ and $q \in \ntree_{\vpair}$ and there is  a long
range connection $\ov{pq} $  in $\YY{2k+1}$, then  $\pair \prec \vpair$.  Meanwhile, the points
in $\ntree_{\vpair}$ locate in two cones of $p$. But only one of the two cones may contain a long range connection. We describe the property formally as follows.

\begin{lemma1}
  \label{lm:aux2norm}
Given a pair $(v_1, v_2)$ at $\level{l}$ for $l < m-1$, with child-pair set $\Phi$, consider two
sibling pairs $\pair$ and $\vpair$ in $\Phi$ where $\pair = (w_1, w_2)$.   $p $ is a point in
extended $\ntree_{\pair}$ and $q$ is a point in extended $\ntree_{\pair}$.  Suppose there is a
directed edge $\ov{pq}$ in $\YY{2k+1}(\calP_{m})$.
\begin{itemize}
\item If $\angle (v_1v_2, w_2w_1) = \theta / 2$ , then $\angle(v_1v_2, pq) \in (-\theta, 0]$.
\item If $\angle (v_1v_2, w_2w_1) = -\theta / 2$ , then $\angle(v_1v_2, pq) \in (0, \theta]$.
\end{itemize}
\end{lemma1}

\begin{figure}[t]
\captionsetup[subfigure]{justification=centering}
  \centering
  \begin{subfigure}{0.7\textwidth}
    \includegraphics[width = 1\textwidth]{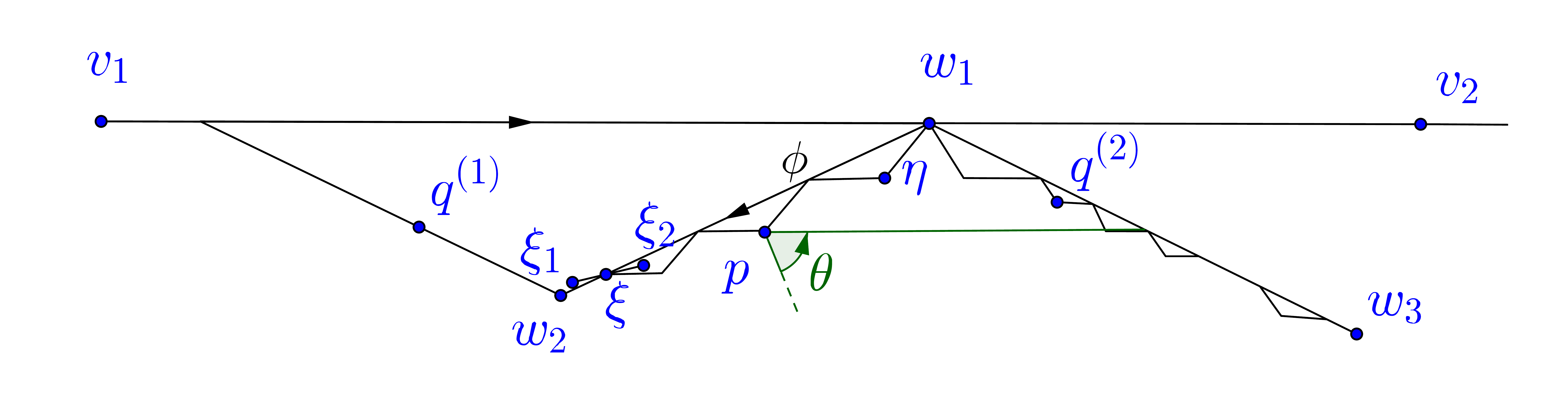}
    \caption{Case 1: $w_1$ is the partition point and $\pair$ is on the  right side of
      $v_1v_2$}
    \label{fig:auxcase1}
  \end{subfigure}

  \begin{subfigure}{0.7\textwidth}
        \includegraphics[width = 1\textwidth]{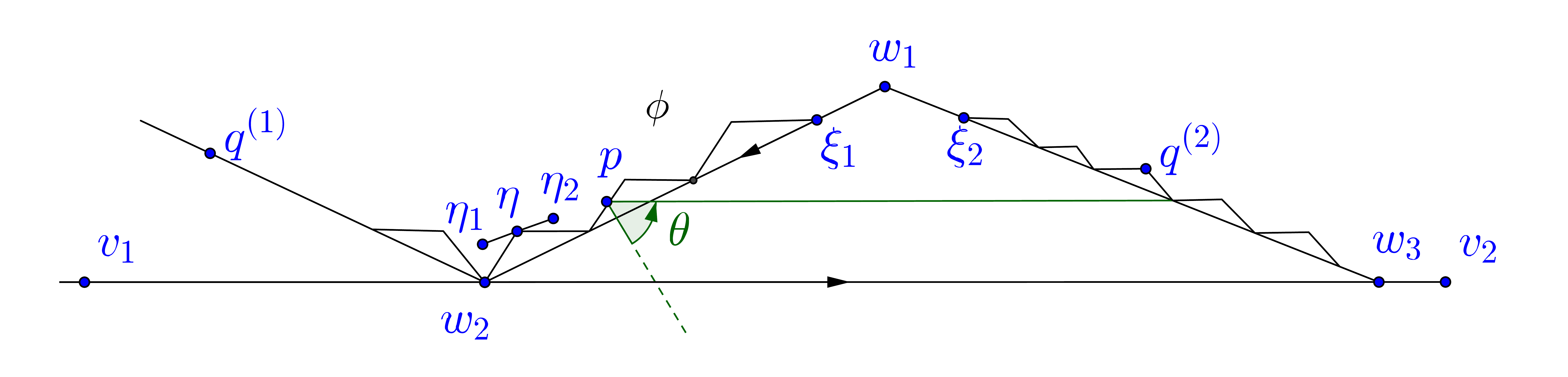}
        \caption{Case 2: $w_1$ is the  apex point  and $\pair$ is on the left side of $v_1v_2$}
    \label{fig:auxcase2}
  \end{subfigure}

  \begin{subfigure}{0.7\textwidth}
    \includegraphics[width = 1\textwidth]{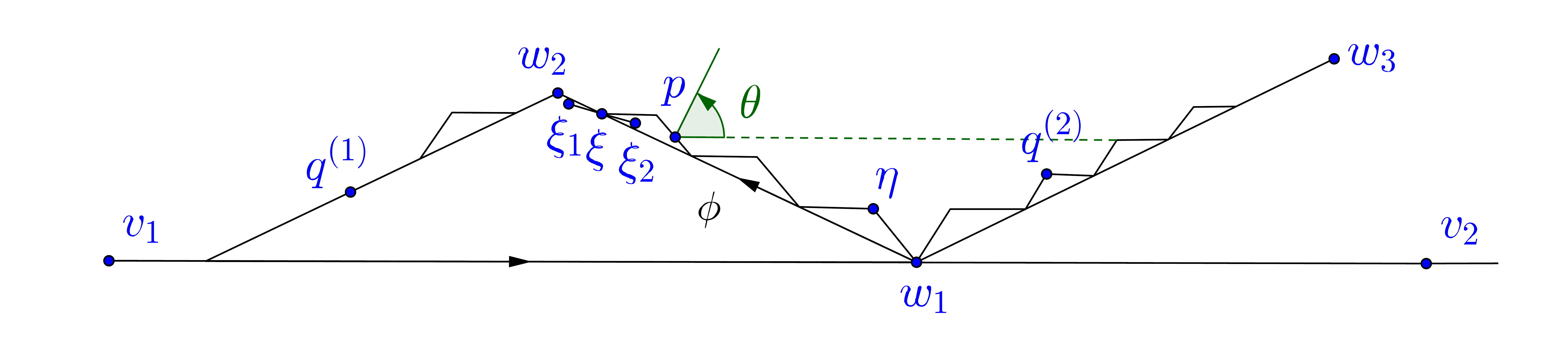}
    \caption{Case 3: $w_1$ is the partition point and $\pair$ is on the  left side of
      $v_1v_2$}
    \label{fig:auxcase3}
  \end{subfigure}

  \begin{subfigure}{0.7\textwidth}
    \includegraphics[width = 1\textwidth]{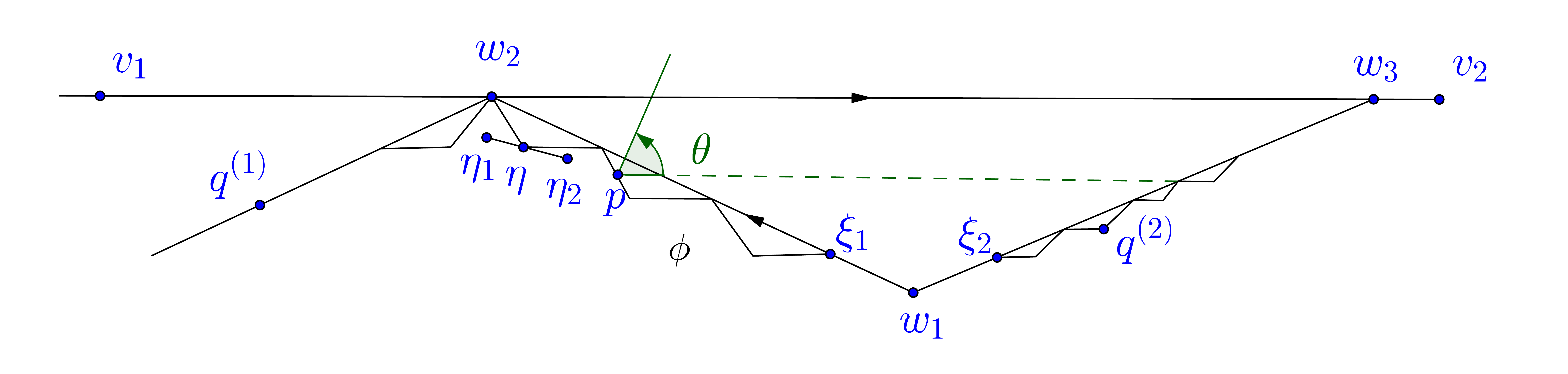}
    \caption{Case 4: $w_1$ is the apex point and $\pair$ is on the right side of $v_1v_2$}
    \label{fig:auxcase4}
  \end{subfigure}
  \caption{Here $\pair = (w_1, w_2)$ and $p $ belongs to the extended $\ntree_{\pair}$ which includes auxiliary points. }
  \label{fig:auxp2}
\end{figure}

\begin{proof}

The proof follows the same procedure as the proof of Lemma~\ref{lm:left} and~\ref{lm:right}. We also distinguish into two cases.
Given a pair $(v_1, v_2)$ and its child-pair set $\Phi$, consider two sibling pairs $\phi$ and $\vpair$ in $\Phi$ where $\phi = (w_1, w_2)$. The first case is that $\angle ( v_1 v_2, w_2w_1) = \theta / 2$. The second case is that $\angle ( v_1 v_2, w_2w_1) = - \theta / 2$. Consider a point $p$ in extended $\ntree_{\pair}$. $p$ can be a normal point or an auxiliary point.

Consider that $\angle ( v_1 v_2, w_2w_1) = \theta / 2$. First, suppose $w_1$ is the partition point and $\pair$ is on the right side of $\ov{v_1v_2}$. See Figure~\ref{fig:auxcase1}.  Suppose $q$ belongs to $\ntree_{\vpair}$  and $\vpair \prec \pair$ (i.e., $q = q^{(1)} $ in Figure~\ref{fig:auxcase1}). Denote the partition point in $\calA_{\pair}$ closest to $w_2$ by $\xi$. Because of the projection of points in $\ntree_{\vpair}$, $\xi$ has two auxiliary points, denoted by $\xi_1$ and $\xi_2$. Thus, $\ov{pq}$ is not an edge in the Yao-step since $\xi_1$ and $q$ are in the same cone of $p$. Then consider that $q$ belongs to $\ntree_{\vpair}$ and $\vpair \succ  \pair$ (i.e., $q = q^{(2)} $ in
  Figure~\ref{fig:auxcase1}).   Denote the point in $\calB_{\pair}$ closest to $w_1$ by $\eta$.  According to  the fact that  $w_1\eta$ is the maximum length pair among the child-pairs of $\pair$ (see Corollary~\ref{cor:p1}),   $\eta$ and $q$ are in the same cone of $p$ when $\angle (v_1v_2, pq) > 0$. Thus, there is no long range connection for $p$ in cone $C_{p}(0, \theta]$.

Second, suppose $w_1$ is the apex point and $\pair$ is on the left side of $\ov{v_1v_2}$. See Figure~\ref{fig:auxcase2}.
Denote the closest point in $\calB_{\pair}$ to $w_2$ by $\eta$. Suppose $q$ belongs to $\ntree_{\vpair}$ and   $\vpair \prec \pair$ (i.e., $q =  q^{(1)}$ in Figure~\ref{fig:auxcase2}). Because of the projection of points in $\ntree_{\vpair}$, $\eta$ has two auxiliary points, denoted by $\eta_1$ and $\eta_2$. There is no edge $\ov{pq}$ in the Yao-step since $\eta_1$ and $q$ are in the same cone $p$ and $|\eta_1 p| < |qp|$. Then consider that  $q$ belongs to  $\ntree_{\vpair}$ and $\vpair \succ \pair$ (i.e., $q =  q^{(2)}$ in
  Figure~\ref{fig:auxcase2}).  If $q$ in the cone $C_{p}(0, \theta]$, $q$ and $w_1$ are in the same cone of $p$ and $|p w_1| < |p q|$. Thus, in the Yao-step, there is no edge from $p$ to $q$ in the cone $C_{p}(0, \theta]$. Thus,  we prove the first part of the lemma.

  Next, we consider the case $\angle ( v_1 v_2, w_2w_1) = - \theta / 2$.   First, we consider the case in which $w_1$ is the partition point and $\pair$ is on the left side of $\ov{v_1v_2}$. See Figure~\ref{fig:auxcase3}. Denote the partition point in $\calA_{\pair}$ closest to $w_2$ by $\xi$. According to the construction for auxiliary point (case 2a),  we add two auxiliary points $\xi_1$ and $\xi_2$ such that $|\xi \xi_1| = |\xi_2 \xi| = \sigma^{\kappa} \chi^{-1} + \e_{0}$.
 If $q$  is in $\ntree_{\vpair}$ and $\vpair \prec  \pair$ (i.e., $q = q^{(1)}$ in Figure~\ref{fig:auxcase3}), $\xi_1$  and $q$  are in the same cone of $p$.  Because the distance $|\xi\xi_1|$  ($> \sigma^{\kappa} \chi^{-1}$) is slightly longer than the distances from other auxiliary points of $\gadget_{\pair}$ to their centers.
 Thus, $\ov{pq}$ is not an edge in the Yao-step since $|\xi_1 p| < |pq|$ and $\xi_1$ and $q$ are in the same cone of $p$.
Then consider that $q$ (i.e., $q = q^{(2)}$) in
    Figure~\ref{fig:auxcase3}) is in  $\ntree_{\vpair}$ and   $\vpair  \succ \pair$. If $\vangle(v_1v_2, pq) \in  [-\theta/2,0]$,
    $\ov{pq}$ is not a directed edge in Yao-step since $w_1$ and $q$ are in the same cone and $|w_1p| < |pq|$.

Finally, we consider the case that  $w_1$ is the apex point and $\pair$ is on the right side of $\ov{v_1v_2}$. 	 Suppose $\eta$ is the apex point in $\calB_{\pair}$ closest to $w_2$. $\eta_1$ and $\eta_2$ are auxiliary points of $\eta$.   If $q$ is in $\ntree_{\vpair}$ and  $\vpair \prec
    \pair$ (i.e., $q = q^{(1)}$ in Figure~\ref{fig:auxcase4}), $\eta_1$ and  $q$ are in the same cone of
    $p$ according to the construction for auxiliary point (case 2b).  Thus, $\ov{pq}$ is not an edge in the Yao-step since $|p\eta_1|< |pq|$.
    Then consider $q$ is in $\ntree_{\vpair}$
    and $ \vpair \succ \pair$ (i.e. $q = q^{(2)}$ in Figure~\ref{fig:auxcase4}).
    If the polar angle of $pq$ belongs to $(-\theta, 0]$,  $w_1$ and $q$ are in the same
    cone. Thus, there is no  edge from $p$ to $q$ since $|pq| > |pw_1|$. Thus, we prove the second part of the lemma.

Overall,  we have proved the lemma.
\end{proof}

Then, we prove that after adding the auxiliary points, there is no long range connection.
\begin{lemma1}
  \label{lm:aux}
There is no long range connection in  $\YY{2k+1}(\calP_{m})$.
\end{lemma1}

\begin{proof}

Consider a pair $(v_1,v_2)$ and  the set $\sset$ of its child-pairs. Suppose $\pair, \vpair \in \sset$
and $\vpair \prec \pair$.  $p$  is a point in $\ntree_{\vpair}$.  Denote an auxiliary  point centered on $p$, if any,
by $\xi$. Let $u \in \{p,\xi \}$.  There exists a point $q$ closest to $p$ such that $q \in
\ntree_{\pair}$ and $pq$ is parallel to $v_1v_2$ based on the projection process.   If not, i.e., $\ntree_{\pair}$ only locates in one cone of $p$, according to Lemma~\ref{lm:aux2norm}, there is no long range connection between $u$ and points in extended $\ntree_{\pair}$.

According to Lemma~\ref{lm:aux2norm}, first, there is no directed edge from a point in (extended) $\ntree_{\pair}$  to
 (extended) $\ntree_{\vpair}$. Next, we prove there is no long range connection from $\ntree_{\vpair}$ to $\ntree_{\pair}$. Since $p$ is an arbitrary point in $\ntree_{\vpair}$, we prove that there is no long range connection between $u$ and the points in $\ntree_{\pair}$,
 (recall $u \in \{p, \xi\}$). According to whether $q$ is in $ \calA_{\pair} \cup \calB_{\pair}$ or $\pair$, there are two cases.

\begin{figure}[t]
  \centering
    \includegraphics[width = 0.7\textwidth]{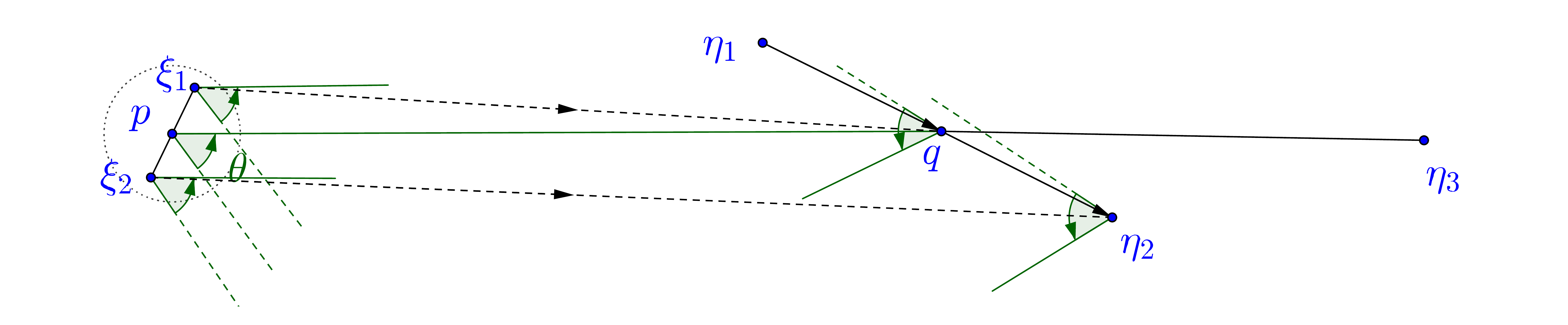}
  \caption{The case that $q \in \calA_{\pair} \cup \calB_{\pair} $. The figure is an enlarged view of Figure~\ref{fig:three}.  $q$ may have three auxiliary points $\{\eta_1, \eta_2, \eta_3\}$. $\xi_1$ and $\xi_2$ are two possible positions of the auxiliary point $\xi$ centered on $p$. }
  \label{fig:twoaux}
\end{figure}

\topic{$q $ belongs to  $ \calA_{\pair} \cup \calB_{\pair}$}   $q$ has two auxiliary
points $\eta_1$ and $\eta_2$ because
of the projection $\ov{pq}$ and $q$ is a candidate center.  Note that $q$ may have a third auxiliary point $\eta_3$. But $p$ and  $\eta_3$ are on the two different
sides of $\eta_1\eta_2$ and $|\eta_3q| > |\eta_1q|= |\eta_2q|$ because of
Property~\ref{prop:aux}[P5]. Therefore, there is no directed edge $p\eta_3$ in the Yao-step because $\eta_2$
and $\eta_3$ are in the same cone of $p$ and $|\eta_2 p| < |\eta_3p|$.
According to Property~\ref{prop:aux}[P4], $|\xi p|$ is much less than $|\eta_1 q|$ or $|\eta_2 q|$
and the perpendicular distance from  $\eta_1$ and $\eta_2$ to the line $pq$ is longer than $|\xi
p|$.  Suppose $u \in \{ p, \xi \}$. See Figure~\ref{fig:twoaux} which is an enlarged view of Figure~\ref{fig:three}, in which $\xi_1$ and $\xi_2$
are two possible positions of $\xi$.
According to Lemma~\ref{lm:aux2norm},
$u \eta_1 $ does not exist in the Yao-step since $\eta_1$ and one point of $\pair$ (denoted by $w_1$, refer to Figure~\ref{fig:three}) are in the same cone of $u$ and $|w_1 u | < |\eta_1 u|$. If there is an edge $u\eta_2$ in
the Yao-step, the edge $u \eta_2$ cannot be accepted by $\eta_2$ in the reverse-Yao step since $q\eta_2$ exists, and point $q$ and
$u$ are in the same cone of $q$ and $|q\eta_2| < |u\eta_2|$. If there is an edge $u q$ in
the Yao-step, the edge $u q$ cannot be accepted by $q$ in the reverse-Yao step since $\eta_1$ and
$u$ are in the same cone and  $|\eta_1q| < |u q|$.  Therefore, there is no long range
connection related to $p$ and its auxiliary points.

\begin{figure}[t]
\captionsetup[subfigure]{justification=centering}
  \centering
  \begin{subfigure}{0.5\textwidth}
    \includegraphics[width = 1\textwidth]{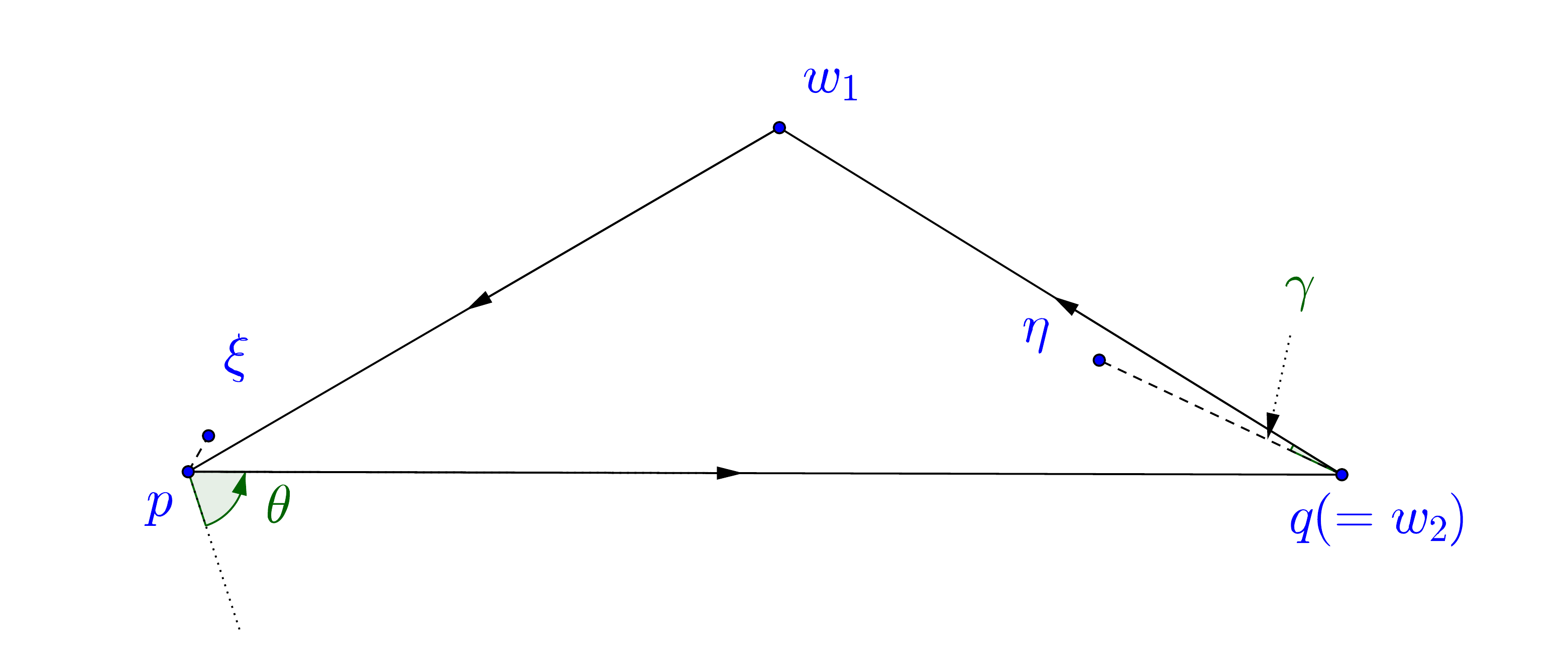}
    \caption{Only one auxiliary point centered on $q$.}
    \label{fig:degepoint1}
  \end{subfigure}%
  \begin{subfigure}{0.5\textwidth}
    \includegraphics[width = 1.1\textwidth]{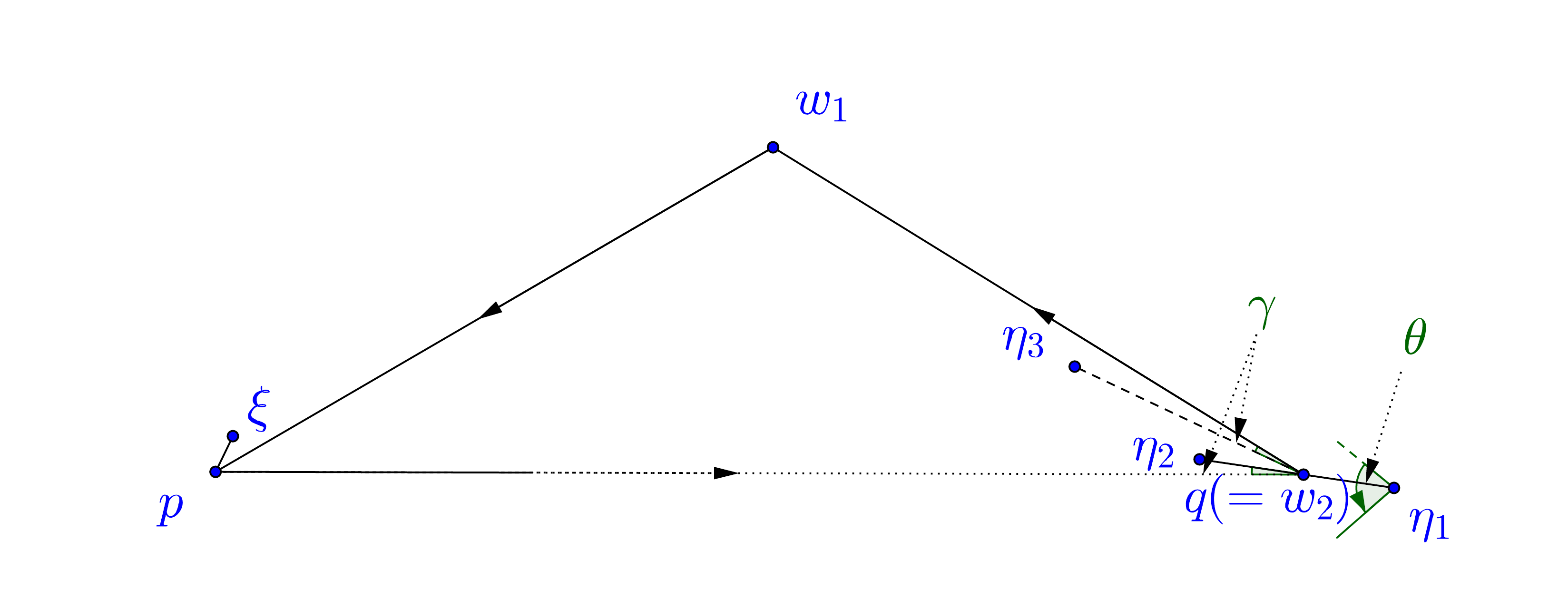}
    \caption{Three auxiliary points centered on $q$.}
    \label{fig:degepoint3}
  \end{subfigure}
  \caption{ $q$ is a point of pair $\pair$, i.e., $q = w_2$ . }
  \label{fig:degepoint}
\end{figure}

\topic{$q$ belongs to $\pair$} See Figure~\ref{fig:degepoint}. Note that in this case,  $q$ is point $w_2$ of $\pair$. According to
Property~\ref{prop:aux}[P3], any point has at most three auxiliary points. Since $q$ is a projection
point of $p$, $q$ has at least one auxiliary point. Thus, there are two possible situations.
One is that there is only one auxiliary point centered on $q$ (see
Figure~\ref{fig:degepoint1}). It means that in the first time that $q$ was able to be a candidate center, there is no auxiliary point added centered on it (see the proof of Property~\ref{prop:aux}[P2]). Denote the auxiliary point of $q$ by $\eta$.
Let $u \in \{p, \xi \}$ where $\xi$ is an auxiliary point centered on $p$. According to the Property~\ref{prop:aux}[P4], $\eta$ and $w_1$ are in the same cone of $u$ and $|uw_1| < |u\eta|$. Therefore, there is no edge $u \eta$ in the Yao-step. Moreover, $uq$ cannot be accepted in the reverse-Yao step. Because the edge $\eta q$ exists in the Yao-step. $u$ and $\eta$ are in the same cone of $q$ and $|\eta q| < |u q|$. Combining with Lemma~\ref{lm:aux2norm}, there is no long range connection from $u$ to $\ntree_{\pair}$.
 The second case  is that there are three auxiliary points of $q$ (see Figure~\ref{fig:degepoint3}). Denote the auxiliary points of $q$ by $\{\eta_1, \eta_2, \eta_3\}$. According to Property~\ref{prop:aux}[P5], we know $\angle \eta_2q\eta_3 = (\theta/2 - 2\gamma)$. Again, denote $u \in \{p, \xi\}$. There is no edge $u \eta_3$ in the Yao-step since $w_1$ and $\eta_3$ are in the same cone of $u$ and $|u w_1|< |u\eta_3|$. There is no edge $u \eta_1$ in the reverse-Yao step, since there is an edge $q\eta_1$ in the Yao-step and $|q\eta_1|< |u\eta_1|$. Similarly, there is no edge $uq$ since there is an edge $\eta_2 q$ in the Yao-step and $|\eta_2 q| < |u q|$. Next, note that $|q\eta_2| \leq \sigma^{-1} |q \eta_3|$. According to Property~\ref{prop:aux}[P5], we know $\eta_3$ and $u$ are in the same cone of $\eta_2$. Thus, there is no edge from $u \eta_2$.

Overall, we prove that there is no long range connection in $\YY{2k+1}(\calP_{m})$.
\end{proof}


\section{The Length Between $\mu_1$ and $\mu_2$ in $\YY{2k+1}(\calP_{m})$}
\label{sec:length}

In this section, we  prove that the length of the shortest path between the  initial points $\mu_1$ and
$\mu_2$ in $\YY{2k+1}(\calP_{m})$  diverges as $m$ approaches infinity.

First, recall that we have extended the concept of hinge sets  to  \emph{extended
  hinge sets}  which consist of the normal points in the  hinge set and the auxiliary points of
these normal points. Consider two extended hinge sets $\hinge{}$ and $\hinge{}'$. Define the shortest path
between $\hinge{}$ and $\hinge{}'$ to be the shortest path in $\YY{2k+1}(\calP_{m})$ between any two points $p$ and $q$ such that $p \in \hinge{}$ and $q \in \hinge{}'$.
Consider any pair $\pair = (w_1, w_2)$ at $\level{(m-1)}$.
We give a lower bound on the shortest path distance between its former extended hinge set and latter extended hinge set.

\begin{lemma1}
  \label{lm:distaux}
Consider any pair $(w_1, w_2)$ at  $\level{(m-1)}$. Denote its former extended hinge set by $\lhinge{\pair}$, and latter extended hinge set by $\rhinge{\pair}$. The shortest path distance between $\lhinge{\pair}$ and $\rhinge{\pair}$ is at least $(1- 6  d^{-1}_0)|w_1w_2|$.
\end{lemma1}
\begin{proof}

   Let $|w_1w_2| = \delta$.  See Figure~\ref{fig:dist}. Note that $\lhinge{\pair}$ and $\rhinge{\pair}$ (the two hinge sets centered on $w_1$ and
   $w_2$) are not overlapping.  Denote the near-empty piece incident on $w_1$ by $w_1\eta_1$ and
   the \ttype\ piece incident on $w_2$ by
  $w_2\xi_1$. $\xi_1\xi_2$ is the  leaf-pair closest to $w_2$. $\xi_2\eta_2$ is perpendicular to
  $w_1w_2$. The shortest  Euclidean distance between the two hinge sets  is no less than
  $|\eta_1\eta_2|$. According to Property~\ref{prop:piece2}, $|w_1\eta_1|
  \leq 2 d^{-1}_0 \delta$, $|w_2\xi_1|
  \leq  d^{-1}_0 \delta$ and $\xi_1\eta_2 \leq 0.5d^{-1}_0 \delta $. Thus,  $|\eta_1\eta_2| >  (1-
  3.5 d^{-1}_0)\delta$.

 Then consider the auxiliary points. Note that according to the
  Property~\ref{prop:aux}[P1], the maximum distance between  an auxiliary point and its center is $ d^{-1}_{0}\Delta$, where $\Delta$ is the minimum distance
  between any two normal points. Since $\Delta \leq \delta$, according to triangle inequality, the
  auxiliary points can reduce the distance between the two hinge sets by at most $2 d^{-1}_{0}\delta$. Overall, the shortest path
  between  $\lhinge{\pair}$ and $\rhinge{\pair}$  is at least $(1- 6  d^{-1}_0)|w_1w_2|$.
\end{proof}

\begin{figure}[t]
    \centering
    \includegraphics[width = 0.55\textwidth]{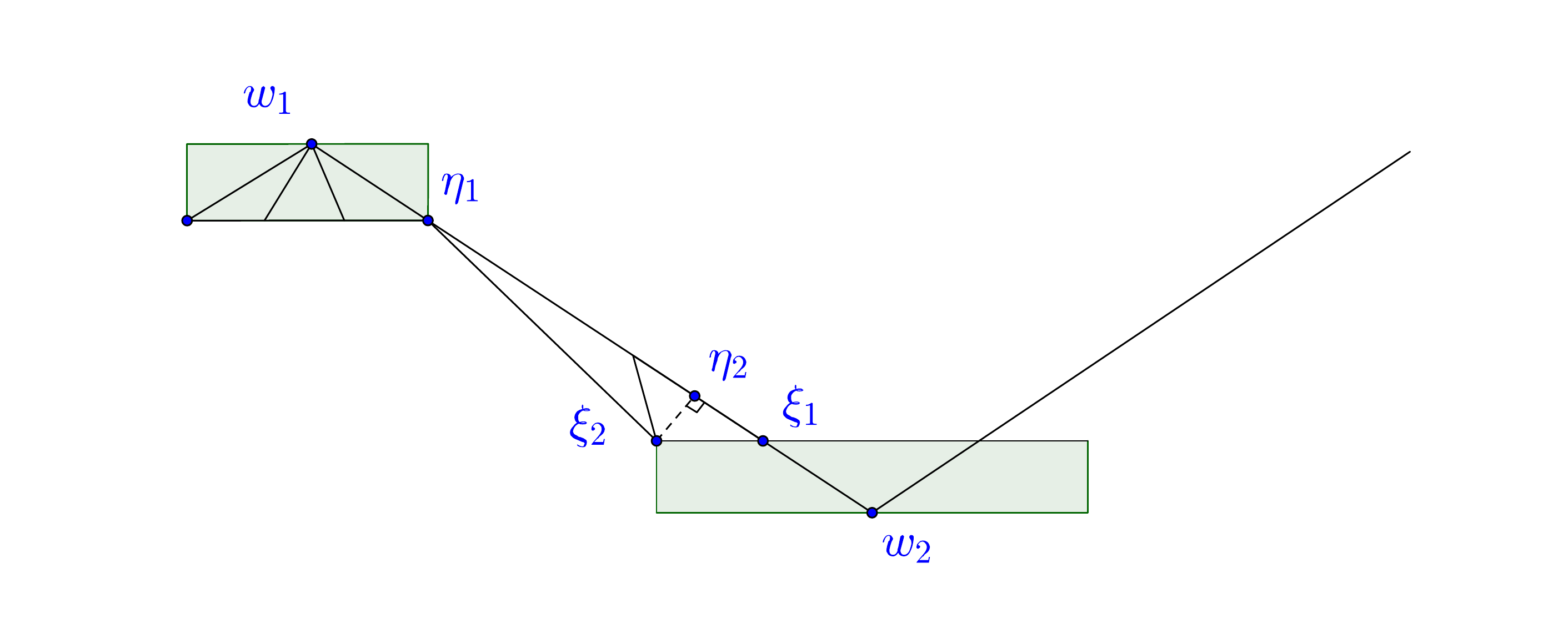}
    \caption{The shortest Euclidean distance between two hinge sets centered on points of a pair $\pair = (w_1, w_2)$ at $\level{(m-1)}.$}
    \label{fig:dist}
\end{figure}

According to Lemma~\ref{lm:aux}, there is no long range connection in $\YY{2k+1}(\calP_{m})$.  Thus,
the shortest path between $\mu_1$ and $\mu_2$ should pass through all hinge sets in order. Thus,   for each pair $\pair$ at
$\level{(m-1)}$, there is a path between $\lhinge{\pair}$ and $\rhinge{\pair}$.

Let the shortest path between $\lhinge{\pair}$ and $\rhinge{\pair}$ be $\Delta_{\pair}$. Then,
we prove that  the sum of lengths of $\Delta_{\pair}$ over  all pairs at $\level{(m-1)}$ diverges
as $m$ approaches infinity. Thus, the length of the shortest path between $\mu_1$ and $\mu_2$ diverges too.

\begin{lemma1}
\label{lm:length}
  The length of the shortest path between $\mu_1 $ and $\mu_2$ in $\YY{2k+1}(\calP_{m})$ for
  $k \geq 3$ is at least $\rho^m$, for some  $\rho =  ( 1 - O(d^{-1}_0))  \cdot \cos^{-1}(\theta/2)$.
  Thus, by setting $d_0 > \lceil 6(1-\cos (\theta/2))^{-1} \rceil $, the length  diverges as $m$ approaches infinity.
\end{lemma1}
\begin{proof}

  We give a lower bound of the sum of lengths $|w_1w_2|$ over all pairs $(w_1, w_2)$ at
  $\level{(m-1)}$.  Recall that the length of   a pair is the length of the segment between the two
  points of the pair. Consider any pair $\pair = (v_1, v_2)$ with length $\delta$.   According to
 Property~\ref{prop:piece2}, the sum of lengths of half-empty, near-empty and empty pieces  is no more  than
  $6d^{-1}_0\delta $.  Thus,  the pieces which  generate internal-pairs  in next level
  have length  at least $(1 - 6d^{-1}_0)\delta$.   For each piece, it generates two child-pairs. The sum of lengths of the two
  pairs is  $\cos^{-1}(\theta/2)$ times larger than the piece itself.  Overall, the sum of the
  lengths of the pairs  in  generated next levels  is at least $(1 -  6d^{-1}_0)\delta \cos^{-1}(\theta/2)$.  Let $\rho =
   (1- 6d^{-1}_0) \cos^{-1}(\theta/2)$.  Thus, after $(m-1)$ rounds, the length of the pairs at
   $\level{(m-1)}$ is at least $\rho^{m-1} |\mu_1\mu_2|$.
  According to Lemma~\ref{lm:distaux}, the shortest path from $\mu_1$ to $\mu_2$  is at
  least $(1- 6 d^{-1}_0) \rho^{m-1} |\mu_1 \mu_2|$. When $d_0 > 6(1 - \cos(\theta/2))^{-1}$, the shortest
  path between $\mu_1$ and $\mu_2$ in $\YY{2k+1}(\calP_{m})$ diverges as $m$ approaches
  infinity.
\end{proof}

\eat{
Finally, we discuss the relation between the number of points $n$ and the level $m$.  Let $N[x, y]$
to be the number of points of a tree with $x$ levels and $y$ child pairs of the root. Note that the
  length of child pairs may not be same according to Definition~\ref{defn:gadget}. Thus, here,
  $N[x,y]$ represents the maximum possible number of points for any possible $y$ child pairs.  Suppose
the root pair is $(\mu_1,\mu_2)$ and its child pairs are $\calS$. The number of pieces of the partition for
the first internal sibling pair is $2d_0$.  Consider the the first $(y-1)$ pairs. There are $N[x,
y-1]$ points. The matching process matching these points to the last pair. Thus, there are at most
$N[x, y-1]$ matching point on the segment of last pair. The filling process $ |\ffill[\calP]| $ adds
at most $ |\calP|^2$ points.  Thus, $N[x,y] = O(N[x-1, N^2[x, y-1]])$. Overall, we have the
recursion about $N[x,y]$ as follows.

\begin{eqnarray}
  && N[x, y] = \left\{
     \begin{array}{ll}
       2y+1  & \text{ for }  x = 1 \\
       O( N[x-1, 2d_0]) & \text{ for } y =1 \\
       O(N[x-1,  N^2[x, y-1]]) & \text{ for }  x > 1, y > 1
     \end{array} \right.
\end{eqnarray}

The number $n $ is $ N[m, d_0]$. It  can be computed according to the recursion. The number is
computable but not primitive recursive which is similar to Ackermann number.
}

Finally, combining with the results that $\YY{3}$~\cite{el2009yao}  and $\YY{5}$~\cite{barba2014new}  may
not be spanners, we have proved Theorem~\ref{thm:odd}.
\begin{reptheorem}{thm:odd}[restated]
    For any $k \geq 1$, there exists a class of instances
    $\{\calP_m\}_{m\in \mathbb{Z}^+}$
    such that the stretch factor
    of $\YY{2k+1}(\calP_m)$ cannot be bounded by any constant,
    as $m$ approaches  infinity.
\end{reptheorem}


\newpage
\bibliographystyle{plainurl}
\bibliography{citation}

\newpage

\end{document}